\documentclass[12 pts]{article}
\usepackage[a4paper, total={6in, 10in}]{geometry}
\usepackage[utf8]{inputenc}
\usepackage{amsmath}
\usepackage{amssymb}
\usepackage{amsthm}
\usepackage{bbm}

\usepackage{mathrsfs}
\usepackage{bbm}
\usepackage[utf8]{inputenc}
\usepackage[T1]{fontenc}
\usepackage{color}
\usepackage{graphicx}
\usepackage{natbib}
\usepackage{enumitem}

\usepackage{hyperref}
\usepackage{theoremref} 
\usepackage{multirow}

\DeclareMathOperator{\E}{\mathbb{E}}

\numberwithin{equation}{section}
\theoremstyle{plain}
\newtheorem{theorem}{Theorem}[section]
\newtheorem{lemma}{Lemma}[section]
\newtheorem{corollary}{Corollary}[section]

\newtheorem{assumption}{Assumption}
\newtheorem{remark}{Remark}[section]
\usepackage{amsthm}

\definecolor{Red}{rgb}{0,0,0}

\definecolor{DR}{rgb}{0,0,0}

\definecolor{Blue}{rgb}{0,0,0}
\newcommand{\Blue}{\color{Blue}}
\definecolor{Green}{rgb}{0,0,0}
\newcommand{\Green}{\color{Green}}

\definecolor{Grey}{rgb}{0,0,0}

\title{Kernel Estimation of Spot Volatility with Microstructure Noise Using Pre-Averaging}
\author{Jos\'e E. Figueroa-L\'opez\thanks{{Department of Mathematics and Statistics, Washington University in St. Louis, St. Louis, MO 63130, USA ({\tt figueroa-lopez@wustl.edu}). Research supported in part by the NSF Grants: DMS-2015323, DMS-1613016.}} \and Bei Wu\thanks{{Department of Mathematics and Statistics, Washington University in St. Louis, St. Louis, MO 63130, USA ({\tt bei.wu@wustl.edu}).}}}
\date{Feb. 6, 2022}

\begin{document}
\maketitle
\begin{abstract}
We first revisit the problem of estimating the spot volatility of an It\^o semimartingale using a kernel estimator. We prove a Central Limit Theorem with an optimal convergence rate for a general two-sided kernel. Next, we introduce a new pre-averaging/kernel estimator for spot volatility to handle the microstructure noise of ultra high-frequency observations. We prove a Central Limit Theorem for the estimation error with an optimal rate and study the optimal selection of the bandwidth and kernel functions. We show that the pre-averaging/kernel estimator's asymptotic variance is minimal for two-sided exponential kernels, hence, justifying the need of working with kernels of unbounded support as opposed to the most commonly used uniform kernel. We also develop a feasible implementation of the proposed estimators with optimal bandwidth. Monte Carlo experiments confirm the superior performance of the devised method.


\end{abstract} \hspace{10pt}

\section{Introduction}
{It\^o semimartingale} models for the dynamics of asset returns have been widely used in financial econometrics. Such a process takes the form 
\begin{equation}\label{ItoModel00}
dX_t = \mu_t dt + \sigma_t d W_t+{dJ_{t}},
\end{equation}
where {\(\{W_t\}_{t\geq{}0}\)} is a standard Brownian motion and {$\{J_{t}\}_{t\geq{}0}$ is the jump component}. {The spot volatility \(\sigma_t\) is a key feature of the model as it} plays a crucial rule in option pricing, portfolio management, and financial risk management. Since last decade, there has been some growing interest in the estimation of volatility due to the wide availability of high frequency data. In this work, we are concerned with spot volatility estimation in an It\^o semimartingale model {\Blue via kernel smoothing}. This is one of the most widely used nonparametric methods in statistics, dating back to the seminal work of  \cite{rosenblatt1956remarks} and \cite{parzen} (see also the monograph \cite{wand1995monographs}).

One of the earliest works on kernel-based estimation of spot volatility dates back to \cite{foster1994continuous_optKernel}, where they studied a weighted rolling window estimator, which is essentially a kernel estimator with compact support.  Asymptotic normality was established under abstract conditions that {\Blue were not directly stated in terms of the coefficients of the It\^o semimartingale (\ref{ItoModel00})}. Concretely, they worked with a time series discretization of the model (\ref{ItoModel00}). 
 \cite{fan2008spot} established the asymptotic normality for a general kernel estimator, {\Blue this time working directly with the model (\ref{ItoModel00})  {\Blue under relatively mild conditions on the coefficients,} but without jumps. {\Blue However, the} result therein also required a certain condition on the convergence rate of the bandwidth to $0$}, which allowed them to neglect the ``target error" coming from approximating the spot volatility by a kernel {\Blue weighted volatility}. 
 As a result, {\Blue the convergence rates of the estimators were suboptimal (see Section 6 in \cite{FigLi} for more details)}. 
 {\Blue \cite{kristensen2010nonparametric} also proved a Central Limit Theorem (CLT) for kernel-based estimators  under the absence of jumps and a non-leverage condition (i.e., $\sigma$ and $W$ were assumed to be independent). \cite{Yuetal2014} generalized Kristensen's result by allowing a jump component of finite activity (FA), but still assuming non-leverage effects. \cite{mancini2015estimation} considered more general It\^o semimartingales, but again FA jumps.  All these works {\Blue only considered CLTs with \emph{suboptimal convergence rates}.}}
\cite{alvarez2012estimation} proposed an estimator of \(\sigma^p_t \) by considering forward finite difference approximations of the realized power variation process of order \(p\), which is essentially a forward-looking kernel estimator with uniform kernel. \cite{JacodProtter} {\Blue (Section 13.3 therein)} considered both backward and forward finite difference approximations of the realized quadratic variation.  Both works obtained the best possible convergence rates for their CLTs {\Blue for a rather} general It\^o semimartingale model ({\Blue in the case \cite{JacodProtter}, also including jumps}). We also refer to \cite{jacodaitsahalia},  {\Blue Chapter 8}, 
for an extensive review of the relevant literature.

More recently, \cite{FigLi} studied the leading order terms of the mean-square error (MSE) of kernel-based estimators {for continuous It\^o semimartingales} under a certain local condition on the covariance function of the spot variance $\sigma_{t}^{2}$, which covers not only Brownian driven volatilities but also those driven by fractional Brownian motion and other Gaussian processes. Using the asymptotics for the MSE, the optimal convergence rate was established and formulas for the optimal bandwidth and kernel functions were derived {\Blue under a non-leverage condition}. CLTs for general {\Blue right-sided} kernel estimators were also obtained (see also Remark 8.10 in \cite{jacodaitsahalia}, where a result for a general right-sided kernel with compact support was stated without proof).
One of the objectives of the present work is then to extend the results of \cite{FigLi} {\Blue and \cite{Yuetal2014}\footnote{{\Blue As explained above, \cite{Yuetal2014} established a CLT for a general two-sided kernel but with {\Blue a} suboptimal converge rate, FA jumps, and non-leverage.}}}, and prove a CLT for a general \emph{two-sided} kernel {\Blue of unbounded support,} {\Blue with optimal convergence rate and in the presence of jumps and leverage {\Blue effects}}. As {\Blue proved} in this paper in greater generality, such kernels can have better performance than either one-sided  or {compactly supported} kernels. Until now, this fact seems to have eluded the literature, which has almost exclusively focussed on uniform kernels.

 While the results described in the previous paragraph are {important} for intermediate intraday frequencies (e.g., 1 to 5 minute), it is widely accepted that financial returns at ultra high-frequency are contaminated by market microstructure noise. 
Specifically, high-frequency asset prices exhibit several stylized features, which cannot be accounted by It\^o semimartingales, such as clustering noises, bid/ask bounce effects, and roundoff errors  {\Blue (cf. \cite{Campbell}, Chapter 3, \cite{Zeng:2003}, \cite{jacodaitsahalia}, Chapter 2)}. Such discrepancies between macro and micro movements are typically modeled by an additive noise. The literature of statistical estimation methods under microstructure noise has grown extensively since last decade and is still a highly researched subject (see \cite{zhang2005tale}, \cite{HanLun}, \cite{Bandi}, \cite{MyklandZhang2012}, \cite{barndorff2008designing}, \cite{podolskij2009estimation}, and \cite{jacod2009microstructure}  for a few seminal works in the area as well as the monograph \cite{jacodaitsahalia}).
Most of the existing literature on volatility estimation for high frequency data with microstructure noise has mainly focused on the estimation of {\Blue the} integrated volatility or variance (IV), defined as $IV_{T}=\int_{0}^{T}\sigma_{t}^2dt$.  \cite{zhang2005tale} showed that scaled by \((2n)^{-1}\),  the realized variance estimator, the gold standard for IV estimation in the absence of microstructure noise, consistently estimates the variance of the microstructure noise, instead of the integrated volatility, as the sampling frequency \(n\) increases. There are several approaches to overcome this problem: the Two Scale Realized Variance (TSRV) estimator by \cite{zhang2005tale} and the efficient Multiscale Realized Variance by \cite{Zhang2006}; the Realized Kernel estimator by \cite{barndorff2008designing};  the pre-averaging method by \cite{podolskij2009estimation} and \cite{jacod2009microstructure}; and the Quasi-Maximun Likelihood Estimator (QMLE) by \cite{xiu2010quasi}.

Spot volatility estimation is often viewed as a byproduct of integrated volatility estimation {\Blue since, in principle, we can recover the spot volatility $\sigma_t^2$ as a finite-difference approximation of an estimate of the integrated volatility}. Following this idea,  \cite{zu2014estimating} {constructed} {\Blue a} Two-Scale Realized Spot Variance (TSRSV)  estimator  based on the TSRV integrated variance estimator {\Blue of \cite{zhang2005tale}}. They proved consistency and derived the asymptotic distribution of the estimation error with a convergence rate of  \(n^{-1/12}\), which is suboptimal. 

The second objective of our work is to construct a kernel based estimator of the spot volatility based on the pre-averaging integrated variance estimator of \cite{jacod2009microstructure}. The basic idea is simple and natural. If we denote $\widehat{IV}^{pre-av}_{t}$ the pre-averaging estimator of $IV_{t}=\int_{0}^{t}\sigma_{s}^{2}ds$, our estimators combines this with a kernel localization technique as follows:
\[
	\hat{\sigma}_{t}^{2}=\int_{0}^{t}\frac{1}{b_n}K\left(\frac{s-t}{b_n}\right)d IV^{pre-av}_{s},
\]
where $K$ is a suitable kernel function and $b_n>0$ is the bandwidth, which should converge to $0$ at an appropriate rate.   We establish the asymptotic mix normality of our estimator and identify two asymptotic regimes for two different bandwidth convergence regimes. One of those regimes yields the optimal convergence rate of \(n^{-1/8}\) for our estimator. It is important to point out that the asymptotic theory for the kernel/pre-averaging estimator cannot be derived from that for the pre-averaging integrated variance and also is substantially different and harder than that for kernel based estimators in the absence of microstructure noise. 

{\Blue Though combining pre-averaging and kernel {\Blue smoothing} is a natural idea, to the best of knowledge, there are only two related results in the literature.} \cite{jacodaitsahalia} {\Blue (Section 8.7 therein)}, stated, without proof, a stable convergence result of {\Blue a pre-averaging estimator for the spot volatility of a continuous It\^o semimartingale}\footnote{{\Blue The estimator therein is different from ours. Our estimator includes a debiasing term, which is {\Blue omitted} in  \cite{jacodaitsahalia}. Our Monte Carlo experiments show that such a correction is important in finite samples.}}, but only in the case of a one-sided uniform kernel  $K(t)={\bf 1}_{[0,1]}(t)$ {\Blue (see also \cite{chen2018inference} for a similar estimator)}. Here we consider {\Blue a truncated version to handle the jumps and a general two-side kernel} (see below as to the need of considering such kernels). {\Blue \cite{Yu2} also proposed {\Blue a pre-averaging kernel estimator for the spot volatility, slightly} different from our estimator. They established asymptotic normality with suboptimal convergence rate for their untruncated estimator in the case of a continuous It\^o semimartingale, and for their truncated estimator in the presence of L\'evy jumps of bounded variation. {\Blue In both situations, a non-leverage condition was adopted. In our case, we consider not only leverage effects, but also more general jump processes, not necessarily of L\'evy type and with no restriction in the index of jump activity, under both the suboptimal and optimal convergence rate regimes.}}

As an important application of our results, we study the problem of bandwidth and kernel function selection.  Using our CLT, we first derive the optimal bandwidth and then the optimal kernel function (the one that minimizes the limiting variance) at the optimal rate. 
We {then} show that the optimal kernel  is a two-sided exponential or Laplace function $K(x)=\frac{1}{2}e^{-|x|}$. This fact justifies the necessity of developing the asymptotic theory for general kernels {\Blue of unbounded support} over the more widely used uniform kernels. Again, we emphasize that our work is critical because it calls into question the indiscriminate use of uniform one-sided kernels in the literature. If we were constrained to compactly supported kernels {\Blue in the suboptimal asymptotic regime}, a uniform kernel would be the best, but this is no longer the case if we allow kernels with unbounded support {\Blue and/or consider an optimal convergence rate regime}. Similarly, two-sided kernels will perform better than one-sided, even if compactly supported.

The implementation of the optimum bandwidth (at the optimum rate) is more challenging because it involves the vol vol and the spot volatility itself. Hence, to implement it we develop a new method, which iteratively estimates the spot volatility, the vol vol, and the optimal bandwidth. Using Monte Carlo simulation,  we compare our estimator with the TSRSV estimator of \cite{zu2014estimating} and show a significant improved accuracy. We also illustrate the improvement achieved by the optimal exponential kernel and the calibrated optimal bandwidth via our iterative method. 

We finish the introduction by giving one more reason as to the importance of estimating the spot volatility.
 As mentioned above, while spot volatility estimation can, at least conceptually, be seen as a byproduct of integrated variance estimation, interestingly enough, one can also use spot volatility estimation as an intermediate step toward the estimation of integrated volatility functionals of the form $I_{T}(g):=\int_{0}^{T}g(\sigma_{s}^{2})ds$. Specifically, once an estimator $\hat{\sigma}_{t}^{2}$ of $\sigma^{2}_{t}$ has been developed, one can naturally devise an estimator for $I_{T}(g)$ of the form $\hat{I}_{T}(g)=\Delta_{n}\sum_{i=1}^{n}g(\hat{\sigma}^{2}_{t_{i}})$, where $t_{i}=i\Delta_{n}$ and $\Delta_{n}=T/n$, followed by an appropriate bias correction adjustment. In the absence of noise, \cite{jacod2013quarticity}, \cite{li2019efficient}, and \cite{mykland2009inference} have developed methods for the estimation of these functionals (see also \cite{li2016generalized}, \cite{ait2019principal}, and \cite{li2017adaptive} for related methods and other applications thereof). Recently, \cite{chen2018inference} developed an estimator for $\hat{I}_{T}(g)$ based on a forward finite difference approximation of the standard pre-averaging estimator of the integrated variance.

The rest of the paper is organized as follows. Section \ref{setting} introduces the setting of the problem and the main result. Section \ref{optimal_para} shows an application of our main theorem: the optimal parameter and kernel selection. The simulations are provided in Section \ref{simulation}.  {\Blue Some conclusions are given in Section \ref{ConcludeSec}}. Proofs of our main results can be found in two appendices.

\section{The Setting, Estimator, and Main Results} \label{setting}
Throughout, {we consider an It\^o semimartingale of the form:}
\begin{equation} \label{eq:X}
\begin{split} 
X_{t}=&X_0 + \int_0^t\mu_{s} ds +\int_0^t\sigma_{s} d W_{s}  \\
&+\int_{0}^{t} \int_{E} \delta(s, z) \mathbbm{1}_{\{{|\delta(s, z)|} \leq 1\}}(\mathfrak{p}-\mathfrak{q})(d s, d z)
+\int_{0}^{t} \int_{E} \delta(s, z) \mathbbm{1}_{\{{|\delta(s, z)|}>1\}} \mathfrak{p}(d s, d z),
\end{split}
\end{equation} 
where all stochastic processes ($\mu :=\left\{\mu_{t}\right\}_{t \geq 0}$, $\sigma :=\left\{\sigma_{t}\right\}_{t \geq 0}$,  $W :=\left\{W_{t}\right\}_{t \geq 0}$, $\mathfrak{p}:=\{\mathfrak{p}(B):B\in\mathcal{B}(\mathbb{R}_{+}\times E)\}$) are defined on a complete filtered probability space  $\left(\Omega^{(0)}, \mathcal{F}^{(0)}, \mathbb{F}^{(0)}, \mathbb{P}^{(0)}\right)$  with filtration $ \mathbb{F}^{(0)}=\big(\mathcal{F}^{(0)}_{t}\big)_{t \geq 0}$  and {are assumed to satisfy standard conditions for $X$ to be well-defined}. Here, $W$ is a standard Brownian Motion (BM)  adapted to the filtration \(\mathbb{F}^{(0)}\), {and} $\mathfrak{p}$   is a Poisson random measure on $\mathbb{R}_{+} \times E$  for some arbitrary Polish space $E$ with compensator  $\mathfrak{q}(\mathrm{d} u, \mathrm{d} x)=\mathrm{d} u \otimes \lambda(\mathrm{d} x)$, {where} \(\lambda\) is a \(\sigma-\) finite measure on \(E\) having no atom. {For further details regarding It\^o semimartingales, see Section 2.1.4 in \cite{JacodProtter}.}

We denote the spot variance process \(c_t = \sigma_t^2\) and assume it is also an It\^o semimartingale with the following {dynamics:}
\begin{equation} \label{eq:sigma}
c_{t}=c_0 + {\int_0^t{\tilde{\mu}_{s}} \mathrm{d} s+\int_{0}^{t}\tilde{\sigma}_{s} \mathrm{d}B_{s}  + \int_{0}^{t}\int_{E} \tilde{\delta}(s, z)(\mathfrak{p}-\mathfrak{q})(\mathrm{d} s, \mathrm{d} z)},
\end{equation}
where {$B:=\{B_t\}_{t\geq0}$} is a standard Brownian Motion adapted to  $\mathbb{F}^{(0)}$ {so that} $d\left<W,B\right>_{t}=\rho_{t} d t$. Here, $\{\tilde{\mu}_t\}_{t\geq{}0}$ is adapted locally bounded; \(\{\rho_t\}_{t\geq{}0}\) is adapted, locally bounded, c\`adl\`ag;  $\{\tilde{\sigma}_t\}_{t\geq{}0}$ is adapted c\`adl\`ag  and $\tilde{\delta}$  is a predictable function on ${\mathbb{R}_{+}} \times E$  satisfying standard conditions for the process above to be well-defined {(see \cite{JacodProtter})}.

We now state the {main assumption on the process \(X\)}:
\begin{assumption} \thlabel{X}
The process \(X\) satisfies (\ref{eq:X}) with \(c_t= \sigma^2_{t}\) satisfying (\ref{eq:sigma}) and, for some \(r \in [0,2]\), measurable functions $\Gamma_{m}, \Lambda_{m}:E\to\mathbb{R}_{+}$, {constants $C_{m}<\infty$,} and a localizing sequence of stopping times {$\left(\tau_{m}\right)_{m\geq{}1}$ such that $\tau_{m}\to\infty$}, we have
\begin{equation*}
t \in\left[0, \tau_{m}\right] \Longrightarrow\left\{\begin{array}{l}
\left|\mu_{t}\right|+\left|\sigma_{t}\right|+\left|\tilde{\mu_{t}}\right|+\left|\tilde{\sigma_{t}}\right| \leq C_m, \\
|\delta(t, z)| \wedge 1 \leq \Gamma_{m}(z), \quad \text{ where }\int \Gamma_{m}(z)^{r} \lambda(d z)<\infty,\\
\left|\tilde{\delta}(t, z)\right| \wedge 1 \leq \Lambda_{m}(z), \quad \text{ where }\int \Lambda_{m}(z)^{2} \lambda(d z)<\infty.
\end{array}\right. 
\end{equation*}
\end{assumption}
The parameter $r$ plays a key role in our asymptotic results. In short, $r$ determines the jump activity of the process: the larger $r$ is, the more active or frequent are the small jumps of the process. When $r=1$, the process exhibit finite many jumps in any bounded time interval (in that case, we say that the jumps are of finite activity). When $r<1$, the jump component of the process is of bounded variation. 

To establish {the} central limit theorem for the kernel estimator $\hat{c}_{t}$, we need some assumptions on the kernel.
\begin{assumption} \thlabel{kernel} 
The kernel function $K : \mathbb{R} \rightarrow \mathbb{R}$ is bounded and 
\begin{enumerate}
\item $\int K(x) d x=1$;
\item K is Lipschitz and piecewise $C^1$  on $(-\infty,\infty)$;
\item (i) $\int|K(x)  x| d x<\infty$ ; (ii) $K(x) x^{2} \rightarrow 0$,  as $|x| \rightarrow \infty$ ; (iii) $\int |K^{\prime}(x)|dx < \infty$.
\end{enumerate}
\end{assumption}
For an arbitrary process $\{U_{t}\}_{t\geq{}0}$ and a given time span $\Delta_{n}>0$, we shall use the notation 
\[
	U_{i}^{n}:=U_{i\Delta_{n}},\qquad 
	\Delta_{i}^{n} U:=U_{i}^{n}-U_{i-1}^{n}.
\]
Stable convergence in law is denoted by \( \stackrel{st}{\longrightarrow}\). See (2.2.4) in \cite{JacodProtter} for the definition of this type of convergence. As usual, $a_{n}\sim b_{n}$  means that $a_{n}/b_{n}\to{}1$ as $n\to\infty$.

Throughout the paper, we consider two settings: observations with and without market microstructure noise. In the absence of microstruture noise, we use standard kernel estimation, while to {handle} the noise we propose a type of pre-averaging kernel estimator.  These two settings together with the main results are presented in the following two subsections.

\subsection{Observations without microstructure noise}
In this subsection, we assume that we can directly observe the process $X$ in (\ref{eq:X}) at discrete times   $t_{i} :=t_{i, n} :=i \Delta_{n}$,  where $\Delta_{n} :=T / n$ and $T\in(0,\infty)$ is a given fixed time horizon.  We also consider a sequence of truncation levels $v_{n}$ satisfying
\begin{equation} \label{eq:trucation_para}
v_{n}=\alpha \Delta_{n}^{\varpi} \quad \text { for some } \alpha>0, \quad \varpi \in\left(0, \frac{1}{2}\right).
\end{equation}
To estimate the spot volatility $c_{\tau}$, at a given time $\tau \in (0,T)$, we adopt the kernel estimator, studied in \cite{fan2008spot,kristensen2010nonparametric}  and {its truncated version}, studied in \cite{Yuetal2014} {\Blue and \cite{mancini2015estimation}}:
\begin{align}\label{MDKNN}
	{\hat{c}^{n}\left({m}_{n}\right)_{\tau}} &:=\sum_{i=1}^{n} K_{{m}_{n}\Delta_n}\left(t_{i-1}-\tau\right)\left(\Delta_{i}^{n} X\right)^{2},\\
	\label{MDKNN_truncated}
	{\hat{c}^{n}\left({m}_{n},v_n\right)_{\tau}} &:=\sum_{i=1}^{n} K_{{m}_{n}\Delta_n}\left(t_{i-1}-\tau\right)\left(\Delta_{i}^{n} X\right)^{2}\mathbbm{1}_{\left\{\left|\Delta_{i}^{n} X\right| \leq v_{n}\right\}},
\end{align}
where \(K_{b}(x):=K(x / b) / b\), $m_{n}\in\mathbb{N}$, and $b_n:={m}_{n} \Delta_n$ is the bandwidth of the kernel function\footnote{Here, ${m}_{n}$ is equivalent  to $k_n$ in the Theorem 13.3.7 of \cite{JacodProtter}, while ${m}_{n}\Delta_n$ is equivalent to the bandwidth $h_n$ of \cite{FigLi}.}. The asymptotic behavior of this estimator with one-sided uniform kernels (i.e., $K(x)={\bf 1}_{[0,1]}(x)$ or $K(x)={\bf 1}_{[-1,0]}(x)$) was studied in \cite{JacodProtter}. {\cite{Yuetal2014} showed a CLT for (\ref{MDKNN_truncated}) at the suboptimal convergence rate ($\beta=0$) under a nonleverage condition (i.e., $d\left<W,B\right>_{t}=0$ in (\ref{eq:X})-(\ref{eq:sigma}) and compound Poisson jump component}. 
In this part, we extend the results to general two-sided kernels with possibly unbounded support, {optimal rate, and more general type of It\^o semimartingale}. There is an important motivation for considering {general unbounded} kernels since, as proved in \cite{FigLi} {in the non-leverage case and without jumps}, exponential and some other nonuniform unbounded kernels can yield estimators with {significantly} better performance than those based on uniform kernels. {In Section \ref{optimal_para} below, we show that this is also true under the more general semimartingale model (\ref{eq:X})-(\ref{eq:sigma}).}

We now proceed to describe the limiting distribution of the estimation error of (\ref{MDKNN})-(\ref{MDKNN_truncated}).  Let \(V, V^{\prime}\) be independent centered Gaussian variables,  {independent of $\mathcal{F}^{(0)}$,} defined on a ``very good'' filtered extension {$\Big(\widetilde{\Omega}^{(0)}, \widetilde{\mathcal{F}}^{(0)},\big(\widetilde{\mathcal{F}}_{t}^{(0)}\big)_{t \geq 0}, \widetilde{\mathbb{P}}^{(0)}\Big)$ 
of $\Big(\Omega^{(0)}, \mathcal{F}^{(0)}, \big(\mathcal{F}_{t}^{(0)}\big)_{t \geq 0}, \mathbb{P}^{(0)}\Big)$} (see \cite{JacodProtter} for definition) such that
\begin{equation} \label{eq:Y_Yprime}
\E\left( V^2\right) =2\int K^2(u) du, \quad \E\left(V^{\prime 2} \right)= \int L^2(t)dt, 
\end{equation} where \(L(t) = \int_{t}^{\infty} K(u) d u \mathbf{1}_{\{t>0\}}-\int_{-\infty}^{t} K(u) d u \mathbf{1}_{\{t \leq 0\}}\). 
Next, let \(Z^{(0)}_{\tau}, Z^{\prime (0)}_{\tau}\) be defined as
\begin{equation} \label{eq:limiting_Z}
Z^{(0)}_{\tau} = c_{\tau}  V, \quad Z^{\prime (0)}_{{\tau}} =  {\tilde{\sigma}_{\tau} }V^{\prime}.
\end{equation}
Now we are ready to introduce our main theorem for a general kernel estimator in the absence of microstructure noise. The proof is given in Appendix \ref{PrfOfMnRsltTh1}.
\begin{theorem} \thlabel{thm_no_noise}
Let the sequence $\{{m}_{n}\}_{n\geq{}1}$ that controls the bandwidth of the kernel estimator be such that
${m}_{n} \rightarrow \infty$, ${m}_{n}\Delta_n\rightarrow 0$, and 
\begin{equation}\label{IDBWN}
	{\Blue {m}_{n} \sqrt{\Delta_{n}} \rightarrow \beta, \quad\text{with}\quad \beta \in [0,\infty]}.
\end{equation} Then, under \thref{X,kernel} above, at a given time \(\tau \in [0,T],\) 
we have: \begin{enumerate}[label=\alph*)]
\item If $X$ is continuous, both the truncated version (\ref{MDKNN_truncated}) and the non-truncated version (\ref{MDKNN}) satisfy the following stable convergence in law,  as $n\to\infty$:
\begin{equation} \label{eq:stable_convergence}
\begin{aligned}
\rm (i) &\quad  \sqrt{{m}_{n}}\left(\hat{c}^{{n}}_{\tau}-c_{\tau}\right)\stackrel{st}{\longrightarrow}Z^{(0)}_{\tau}+\beta Z_{\tau}^{\prime (0)},\quad \text{if}\quad \beta < \infty,\\
\rm (ii) &\quad \frac{1}{\sqrt{{m}_{n}\Delta_n}}\left(\hat{c}^{{n}}_{\tau}-c_{\tau}\right) \stackrel{st}{\longrightarrow}Z_{\tau}^{\prime (0)},\quad\text{if}\quad \beta =\infty,\end{aligned}
\end{equation} where  \(Z^{(0)}_{\tau}, Z_{\tau}^{\prime  (0)}\) are defined as in (\ref{eq:limiting_Z}).
\item {\Blue Suppose 
\begin{equation} \label{eq:m_n_range}
m_{n} \Delta_{n}^{a} \rightarrow \beta^{\prime} \in(0, \infty), \quad \text {where } a \in(0,1),  
\end{equation} 
so that \eqref{IDBWN} holds with $\beta=0$ when $a<1/2$, $\beta=\beta'\in(0,\infty)$ when $a=1/2$, {\Blue or} $\beta=\infty$ when $a>1/2$. Then, when $X$ is discontinuous,} we have (\ref{eq:stable_convergence}) for the non-truncated version (\ref{MDKNN}), as soon as 
\begin{equation*}
\text{either } r<\frac{4}{3}, \quad \text { or }\quad \frac{4}{3} \leq r<\frac{2}{1+a} \quad\left(\text{and then } a<\frac{1}{2}\right).
\end{equation*}
\item Under (\ref{eq:m_n_range}), when $X$ is discontinuous, we have (\ref{eq:stable_convergence}) for the truncated version, as soon as
\begin{equation} \label{eq:r_range_truncated}
r<\frac{2}{1+a \wedge(1-a)}, \quad \varpi>\frac{a\wedge(1-a)}{2(2-r)}.
\end{equation}
\end{enumerate}
\end{theorem}
\begin{remark}
{The CLTs  above generalize the results in  \cite{FigLi}, where only right-sided kernels were considered under the absence of jumps,  in \cite{JacodProtter} and \cite{alvarez2012estimation}, where only one-sided uniform kernels (i.e., $K(x)={\bf 1}_{[0,1]}(x)$ or $K(x)={\bf 1}_{[-1,0]}(x)$) were studied, and in \cite{jacodaitsahalia}, where a CLT for a general right-sided kernel with compact support was stated without proof.  The proof of Theorem \ref{thm_no_noise} is also different from that in} \cite{FigLi} and is based on the approach of \cite{JacodProtter}. {\Blue The case with $\beta=0$ produces a CLT with convergence rate $m_n^{-1/2}$, which vanishes slower than $\Delta_n^{1/4}$, the optimal rate. In that case, our result generalizes  \cite{fan2008spot}, \cite{kristensen2010nonparametric}, \cite{Yuetal2014}, and \cite{mancini2015estimation} by allowing jumps of both finite and infinite activity and {\Blue dependence} between the volatility and the Brownian motion driving the log-return process $X$ ({\Blue leverage effects}).}
\end{remark}
\begin{remark}
	{\Blue As stated by the points (b)-(c) above, in the presence of jumps, both estimators (\ref{MDKNN}) and \eqref{MDKNN_truncated} can attain the optimal convergence rate of $\Delta_n^{1/4}$, but only if the index of jump activity is less than $4/3$. In the presence of higher jump activity, the estimators can only achieved the suboptimal convergence rate of $m_n^{-1/2}\gg \Delta_n^{1/4}$ (case $a<1/2$ and $\beta=0$). It is worth noting the surprising fact that, even in the presence of jumps, the untruncated kernel estimator (\ref{MDKNN}) can still consistently estimate the spot volatility. In the case of finitely many jumps, we may explain this fact by noting that in a small local window, there could be at most a finite number of jumps, while, in the limit, there are increasingly more increments that do not contain jumps\footnote{{\Blue We thank a referee for pointing out this interesting insight.}}. Nevertheless, in practice and for better finite sample performance, one typically would prefer the truncated version of the estimator.}
\end{remark}

\begin{remark}
{As explained in the introduction, it is critical to expand the results to general two-sided kernels of unbounded support since these kernels exhibit superior performance. For instance, in the suboptimal rate case ($\beta=0$), the kernel $K$ with support $[0,1]$ that minimizes the asymptotic variance $2\int K^{2}(u)du$ is the uniform kernel $K_{unif}(x)={\bf 1}_{[0,1]}(x)$ since, by Jensen's inequality, $\int_0^1 K^{2}(u)du\geq{}(\int_{0}^{1}K(x)dx)^{2}=1=\int_0^1 K_{unif}^{2}(u)du$. However, there are many other kernels that are two-sided or of unbounded support and that attain smaller variance, even in the suboptimal rate case $\beta=0$. For instance, both $K(x)=2^{-1}{\bf 1}_{[-1,1]}(x)$ and $K_{exp^+}(x)=e^{-x}{\bf 1}_{(0,\infty)}(x)$ are such that $\int K^{2}(u)du=1/2$. In the optimal rate case ($\beta\in(0,\infty)$), the optimal kernel is the two-sided exponential $K(x)=2^{-1}e^{-|x|}$ as shown in Subsection \ref{OptNoMicroSection} below.} 
\end{remark}

\subsection{Observations in the presence of microstructure noise}
In this part, we assume  that our observations of $X$ are contaminated by ``microstructure'' noise. That is, we assume we observe  
\begin{equation} \label{SmplSchm0}
	Y_{{t_i}}:=X_{t_i} +  \epsilon_{t_i},
\end{equation}
where $\epsilon = \{\epsilon_t\}$ is the noise process and, as before, $t_{i} :=t_{i, n} :=i \Delta_{n}$, $0 \leq i \leq n$,  with $\Delta_{n} :=T / n$ and a fixed time horizon $T\in(0,\infty)$. 
We allow the noise $\epsilon $ to depend on $X$, but in such a way that, conditionally on the whole process $X$, $\{\epsilon_t\}_{t\geq 0} $ is a family of independent, centered random variables. More formally, following the framework of \cite{JacodProtter}, for each time $t$, we consider a transition probability $Q_{t}\left(\omega^{(0)}, d z\right)$ from $\left(\Omega^{(0)}, \mathcal{F}^{(0)}_{t}\right)$ into $(\mathbb{R},\mathcal{B}(\mathbb{R}))$, and the canonical process $\{\epsilon_{t}\}_{t\geq{}0}$ on $\mathbb{R}^{[0, \infty)}$ defined as $ \epsilon_{t}(\tilde{\omega})=\tilde{\omega}(t)$ for $t\geq{}0$ and $\tilde{\omega}\in\mathbb{R}^{[0, \infty)}$. Next, we construct a new probability space	$\left(\mathbb{R}^{[0, \infty)}, \mathbb{B}, \sigma\left(\epsilon_{s} : s \in[0, t)\right), \mathbb{Q}\right) $, where $\mathbb{B}$ is the product Borel $\sigma$-field and $\mathbb{Q}=\otimes_{t \geq 0} Q_{t}$. We then define an enlarged  filtered probability space $\left(\Omega, \mathcal{F},\left(\mathcal{F}_{t}\right)_{t \geq 0}, \mathbb{P}\right)$  and a filtration $(\mathcal{H}_t) $ as follows:
\begin{equation*}
\left\{\begin{array}{l}\Omega=\Omega^{(0)} \times {\mathbb{R}^{[0, \infty)}}, \\ {\mathcal{F}_{t}=\mathcal{F}^{(0)}_{t} \otimes \sigma\left(\epsilon_{s} : s \in[0,t)\right), \quad \mathcal{H}_{t}=\mathcal{F}^{(0)}\otimes \sigma\left(\epsilon_{s} : s \in[0,t)\right)} \\ {\mathbb{P}(d \omega^{(0)}, {d \tilde{\omega}}  )=\mathbb{P}^{(0)}(d \omega^{(0)}) \mathbb{Q}(\omega^{(0)}, d {\tilde{\omega}}).}\end{array}\right.
\end{equation*}
Any variable or process in either $\Omega^{(0)} $ or $\mathbb{R}^{[0,\infty)} $ can be {extended} in the usual way to a variable or a process on $\Omega $. 
We now state the assumptions on the $\mathcal{F}^{(0)} $-conditional law of the noise process as well as some slightly different assumptions on the spot variance process and kernel function.
\begin{assumption} \thlabel{noise}
All variables $\left(\epsilon_t: t \geq  0 \right)$ are independent conditionally on $\mathcal{F}^{(0)}$, and we have
\begin{itemize} 
\item $\mathbb{E}\left(\left.\epsilon_{t} \right| \mathcal{F}^{(0)}\right)=0$,
\item For all $p>0$, the process $ \mathbb{E}\left(\left. \left|\epsilon_{t}\right|^{p} \right| \mathcal{F}^{(0)}\right)$ is $\left(\mathcal{F}^{(0)}_{t}\right)$-adapted and locally bounded,
\item The conditional variance process $\gamma_{t}= \mathbb{E}\left(\left. \left|\epsilon_{t}\right|^{2} \right| \mathcal{F}^{(0)}\right)$ is c\`adl\`ag.
\end{itemize}
\end{assumption}


Along the lines of \cite{JacodProtter} (originally proposed in \cite{jacod2009microstructure}), to construct the pre-averaging estimator, we need:
\begin{itemize}
\item[(i)] A sequence of positive integers $k_n$, which represent the length of the pre-averaging window,  satisfying
    \begin{equation}\label{AsympCndkn}
        k_{n}=\frac{1}{\theta \sqrt{\Delta_{n}}}+\mathrm{o}\left(\frac{1}{\Delta_{n}^{1/4}}\right), \quad \mbox{ {for some} }\; \theta>0;
    \end{equation}
\item[(ii)] A real-valued weight function $g$ on [0, 1], satisfying that g is continuous, piecewise $C^1$ with a piecewise Lipschitz derivative $g'$ such that\footnote{{It is enough to ask $g\in L^{2}([0,1])$, but, since the pre-averaging estimator is invariant to scalings of the weight function $g$, without loss of generality, we can impose the condition $|g|_{L^{2}}=1$.}}
\begin{equation*}
    g(0)=g(1)=0, \quad   \int_{0}^{1} g(s)^{2} d s=1{.}
\end{equation*}
\item[(iii)] A sequence $v_n$ representing the truncation level, satisfying 
\begin{equation} \label{eq:v_n_noise}
    v_{n}=\alpha \left(k_{n}\Delta_n \right)^{\varpi} \quad \text{       for some  } \alpha > 0, \;\varpi \in (0,\frac{1}{2});
\end{equation}
\end{itemize}
Next, for an arbitrary process $U$, we define the sequences:
\begin{equation} \label{eq:notation_bar_hat}
      \begin{array}{l}
{\overline{U}_{i}^{n}=\sum_{j=1}^{k_{n}-1} g\left(\frac{j}{k_{n}}\right) \Delta_{i+j-1}^{n} U} = -\sum_{j=1}^{k_{n}} \left(g\left(\frac{j}{k_{n}}\right) - g\left(\frac{j-1}{k_{n}}\right)\right) U_{i+j-2}^{n},\\
{\widehat{U}_{i}^{n}=\sum_{j=1}^{k_{n}}\left(g\left(\frac{j}{k_{n}}\right)-g\left(\frac{j-1}{k_{n}}\right)\right)^{2} \left(\Delta_{i+j-1}^{n} U\right)^2}.
\end{array} 
\end{equation}
As seen from the definition, $\overline{U}_{i}^{n} $ is the weighted average of the increments $\Delta_{i+j-1}U,    j= 1, \cdots, k_n -1 $, while $ \widehat{U}_{i}^{n}$ is a de-biasing term.
For a weight function $g$ as above, let
\begin{equation} \label{eq:phi}
    \phi_{k_{n}}(g)=\sum_{i=1}^{k_{n}} g(\frac{i}{k_{n}})^{2};\quad
    \phi_{k_{n}}^{\prime}(h)=\sum_{i=1}^{k_{n}}\left(g(\frac{i}{k_{n}})-g(\frac{i-1}{k_{n}})\right)^{2},
\end{equation}
and note that
\begin{equation} \label{eq:phi_g}
    \begin{array}{l}
\phi_{k_{n}}(g) =k_{n} \int_0^1 g^2(s) ds+\mathrm{O}(1) = k_n +\mathrm{O}(1) ,\\ 
\phi_{k_{n}}^{\prime}(g) =\frac{1}{k_{n}} \int_0^1 {(g^{\prime}(s))^{2}} ds+\mathrm{O}\left(\frac{1}{k_{n}^{2}}\right). \end{array}
\end{equation}
Now we can define the pre-averaging estimator of the spot variance $c_{\tau} $ at $\tau\in(0,T)$. We consider a non-truncated version, defined as
\begin{equation} \label{eq:non_truncated_pre}
\hat{c}\left(k_{n}, {m}_{n}\right)_{\tau} =\frac{1}{\phi_{k_{n}}\left(g\right)} \sum_{j=1}^{n-k_n+1} K_{{m}_{n} \Delta_n}\left(t_{j-1} - \tau\right)\left(\left(\overline{Y}_{j}^{n}\right)^2   - \frac{1}{2}\widehat{Y}_{j}^{n}\right),
\end{equation}
as well as, two truncated versions:
\begin{equation}\label{PreAverEst0} 
\begin{split}
\hat{c}\left(k_{n}, {m}_{n},v_n,1\right)_{\tau} &=\frac{1}{\phi_{k_{n}}\left(g\right)} \sum_{j=1}^{n-k_n+1} K_{{m}_{n} \Delta_n}\left(t_{j-1} - \tau\right)\left(\left(\overline{Y}_{j}^{n}\right)^2 \mathbbm{1}_{\{|\bar{Y}^n_j | \leq v_n\}}  - \frac{1}{2}\widehat{Y}_{j}^{n}\right),\\
\hat{c}\left(k_{n}, {m}_{n},v_n,2\right)_{\tau} &=\frac{1}{\phi_{k_{n}}\left(g\right)} \sum_{j=1}^{n-k_n+1} K_{{m}_{n} \Delta_n}\left(t_{j-1} - \tau\right)\left(\left(\overline{Y}_{j}^{n}\right)^2  - \frac{1}{2}\widehat{Y}_{j}^{n}\right)\mathbbm{1}_{\{|\bar{Y}^n_j | \leq v_n\}} .
\end{split}
\end{equation} 
{\Blue The basic idea is the same as in the case where the efficient process $X$ is observed without noise. We see $\overline{Y}_{j}^{n}$ as a noise-free proxy of the increment $\Delta_{j}^n X$. By properly choosing the truncation  {\Green level} $v_n\to{}0$ (e.g., $v_n\gg \sqrt{u_n ln(1/u_n)}$ with $u_n:=k_n \Delta_n$), the event $|\bar{Y}^n_j | >v_n$ will indicate the occurrence of a ``big" jump happening during the time interval $[j\Delta_n,(j+k_n)\Delta_n]$ and, thus, we eliminate such a term from the summations in \eqref{PreAverEst0}. The estimator $\hat{c}\left(k_{n}, {m}_{n},v_n,2\right)_{\tau}$ is closer to the one defined in \cite{Yu2}, while $\hat{c}\left(k_{n}, {m}_{n},v_n,1\right)_{\tau}$ is similar to the one considered in \cite{chen2018inference}, though therein only the one-sided kernel $K(x)={\bf 1}_{[0,1]}(x)$ is studied. It will be interesting to compare their statistical properties and finite-sample performance.} 

Before giving the asymptotic behavior of the pre-averaging estimators  {\Blue (\ref{eq:non_truncated_pre}) and (\ref{PreAverEst0})}, we introduced the limiting distributions. Below, $Z_{\tau}, Z^{\prime}_{\tau} $ are defined on a good extension  $\Big(\widetilde{\Omega}, \widetilde{\mathcal{F}},\big(\widetilde{\mathcal{F}}_{t}\big)_{t>0}, \widetilde{\mathbb{P}}\Big)$   of the space $\left(\Omega, \mathcal{F},\left(\mathcal{F}_{t}\right)_{t \geq 0}, \mathbb{P}\right)$ {so that,} conditionally on $ \mathcal{F}$, {they} are independent Gaussian random variables with conditional variance  \begin{equation} \label{eq:delta12}
\begin{aligned}
&\delta_1^2(\tau) := \widetilde{\mathbb{E}}\left(Z_{\tau}^2 | \mathcal{F}\right) =  4\left(\Phi_{22} c_{\tau}^{2}/\theta+2 \Phi_{12} c_{\tau} \gamma_{\tau}\theta+\Phi_{11} \gamma_{\tau}^{2}\theta^{3}\right) \int K^2(u) d u,\\
&\delta_2^2(\tau) := \widetilde{\mathbb{E}}\left(Z_{\tau}^{\prime 2} | \mathcal{F}\right)= \tilde{\sigma}_{\tau}^2\int L^2(t)dt,
\end{aligned}
\end{equation} with $\phi_{1}(s)=\int_{s}^{1} g^{\prime}(u) g^{\prime}(u-s) \mathrm{d} u$, $\phi_{2}(s)=\int_{s}^{1} g(u) g(u-s) \mathrm{d} u$, $ \Phi_{i j}=\int_{0}^{1} \phi_{i}(s) \phi_{j}(s) \mathrm{d} s$, and \(L(t) = \int_{t}^{\infty} K(u) d u \mathbf{1}_{\{t>0\}}-\int_{-\infty}^{t} K(u) d u \mathbf{1}_{\{t \leq 0\}}\). 
The following result establishes the asymptotic behavior of the estimation error for the proposed estimators. The proof is given in Appendix \ref{PrfOfMnRslt}.
\begin{theorem} \label{thm}
Let $\{{m}_{n}\}_{n\geq{}1}$ be a sequence of positive integers such that $m_{n}\to\infty$, ${m}_{n} \Delta_n \rightarrow 0$, ${m}_{n}\sqrt{\Delta_n} \rightarrow \infty $, and ${m}_{n}\Delta_n^{3/4}\rightarrow \beta$ for some $\beta \in [0,\infty]$, and let $k_{n}$, $v_n$, and $g$ be as described in (i)-(iii) above. {Then,} under  \thref{X,kernel,noise}, we have:
\begin{enumerate}
\item When $X$ is continuous, the pre-averaging estimators {\Blue (\ref{eq:non_truncated_pre}) and (\ref{PreAverEst0})} are {\Blue all} such that, as $n\to\infty$, 
\begin{equation}\label{CLTTCs}
\begin{aligned}
\rm (i) &\quad {m}_{n}^{1 / 2} \Delta_{n}^{1 / 4} \left(\hat{c}_{\tau} -c_{\tau}\right) \stackrel{st}{\longrightarrow} Z_{\tau} + \beta Z^{\prime}_{\tau}, \mbox{ if } {\Blue \beta\in[0,\infty)}, \\
\rm (ii) &\quad  \frac{1}{\sqrt{{m}_{n} \Delta_n}} \left(\hat{c}_{\tau} -c_{\tau}\right) \stackrel{st}{\longrightarrow} Z^{\prime}_{\tau}, \mbox{ if } \beta = \infty;
\end{aligned}
\end{equation}
\item When $X$ is discontinuous and $r \in (0,2]$, with \begin{equation} \label{eq:m_n_range_noise}
m_{n} \Delta_{n}^{a} \rightarrow \beta^{\prime} \in(0, \infty), \quad \text { where } a \in(\frac{1}{2},1),  
\end{equation} and $r, \varpi$ satisfying 
\begin{equation} \label{eq:r_w_condition_noise}
r < \frac{5}{2} - 2\left[(a - \frac{1}{4})\wedge (1-a+\frac{1}{4})\right],\quad 
\varpi \geq \frac{(a- \frac{1}{4})\wedge (1-(a- \frac{1}{4} ) ) - \frac{1}{4}}{2-r},
\end{equation} 
the truncated pre-averaging estimator {\Blue $\hat{c}(k_n, m_n, v_n, 1)_{\tau} $ in (\ref{PreAverEst0}) } satisfies (\ref{CLTTCs})  with $\beta=0$ when $a<3/4$, with $\beta=\beta'$ when $a=3/4$, {\Blue or} with $\beta=\infty$ when $3/4<a<1$.
\item {\Blue When $X$ is discontinuous and $r \in (0,2]$, with (\ref{eq:m_n_range_noise}) and $r, \varpi$ satisfying 
\begin{equation} \label{eq:w_condition_noise_2}
{\Blue r \le 4-\frac{2}{a \vee \left(3/2-a \right)},\quad \frac{(a- \frac{1}{4})\wedge (1-(a- \frac{1}{4} ) ) - \frac{1}{4}}{2-r}\leq \varpi \leq  \frac{(2-2a) \vee (2a-1)}{r},}
\end{equation}
the truncated pre-averaging estimator $\hat{c}(k_n, m_n, v_n, 2)_{\tau} $ in (\ref{PreAverEst0}) satisfies (\ref{CLTTCs})  with $\beta=0$ when $a<3/4$, with $\beta=\beta'\in(0,\infty)$ when $a=3/4$, and with $\beta=\infty$ when $3/4<a<1$.}
\end{enumerate}
\end{theorem}

\begin{remark}
	{\Blue The second and third points in the above theorem show that the truncated pre-averaging estimators (\ref{PreAverEst0}) can achieve the optimal convergence rate of $\Delta_n^{1/8}$, but only if the index of jump activity is restricted to be $r<3/2$ for $\hat{c}(k_n, m_n, v_n, 1)$ and $r<4/3$ for $\hat{c}(k_n, m_n, v_n, 2)$. If the index of jump activity is larger than $3/2$ and $4/3$ respectively , the estimators can only achieved suboptimal convergence rates. When comparing their theoretical properties, $\hat{c}(k_n, m_n, v_n, 1) $ can, in principle, handle jumps with higher index $r$ than  $\hat{c}(k_n, m_n, v_n, 2) $ at the optimal bandwidth. However, as we will see in Section \ref{simulation}, $\hat{c}(k_n, m_n, v_n, 2) $ appears to be more effective at eliminating jumps when the jump size is large in the presence of finite activity jumps. The two estimators have similar performance when the jump's size is relatively small.}
\end{remark}

\begin{remark}
{Let us give some intuition or heuristic explanation of the estimator (\ref{PreAverEst0}) and {\Blue its asymptotic behavior established above.}}
For the estimation of the integrated variance (IV), \([X,X]_T = \int_0^T c_t dt  \), \cite{jacod2009microstructure} proposed the {following pre-averaging estimator:}
		 \begin{equation*}
		\widehat{[X,X]_s}  := \frac{1}{\phi_{k_n}(g)} \frac{s}{s- k_n \Delta_n}\sum_{j=1}^{[s/\Delta_n]-k_n+1}\left(\left(\overline{Y}_{j}^{n}\right)^2 - \frac{1}{2}\widehat{Y}_{j}^{n}\right), \quad s \in (0, T],
		\end{equation*}
for {\Blue a continuous It\^o semimartingale} $X$. It was shown that: 
\begin{equation*}
		\frac{1}{\Delta_{n}^{1/4}}\left(		\widehat{[X,X]_T} -[X,X]_T\right) \stackrel{st}{\longrightarrow} \mathcal{U}_{T}^{\text {noise }},
		\end{equation*} 
where $\mathcal{U}_{T}^{\text {noise }}$ is a centered Gaussian process with conditional variance
\[\delta_T:=\E\left(\left(\mathcal{U}_{T}^{\text {noise }}\right)^{2} | \mathcal{F}\right)= 
{\int_{0}^{T} \zeta_t d t}:=\int_{0}^{T} 4\left(\Phi_{22} c_{t}^{2}/\theta+2 \Phi_{12}c_t\gamma_{t}\theta+\Phi_{11} \gamma_{t}^{2}\theta^{3}\right) d t. \] 
{In the no-thresholding case ($v_{n}=\infty$)}, the spot volatility estimator (\ref{PreAverEst0}) can be viewed as a localization of {the IV process} in that  \begin{equation*}
\hat{c}_{t}\approx \int K_{m_n \Delta_n}(s-t)d\widehat{\left[X,X\right]}_{s}.
\end{equation*}
More specifically, the factor $\frac{s}{s-k_n \Delta_n} $ is omitted for the spot volatility estimator. 
{If we use the representation $\mathcal{U}_{t}^{\text {noise }}=\int_{0}^{t}(\zeta_{s})^{1/2} dB^{U}_s$, where $B^U$ is a Wiener process, we can then heuristically argue that $\hat{c}_{t}-c_{t}\approx \int K_{m_n \Delta_n}(s-t)d(\widehat{\left[X,X\right]}_{s}-{\left[X,X\right]}_{s})=\Delta_{n}^{1/4}\int K_{m_n \Delta_n}(s-t)dU_{s}^{\text{noise}}=\Delta_{n}^{1/4}\int K_{m_n \Delta_n}(s-t)(\zeta_{s})^{1/2}dB^{U}_{s}$. Therefore, the variance of the estimation error at time \(t\) is expected to be close to} 
\begin{equation*}
\sqrt{\Delta_n}\int K^2_{m_n\Delta_n}(s-t) {\zeta_{s}ds} \approx \frac{1}{m_n\sqrt{\Delta_n}}4\left(\Phi_{22} c_{t}^{2}/\theta+2 \Phi_{12} c_{t}\gamma_{t}\theta+\Phi_{11} \gamma_{t}^{2}\theta^{3}\right)   \int K^2(u) du.
\end{equation*} 
which is indeed the case, {but only} when \({m}_{n}\Delta_n^{3/4}\rightarrow \beta = 0\) as formally shown in Theorem \ref{thm}. {It is important to remark that the proof of Theorem \ref{thm} does not rely on the heuristic arguments above.}
\end{remark}

\section{An application: Optimal Parameter Tuning}\label{optimal_para}
In this section, as an application of our main {\Blue Theorems \ref{thm_no_noise} and} \ref{thm}, we show how to tune {\Blue the bandwidth parameter $\beta$ and the pre-averaging parameter $\theta$, as well as the kernel function $K$ of the estimator,}  in order to minimize the asymptotic variance of the estimation error {\Blue $\hat{c}_{\tau}-c_{\tau}$}.  {Two possible approaches can be taken. Minimize the asymptotic variance of $\hat{c}_{\tau}$, say {\Blue $\tilde{\delta}^2(\tau)$}, at each time $\tau$ or minimize the \emph{integrated asymptotic variance} {\Blue $\int_{0}^{T}\tilde{\delta}^{2}(t)dt$} over the period $[0,T]$. In our simulations of Section \ref{simulation}, we implemented both methods and found out that the second method yields slightly better results. An explanation for this is given in Remark \ref{ExplLocvsGlob} below. Therefore, in this part, we focus on the second approach.}

By necessity, the optimal choices of $\theta$ and $\beta$ {under the criterion of the previous paragraph} will be expressed in terms of the integrated variance and quarticity,  $IV_{T}:=\int_0^T c_{t}dt$ and $QrT_{T}:=\int_0^Tc_{t}^{2}dt$, respectively, the Integrated Volatility of Volatility (IVV), $\int_0^T \tilde{\sigma}_{t}^2 dt$, and the integrated variance  of the noise $\epsilon_{t}$, $\int_0^T \gamma_{t}d t $.  We can estimate  $\int_0^T \tilde{\sigma}_{t}^2 dt$ and $\int_0^T \gamma_{t}d t $ separately, while for $IV_{T}$ and $QrT_{T}$, we {propose an iterative procedure} in which an initial rough estimate of $c_{t}$ on a grid of $[0,T]$ is used to determine estimates of {$IV_{T}$ and $QrT_{T}$}. These estimates are then used to find suitable estimates of the optimal values for $\theta$ and $\beta$. {Finally,} these estimated $\hat{\theta}$ and $\hat{\beta}$  are applied in the kernel pre-averaging estimator {(\ref{PreAverEst0})} 
to refine our estimates of $c_{t}$ on the grid.

\begin{remark}
{\Blue A related problem, that is not being considered here in detail, is that of tuning the truncation level $v_n$ in the truncated estimators \eqref{MDKNN_truncated} and \eqref{PreAverEst0}. Most of the literature about this problem {\Blue has been} in the context of estimating the integrated variance $IV_T=\int_0^T \sigma_s^2ds$. It has been customary in econometric studies to adopt a power threshold of the form $v_n= c \Delta_n^\gamma$. The rule of thumb is to take a value of $\gamma$ close to $.5$ and $c$ that depends on an estimate of the volatility level. For instance, \cite{JacodTodorov:2014} took $\gamma=.49$ and $c=4\sqrt{BPV}$, where $BPV:= \frac{\pi}{2} \sum_{i=2}^n |\Delta_{i-1}^nX||\Delta_i^nX|$ is the Bipower variation. More recently, this issue has {\Blue also been studied in the literature} using more objective and {\Blue statistically valid approaches}, but only in the absence of microstructure noise. In the case of FA jumps and constant volatility $\sigma$, \cite{FigueroaLopezMancini:2019} showed that the optimal threshold (in terms of minimizing the conditional mean-square error) is asymptotically equivalent to $\sqrt{2\sigma^2 \Delta_n\ln(1/\Delta_n)}$ and proposed an itervative method to estimate $\sigma$. In the presence of small jumps that behave like those of an $\alpha$-stable L\'evy process, \cite{gong2021} showed that the optimal threshold is asymptotically equivalent to $\sqrt{(2-\alpha)\sigma^2 \Delta_n\ln(1/\Delta_n)}$. Again, these results are in the absence of microstructure noise and for the problem of estimating the integrated variance. However, given the local nature of spot volatility estimation, one can imagine that similar results may hold for the estimators \eqref{MDKNN_truncated} and \eqref{PreAverEst0}. We leave this problem for future research.}
\end{remark}

\subsection{Optimal selection of \texorpdfstring{$\theta$}{Lg} }
Recall we set $k_n = \frac{1}{\theta \sqrt{\Delta_n}} +\mathrm{o}\left(\frac{1}{\Delta_{n}^{1/4}}\right)$ and, thus, the parameter $\theta $ determines the length of the pre-averaging window $k_n$. The following corollary, which follows easily from Theorem \ref{thm}, {gives us a method to tune $\theta$ up.}
\begin{corollary}
{The optimal value $\theta^{*}$ of $\theta$, which is set to minimize the integrated asymptotic variance of the pre-averaging kernel estimator (\ref{PreAverEst0}), is such that}
\begin{equation} \label{eq:theta_optimal}
{(\theta^{\star})^{2} = \frac{\sqrt{\Phi_{12}^2 \left(\int_0^T c_t \gamma_t d t \right)^2 + 3\Phi_{11}\Phi_{22}\int_0^T \gamma_t^2 d t \int_0^T c_t^2 d t}  -  \Phi_{12}\int_0^T  c_t \gamma_t d t}{3\Phi_{11}\int_0^T\gamma_{t}^2 d t}.}
\end{equation}
\end{corollary}
\begin{remark} {Note that the local version of (\ref{eq:theta_optimal}) (i.e., the value of $\theta$ that minimizes the spot asymptotic variance $\delta^{2}_{1}(t)$) is such that:
\begin{equation} \label{eq:theta_optimalLocal}
(\theta^{\star,local}_{t})^{2} = c_{t}\frac{\sqrt{\Phi_{12}^2 + 3\Phi_{11}\Phi_{22}}  -  \Phi_{12}}{3\Phi_{11}{\Blue \gamma_{t}}}.
\end{equation}
In the context of integrated volatility estimation, \cite{jacod2015microstructure} obtained the same formula (see Eq.~(3.8) therein). It was also proposed a two-step procedure to implement it. However, in our simulation, we found out that the performance of the estimator  is less sensitive to the choice of \(\theta\) than to that of the bandwidth.}
\end{remark}

\subsection{Optimal bandwidth selection} \label{band_section}
From Theorem \ref{thm}, we can deduce that when ${m}_{n} $ (the {bandwidth in} $\Delta_{n}$ units) is of the form ${m}_{n} = \beta \Delta_n^{-3/4} $ for some constant $\beta \in (0,\infty)$, the optimal convergence rate of $\Delta_{n}^{1/8}$ is attained and 
we further have: 
\begin{equation*}
\Delta_n^{-1/8} \left(\hat{c}\left(k_{n}, {m}_{n},v_n\right)_{\tau} -c_{\tau}\right)\stackrel{st}{\longrightarrow} 
\beta^{-1/2} \left( Z_{\tau} + \beta Z^{\prime}_{\tau} \right).
\end{equation*} Therefore, the limiting distribution {\Blue above} has conditional variance ${\bar{\delta}^{2}(\tau)}:=\frac{1}{\beta}\delta_1^2(\tau)+\beta \delta_2^2(\tau)$, where \(\delta_1^2(\tau)\) and \( \delta_2^2(\tau) \) are given as in (\ref{eq:delta12}).
{The following result gives the optimal value of $\beta$ that minimizes $\int_{0}^{T}\bar{\delta}^{2}(\tau)d\tau$.}
\begin{corollary} \thlabel{bandwidth}
Let 
\begin{equation*}
	\Theta(\theta):=\Theta(\theta;g):=\frac{\Phi_{22}}{\theta}\int_0^Tc_{t}^{2}dt+2 \Phi_{12} \theta \int_0^T\gamma_{t}c_{t}dt+\Phi_{11}\theta^{3} \int_0^T\gamma_{t}^{2}dt.
\end{equation*}
With the bandwidth $b_n = m_n \Delta_n=\beta \Delta_{n}^{1/4}$, the optimal value of $b_n$, which is set to minimize {$\int_{0}^{T}\bar{\delta}^{2}_{\tau}d\tau$},
is given by 
\begin{equation} \label{eq:optimal_h}
{b}_{n}^{\star} =  \sqrt{\frac{\int_0^T\delta_1^2(t) dt}{\int_0^T\delta_2^2(t) dt}}\Delta_n^{1/4}= \Delta_n^{1/4}\sqrt{\frac{4{\Theta(\theta)} \int K^2(u) d u}{\int_0^T\tilde{\sigma}_{t}^2 dt\int L^2(v)dv}}.
\end{equation}
With this optimal bandwidth choice, the integrated variance {$\int_{0}^{T}\bar{\delta}^{2}(\tau)d\tau$} of the limiting distribution for the {\Blue scaled estimation error $\Delta_n^{-1/8} \left(\hat{c}\left(k_{n}, {m}_{n},v_n\right)_{\tau} -c_{\tau}\right)$} is given by 
\begin{equation}\label{AsymptVarbb}
2\sqrt{\int_0^T \delta_1^2(t)dt \int_0^T \delta_2^2(t) dt} = 4  \sqrt{{\Theta(\theta)}\int_0^T\tilde{\sigma}_{t}^2 dt
 \int K^2(u) d u \int L^2(v)dv}.
\end{equation}
\end{corollary}

Note that \(b_n^{\star}\) contains unknown theoretical quantities that need to be estimated in order to devise a plug in type estimator.  Under the assumption of $\gamma_t \equiv \gamma $, the variance of the noise, \(\gamma \), can be estimated using the estimator in \cite{zhang2005tale}: 
\begin{equation*}
\hat{\gamma} = \frac{1}{2n}\sum_{i=1}^n \left(Y^n_{i} - Y^n_{i-1}\right)^2.
\end{equation*}
For the estimation of  the IVV,   \(\int_0^T \tilde{\sigma}_{t}^2 dt\),  
we start by obtaining a preliminary estimate of the spot variance $c$ on the grid $\tau\in\{t_{i}\}_{i=0,\dots,n}$, via the estimator (\ref{PreAverEst0}), staring with some sensible initial estimates of the tuning parameter values.  For example, we can set \( b_n =m_{n}\Delta_{n}= \Delta_n^{1/4}\). Let us denote these initial {estimates as $\hat{c}_{t_{i},0}$}. We then {compute the sparse realized quadratic variation of the $\hat{c}_{t_{i}}$'s to estimate  the Integrated Volatility of Volatility \({\rm IVV}=\int_0^T \tilde{\sigma}_{t}^2 dt\):}
\[
	{\widehat{IVV}_{T,0}:=\sum_{i=0}^{[n/p]-1}(\hat{c}_{t_{(i+1)p},0}-\hat{c}_{t_{i p},0})^{2},}
\]
for some positive integer $p\ll n$. We also implemented a pre-averaging integrated variance estimator for the IVV based on the spot variance estimates. However, the choice of tuning parameters here could be tricky and the performance is similar to the {simpler} sparse Realized Variance estimator above. As for \(\int_0^T c^2_{{t}}dt\), we can simply compute the sum of squares of the preliminary estimates {$\hat{c}^{2}_{t_{i},0}$} and multiply by $\Delta_n$\footnote{{In the simulations, we also tried the preaveraged quarticity estimator of \cite{jacod2009microstructure} (Eq.~(3.14) therein) but the results were suboptimal.}}. 
Now with these estimates, we can calculate an estimate of the optimal bandwidth \(b_n^{\star} \) using the result of Corollary \ref{bandwidth}. Such an approximate optimal bandwidth can then be used to refine our estimates of the spot variance grid. Continuing this procedure iteratively, we hope to obtain good estimates of the optimal bandwidth.

Note that (\ref{eq:optimal_h}) sets the same bandwidth for the entire path of $X$.
We can also consider a local or non-homogeneous bandwidth: for $\tau \in [0,T]$,  the local bandwidth {is set} to minimize the {asymptotic} variance of  the estimation error at time $\tau$. {Concretely, by setting ${m}_{n} = \beta \Delta_n^{-3/4} $ and minimizing the asymptotic spot variance ${\bar{\delta}^{2}(\tau)}=\beta^{-1}\delta_1^2(\tau)+\beta \delta_2^2(\tau)$},  the optimal bandwidth is given by
\begin{equation} \label{OptSpot00}
{b_n^{\star, local}(\tau)}=  \frac{\delta_1(\tau)}{\delta_2(\tau)}\Delta_n^{1/4}
={\Delta_n^{1/4}}\sqrt{\frac{4{\Theta_{\tau}(\theta)} \int K^2(u) d u}{\tilde{\sigma}_{\tau}^2\int L^{2}(u)du}},
\end{equation} 
with $\delta_{1}(\tau)$ and $\delta_{2}(\tau)$ defined as in Theorem \ref{thm} and $\Theta_{\tau}(\theta)$ defined as:
	\begin{equation*}
	\Theta_{\tau}(\theta):=\frac{\Phi_{22}}{\theta}c_{\tau}^{2}+2 \Phi_{12}  \gamma_{\tau}\theta c_{\tau}+\Phi_{11} \gamma_{\tau}^{2}\theta^{3}.
\end{equation*}
With this optimal bandwidth, the variance of the limiting distribution for the estimation error is given by 
\begin{equation}\label{AsymptVarbbb}
2\delta_1(\tau) \delta_2(\tau) = 4  \sqrt{{\Theta_{\tau}(\theta)}\tilde{\sigma}_{\tau}^2
 \int K^2(u) d u \int L^{2}(u) du}.
\end{equation}
Since the local bandwidth has the flexibility to adapt to the volatility level, we may expect {that a data-driven estimate of the bandwidth $b_n^{\star, local}(\tau)$ in (\ref{OptSpot00}) should outperform a data-driven estimate of the homogeneous bandwidth $b_n^{\star}$ in (\ref{eq:optimal_h}). However, in our Monte Carlo simulations of Section \ref{simulation}, we found out this is not always the case. A possible explanation for this is given below (see also Remark \ref{ExplLocvsGlob} for further analysis).}

\begin{remark}\label{PossExpl00}
	 We can see the constant bandwidth (\ref{eq:optimal_h}) as an approximation of the optimal local bandwidth (\ref{OptSpot00}), where the {mean} values $\int_{0}^{T}\Theta_{t}(\theta)dt/T$ and $\int_{0}^{T}\tilde{\sigma}_{t}^2dt/T$ are used as proxies of the spot values $\Theta_{\tau}(\theta)$ and $\tilde{\sigma}_{\tau}^{2}$, respectively. These global proxies have the advantages of being easier and more accurate to estimate. {This may be one of the reasons why a data-driven estimate of the constant bandwidth $b_n^{\star}$ may be able to outperform a data-driven estimate of the local version $b_n^{\star, local}(\tau)$ in some situations.}
\end{remark}

\subsection{Optimal kernel function}\label{OptKernelSection}
With the optimal bandwidths of Section \ref{band_section}, we can now obtain a formula for the asymptotic variance, which enjoys an explicit dependence on the kernel function $K$. It is then natural to attempt to find the kernel that minimizes such a variance. As observed from (\ref{AsymptVarbb})  or (\ref{AsymptVarbbb}), we only need to minimize 
\begin{equation*}
I(K)  =  \int K^2(u) d u \int L^2(u) du=
\int K^2(u) d u \iint_{x y\geq 0} K(x)K(y)(|x|\wedge{}|y|)d x d y,
\end{equation*} 
over all kernels $K$ such that $\int K(u)du=1$, where for the second equality above we used that \(L(t) = \int_{t}^{\infty} K(u) d u \mathbf{1}_{\{t>0\}}-\int_{-\infty}^{t} K(u) d u \mathbf{1}_{\{t \leq 0\}}\).
It has been proved in \cite{FigLi}, Section 4.1, that,  among all the kernel functions satisfying \thref{kernel}, the exponential kernel function $K^{\exp}(x)=\frac{1}{2} \exp (-|x|)$ is the one that minimizes the functional $I(K)$. {\Blue \cite{FigLi} (see Remark 4.2 therein) showed that, compare to the two-sided uniform (resp., Epanechnikov) kernels, the integrated asymptotic variance can be reduced by about 14\% (resp., 6\%) when using exponential kernel.}
\cite{FigLi} also showed that exponential kernels have {\Blue a computational advantage} since they {\Blue enable us} to reduce the time complexity for estimating the volatility on all the grid points {\Blue $t_{1}<\dots<t_n$}, from $O(n^{2})$ to $O(n)$. This property is particularly useful when working with high-frequency observations, {\Blue where $n$ is quite large.}


\subsection{Tuning parameters under the absence of microstructure noise}\label{OptNoMicroSection}
{By following the same arguments as above, we can determine the optimal bandwidth parameter and kernel function for the estimators (\ref{MDKNN})-(\ref{MDKNN_truncated}) under the no-microstructure-noise model (\ref{eq:X})-(\ref{eq:sigma}). Specifically, we first take a bandwidth of the form $b_{n}=\beta\Delta_{n}^{1/2}$ ($\beta\in(0,\infty)$), which, from Theorem \ref{thm_no_noise}, leads {\Blue to} the best possible rate of convergence $\Delta_{n}^{-1/4}$ of (\ref{MDKNN})-(\ref{MDKNN_truncated}). In that case, the asymptotic variance will take the form $\bar{\delta}_{\tau}^{2}=\beta^{-1}\delta_1^2(\tau)+\beta \delta_2^2(\tau)$, where 
\[
	\delta_{1}^{2}(\tau)=2 c_{\tau}^{2}\int K^2(u) du,\qquad 
	\delta_{2}^{2}(\tau)=2\tilde{\sigma}_{\tau}^{2} \int L^2(t)dt.
\]
Then, the optimal value of $\beta$ that minimizes the asymptotic variance is $\beta^{*}=\delta_{1}(\tau)/\delta_{2}(\tau)$, leading to the optimal bandwidth 
\begin{equation} \label{OptSpot00NMN}
\tilde{b}_n^{\star, local} =  \frac{\delta_1(\tau)}{\delta_2(\tau)}\Delta_n^{1/2}
=\Delta_n^{1/2}\sqrt{\frac{2 c_{\tau}^{2}\int K^2(u) d u}{\tilde{\sigma}_{\tau}^2\int L^{2}(u)du}}.
\end{equation} 
Plugging $\beta^{*}$ into $\bar{\delta}_{\tau}^{2}$, leads to the optimal asymptotic variance of 
\[
	2\delta_1(\tau)\delta_2(\tau)= 4 \sqrt{c_{\tau}^{2}\tilde{\sigma}_{\tau}^2
 \int K^2(u) d u \int L^{2}(u) du},
\]
which, as before, is minimized by the two-side exponential kernel $K(x)=2^{-1}e^{-|x|}$.}

\section{Simulation Study} \label{simulation}
In this section, we study the performance of the kernel pre-averaging estimators {\Blue (\ref{eq:non_truncated_pre}) and (\ref{PreAverEst0}),} together with the implementation procedure described in Subsection \ref{band_section}, and compare the results with  the Two Scale Realized Spot Variance (TSRSV) estimator proposed in \cite{zu2014estimating}.

\subsection{Simulation design and performance metrics} \label{heston_section}

 We implemented two {\Blue different data generating} models:  {\Blue a Heston model and a one-factor stochastic volatility (SV1F) model. More specifically, in Subsections \ref{ElmiJumpSec}-\ref{OptimalBdSec}, we consider the Heston model:}
\begin{equation} \label{HestonMld00}
\begin{split}
&{\Blue Y_{t_i}=X_{t_i}+\varepsilon_{t_i}},\\
    &\mathrm{d} X_{t}=\left(\mu-c_{t} / 2\right) \mathrm{d} t+{c_{t}^{1/2}} \mathrm{d}  W_{t} +  J_{t}^{X}\mathrm{d} N_{t}^{X},\\
     &\mathrm{d} c_{t}=\kappa\left(\alpha-c_{t}\right) \mathrm{d} t+\gamma c_{t}^{1 / 2} \mathrm{d}  B_{t} +\sqrt{c_{t-}} J_{t}^{c} \mathrm{d} N_{t}^{c},
     \end{split}
\end{equation}
where we assume $B_t = \rho W_t + \sqrt{1 - \rho^2} {\tilde{W}_t}$, with $\tilde{W}$ being a Brownian motion independent with $W$. 
We adopt the same parameter values as in  \cite{zhang2005tale}, but properly normalized so that the time unit is one day:
\begin{equation} \label{ParValHston00}
{\mu=0.05 / 252,\quad \kappa=5 / 252, \quad \alpha=0.04 / 252,\quad \gamma=0.5 / 252, \quad \rho=-0.5.}
\end{equation}
We set the noise as $\epsilon_{i}^{n}:=\epsilon_{t_{i}} \stackrel{i.i.d.}{\sim} \mathcal{N}\left(0,0.0005^{2}\right)$,  and the initial values to $X_0 = 1$ and $c_0 = 0.04/252$.  The jump parameters are taken from \cite{chen2018inference} and set to be \(J^X_t\, {\Blue \stackrel{{\rm iid}}{\sim}}\, N(-0.01,0.02^2)\), \(N_{t+\Delta}^{X}-N_{t}^{X} \sim \text {Poisson }(36 \Delta/252)\), \(\log \left(J_{t}^{c}\right) \,{\Blue \stackrel{{\rm iid}}{\sim}}\, N(-5, 0.8)\), and \(N_{t+\Delta}^{c}-N_{t}^{c} \sim \frac{1}{\sqrt{252}}\operatorname{Poisson}(12 \Delta/252)\), {\Blue with all these random processes being mutually independent.} 

{\Blue We also consider the One Factor Stochastic Volatility (SV1F) model (cf. \cite{zu2014estimating}, \cite{barndorff2008designing}, \cite{Yu2}):
\begin{equation} \label{eq:sv1f_model}
\begin{aligned}
& Y^n_i = X^n_i + \epsilon^n_i\\
&\mathrm{d} X_{t}=\mu \mathrm{d} t+\exp \left(\beta_{0}+\beta_{1} {\Green \gamma_{t}}\right) \mathrm{~d} W_{t} + \mathrm{~d} J_t, \\
&\mathrm{d} {\Green \gamma_{t}}=\alpha {\Green \gamma_{t}} \mathrm{~d} t+\mathrm{d} B_{t}.
\end{aligned}
\end{equation} 
The model above is adopted in Subsections \ref{ComSubH} and \ref{CmpTrNoTR} with different parameter values that will be specified therein.}

Throughout, we use the usual triangular weight function \(g(x) = 2 x \wedge(1-x) \). We simulate data for one day ($T=1$), and assume the data is observed once every second, with 6.5 trading hours per day. The number of observation is  then \(n = 23 400\).
For the $j$th simulated path \(\{{X^{(j)}_{t_i}}: 0\leq i\leq n, t_i = i T/n\}\), we estimate the corresponding skeleton of the spot variance process, \(\{c_{t_i,j}\}_{i=1,\dots,n}\), for a given pre-averaging parameter \(\theta\) and a bandwidth parameter \({\tilde{\beta}}\)  (the bandwidth is {\Blue then given by} \(\tilde{\beta} \Delta_n^{1/4}\)). The estimated path is denoted as \(\{\hat c_{t_i,j}\}_{i=1,\dots,n}\). Next, we calculate the average of the squared errors (ASE),
 \[
 	 ASE_{j} =  \frac{1}{n-2l+1}\sum_{i=l}^{n-l}\left(\hat{c}_{t_i,j} - c_{t_i,j} \right)^2.
\]
Here, \(l=[0.1n] \)  is used to further alleviate boundary effects. Then, we take the square root of the average of the \(ASEs\) over all the simulated  paths:
 \[
 	\widehat{RMSE} = \sqrt{\frac{1}{m}\sum_{j=1}^{m} ASE_{j}},
\] 
where $m$ is the number of simulations. This is an estimate of
 \[
 	RMSE = \sqrt{\E\left[\frac{1}{n-2l+1}\sum_{i=l}^{n-l}\left(\hat{c}_{t_i} - c_{t_i} \right)^2 \right]}.
\] 

\subsection{Elimination of jumps and truncation}\label{ElmiJumpSec}
{In this subsection, we will show that} the truncation in the estimator (\ref{PreAverEst0}) {does a good job in eliminating} the jumps of the process (\ref{HestonMld00}). {To this end, we compare the performance of the truncated estimator \( \hat{c}\left(k_{n}, {m}_{n},v_n,1\right)\) in (\ref{PreAverEst0}), with that of the non-truncated estimator 
\begin{equation}\label{PreAverEst0_cts}
\hat{c}\left(k_{n}, {m}_{n}, Y^*\right)_{\tau}:=\frac{1}{\phi_{k_{n}}\left(g\right)} \sum_{j=1}^{n-k_n+1} K_{{m}_{n} \Delta_n}\left(t_{j-1} - \tau\right)\left(\left(\overline{Y}_{j}^{*n}\right)^2   - \frac{1}{2}\widehat{Y}_{j}^{*n}\right),
\end{equation} 
applied to the continuous} Heston model:
\begin{equation} \label{HestonMld00_cts}
\begin{split}
&Y_{i}^{*n}=X_{i}^{*n}+\varepsilon_{i}^{n},\\
    &\mathrm{d} X^*_{t}=\left(\mu-c^*_{t} / 2\right) \mathrm{d} t+\sqrt{c^*_{t}} \mathrm{d}  W_{t} ,\\
     &\mathrm{d} c^*_{t}=\kappa\left(\alpha-c^*_{t}\right) \mathrm{d} t+\gamma\sqrt{c_{t}^{*}} \mathrm{d}  B_{t} .
     \end{split}
\end{equation}
We set \(\beta = 1\), \(\theta = 5\), and $v_n = 1.8\times \sqrt{BPV}(k_n \Delta_n)^{0.47}$, where $BPV = \frac{\pi}{2} \sum_{i=2}^{n}\left|\Delta_{i-1}^{n} X \| \Delta_{i}^{n} X\right|$\footnote{A similar threshold is applied in \cite{JacodTodorov:2014}.}. {\Blue The results are based on 2000 simulated paths of both \(Y\) and \(Y^{*}\)}. 
As proposed in \cite{kristensen2010nonparametric}, in order to alleviate the {edge effects, we replace $K_{{m}_{n} \Delta_n}\left(t_{i-1} - \tau\right)$ in (\ref{PreAverEst0}) and (\ref{PreAverEst0_cts}) with} 
\begin{equation*}
K^{adj}_{m_n \Delta_n}\left( t_{i-1} - \tau\right) = \frac{K_{{m}_{n} \Delta_n}\left(t_{i-1} - \tau\right)}{\Delta_n\sum_{j=1}^{n-k_n+1} K_{{m}_{n} \Delta_n}\left(t_{j-1} - \tau\right)}.
\end{equation*} 
The \(\widehat{RMSE}\) of the three estimators is reported in Table \ref{table:jump}.
\begin{table}[h!]
\begin{center} 
\begin{tabular}{ | c | c | c | c | c |}
\hline
  &{\begin{tabular}{c}
           \( \hat{c}\left(k_{n}, {m}_{n}, Y^*\right) \) 
     \end{tabular}}&{\begin{tabular}{c}
     \( \hat{c}\left(k_{n}, {m}_{n}, Y\right)\) 
     \end{tabular}} & {\begin{tabular}{c}
$\hat{c}\left(k_{n}, {m}_{n}, v_n\right)$
     \end{tabular}}\\
\hline
$\widehat{RMSE}\times 10^{5}$ &5.483386    &16.83420
 &5.419338    \\
\hline
\end{tabular}
\caption{Comparison between truncated and non-truncated estimators.}
\label{table:jump}
\end{center}
\end{table}

{\Blue These results suggests that} the {truncation procedure} can effectively eliminate the jumps under this Heston model, since the estimated RMSE of the truncated estimator for the model (\ref{HestonMld00}) is even less than that of the non-truncated estimator based on the continuous model (\ref{HestonMld00_cts}). 

\subsection{Validity of the asymptotic theory and necessity of de-biasing}

We first show that the asymptotic behavior of the estimation error is consistent with our theoretical result. By \thref{bandwidth}, the optimal rate of convergence of the estimation error is attained when the bandwidth takes the form  \(m_n^{\star} \Delta_n = \beta \Delta_n^{1/4}\), for some \(\beta \in (0,\infty) \), and, thus, we only analyze the {\Blue case 1(i) ($\beta\in(0,\infty)$)} of Theorem \ref{thm}. We aim to estimate the {\Blue spot variance \(c_{0.5}\) in the Heston model \eqref{HestonMld00} without jumps. Accordingly and for simplicity, we use the untruncated pre-averaging kernel estimator \eqref{eq:non_truncated_pre}. We take \(\beta = 1\) and exponential kernel}. The histogram of the estimation errors, $\hat{c}_{0.5} - c_{0.5}$, based on 25,000 simulated paths, is shown in Figure \ref{fig:hist_exp}. We also plot the theoretical density of the estimation error as prescribed by Theorem \ref{thm} but with the true parameter values for $\gamma$ and $\theta $, and replacing $c_{0.5}$ with the {\Blue average value} of $c_{0.5}$ over all 25,000 path. As it can be seen, the theoretical density is consistent with the empirical results.
\begin{figure}[h]
    \centering
    \includegraphics[scale = 0.5]{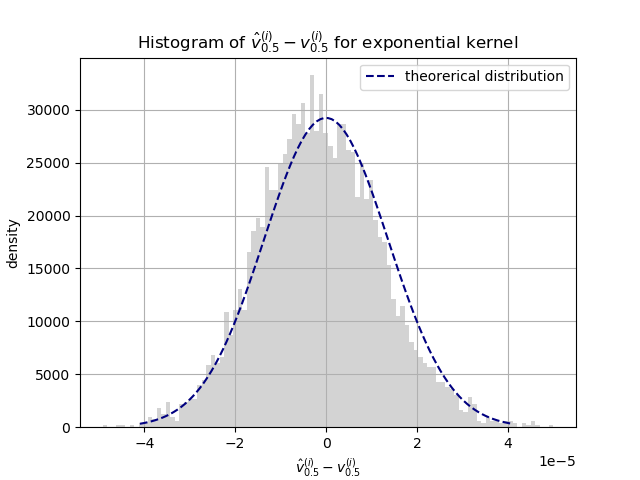}
    \caption{Histogram of $\hat{c}_t - c_t$ at $t=0.5$ and the density of the theoretical limiting distribution.}
    \label{fig:hist_exp}
\end{figure} 

To investigate the need of the bias correction term \(\widehat{Y}_{j}^{n}\)  in $\hat{c}\left(k_{n}, {m}_{n},v_n,1\right)_{\tau}$, let us consider a new estimator without the bias correction term, \(\tilde{c}_{\tau} = \sum_{j=1}^{n-k_n+1} K_{m_n \Delta_n}\left(t_{j-1} - \tau\right)\left(\overline{Y}_{j}^{n}\right)^2 \mathbbm{1}_{\{|\bar{Y}^n_j | \leq v_n\}}  \). We show the histogram of the estimation errors  $\tilde{c}_{0.5} - c_{0.5} $ for 25,000 simulated paths, and, for comparisons, also plot the same theoretical asymptotic density function of Figure \ref{fig:hist_exp}.   As shown in {left panel of} Figure \ref{fig:debias}, the estimator {\Blue $\tilde{c}_{0.5}$} significantly overestimates the spot variance, which shows the necessity of  the bias correction term  \(\widehat{Y}_{j}^{n}\)   in (\ref{PreAverEst0}).
\begin{figure}
    \centering
    \includegraphics[scale = 0.45]{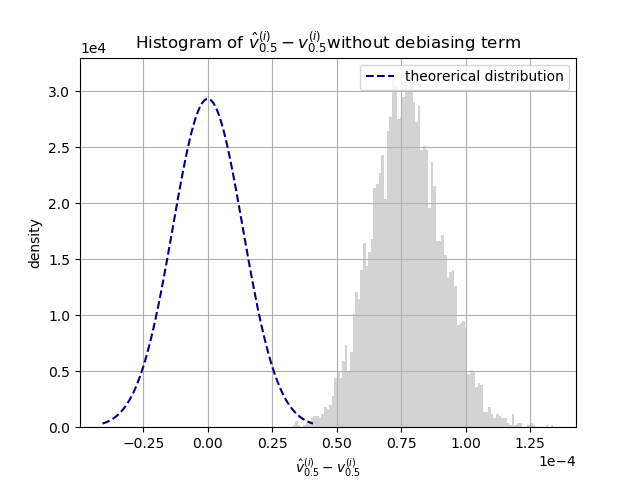}
    \includegraphics[scale = 0.45]{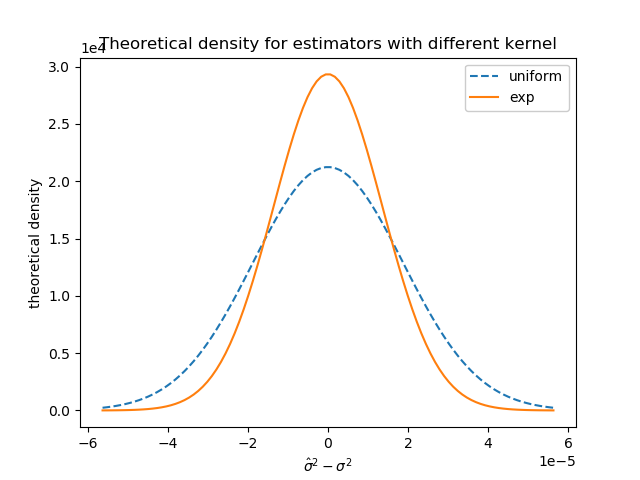}
    \caption{Left Panel: The effect of bias correction term. Right Panel: Comparison of the asymptotic distribution between uniform and exponential kernel.}
    \label{fig:debias}
\end{figure}

\subsection{Performance for different kernels}

Before analyzing the empirical performance of the estimators for different kernels, we compare the theoretical asymptotic densities of the estimation error for the exponential and uniform kernels. This is shown in {right panel of Figure \ref{fig:debias}}. We can see therein that, as predicted in Subsection \ref{OptKernelSection}, the exponential kernel estimator {\Blue has smaller} asymptotic variance.

We now proceed to compare the finite sample performance of {\Blue the untruncated pre-averaging kernel estimator \eqref{eq:non_truncated_pre}} 
for different kernels in the Heston model \eqref{HestonMld00} {\Blue without jumps}. We assume both a non-leverage setting (\(\rho =0\)) and a negative correlation setting (\(\rho = -0.5\)). 
{\Blue We} fix \(\theta  = 5\) and apply the iterative homogeneous bandwidth selection method  introduced in Subsection \ref{band_section} {\Blue with different} kernels. We report the estimated $RMSE$  with the initial bandwidth $\beta =1 $ and the result of iterative bandwidth selection method after one iteration in Table \ref{table:diff_kernel}  for the following {\Blue four kernels}:
\begin{equation*}
\begin{split}
&K_{{\rm exp}}(x) = \frac{1}{2} e^{-|x|} , \quad
K_{unif}(x) = \frac{1}{2} \mathbbm{1}_{\{|x|<1\}} \\
&K_{1}(x) = |1-x| \mathbbm{1}_{\{|x|<1\}} , \quad
K_{2}(x) = \frac{3}{4} (1-x^2) \mathbbm{1}_{\{|x|<1\}} .
\end{split}
\end{equation*}
This shows that, indeed, the exponential kernel provides the best performance.
 
\begin{table}[h!]
\begin{center} 
\begin{tabular}{ | l | l | c |}
\hline
 \multicolumn{3}{|c|}{{$\widehat{RMSE}\times 10^{5}$ ($\rho  = 0$)}} \\
\hline
Kernel & $\beta = 1$& \begin{tabular}{c}Optimal\\
 Bandwidth Selection\end{tabular}  \\
\hline
$K_{{\rm exp}}$ &1.400  & 1.068   \\
\hline
$K_{unif}$ &1.890  & 1.608 \\
\hline
$K_1$ &2.173  & 1.648 \\
\hline
$K_2$ &2.064  & 1.476 \\
\hline
\end{tabular}
\caption{Comparison of different kernel functions.}
\label{table:diff_kernel}
\end{center}
\end{table}

\subsection{Optimal bandwidth}\label{OptimalBdSec}
First, we  show that the suboptimal bandwidth, which corresponds to  $\beta = 0$ in Theorem \ref{thm},  indeed performs worse than the optimal bandwidth, even though its asymptotic variance is easier to estimate without the \(\beta Z^{\prime}_{\tau}\) term.  {\Blue For simplicity, we again only consider the Heston model \eqref{HestonMld00} without jumps and the untruncated pre-averaging kernel estimator \eqref{eq:non_truncated_pre}. We will compare the truncated and untruncated versions in more detail below in Subsection \ref{CmpTrNoTR}.} 

In Table \ref{table:opt_subopt}, we compare the optimal bandwidth \(h_1 = \beta \Delta_n^{1/4} \) with the suboptimal  bandwidths \( h_2 = \beta \Delta_n^{0.28} \) and \(h_3 = \beta \Delta_n^{0.3} \), using the exponential kernel with \(\beta = 1,2,3,4\) respectively, {\Blue based} on 1000 simulated path.  The results show the advantage in using the optimal bandwidth for the same level of the coefficient $\beta $.

\begin{table}[h!]
\begin{center} 
\begin{tabular}{ | c | c | c |c|}
\hline
 \multicolumn{4}{|c|}{{$\widehat{RMSE}\times 10^{5}(\rho = -0.5)$}} \\
 \hline
Bandwidth & $h_1$(optimal) &  $h_2$ (suboptimal) & $h_3$ (suboptimal)\\
\hline
$\beta = 1$ & 1.418  & 1.605   & 1.754 \\
\hline
$\beta = 2$ & 1.133 & 1.225  &  1.308\\
\hline
$\beta = 3$ & 1.077 & 1.121 & 1.678\\
\hline
$\beta = 4$ & 1.050 & 1.073 & 1.104 \\
\hline
\end{tabular}
\caption{Comparison between optimal bandwidth and suboptimal bandwidth}
\label{table:opt_subopt}
\end{center}
\end{table}

Next, we compare the results of the iterative homogeneous and local bandwidth selection methods, as discussed in Subsection \ref{band_section}. Based on some initial simulations, we observed that the parameter \(\theta  \), which controls the length of the pre-averaging window \(k_n\)  as \(k_n = \frac{1}{\theta\sqrt{\Delta_n}}\),  has comparatively smaller effect on the performance of estimator than  that of the bandwidth. Therefore, throughout this section, we fix \( \theta = 5\), which is computed by (\ref{eq:theta_optimal}) using true parameter values, and consider different bandwidth selection techniques\footnote{We also consider other values of $\theta$ and the results were similar.}. 

In Table \ref{table:diff_band}, we report the estimated RMSE  for different bandwidth selection methods.  For the homogeneous bandwidth selection method  (\ref{eq:optimal_h}),  we apply the realized variance of sparsely sampled  (5 min) spot variance estimates $\{\hat{c}_{t_{i}}\}$ to estimate  the vol vol \(\int_0^T \tilde{\sigma}^2_t dt \) as described in Section \ref{band_section}.  We fix the estimated vol vol after the first iteration to prevent the increased variance brought by the iterative method. The first two iterations are shown in the first two columns of the table and we can see that the second iteration  does not improve the result significantly. Therefore, one iteration of the bandwidth selection method is sufficient in practice.  
For the local bandwidth method,  we use \(\int_0^T \tilde{\sigma}^2_t dt/T \) as a proxy of $\tilde{\sigma}^2_{\tau}$  in the formula (\ref{OptSpot00}). 
As a reference, we  also give the results of using an oracle optimal bandwidth, which is computed by the true parameter values and the simulated spot variance process with Eqs.~(\ref{eq:optimal_h}) and (\ref{OptSpot00}) for the optimal homogeneous and optimal local bandwidths,  respectively.  In the last column, we provide the result of a semi-oracle type of bandwidth, where we use the estimated spot variance ``skeleton" $\{\hat{c}_{t_{i}}\}$ to estimate $\int_0^Tc_t dt$ and $\int_0^T c^2_t dt$, via Riemann sums\footnote{We also apply the pre-averaging estimate of quarticity given in \cite{jacod2010limit},  but the results were less optimal.}, while using the true parameter of $\gamma$ given in (\ref{ParValHston00}) to estimate $\int_0^T \tilde{\sigma}^2_t dt=\gamma^{2} \int_{0}^{T}c_{t}dt$. The last simplification is possible due to the special structure of the diffusion coefficient of variance process in the Heston model (\ref{HestonMld00}). A  similar approach can be applied to other popular volatility models such as CEV models. As we can see therein, the data-driven approaches (1st two columns) are quite close to the oracle and semi-oracle estimates.
\begin{table}[h!]
\begin{center} 
\begin{tabular}{ | c | c | c | c | c |}
\hline
 \multicolumn{5}{|c|}{ $\widehat{RMSE}\times 10^{5}(\rho = -0.5)$} \\
 \hline
   			 & 1st Iter. & 2nd Iter & Oracle & Semi-oracle\\
 \hline
 homogeneous &1.0530 	&1.0529&1.0540 &1.0533 \\
 \hline
 local &	1.0571 &	1.0551 &1.0542 &1.0547 \\
\hline
\end{tabular}
\caption{Comparison of different bandwidth selection methods based on 1000 simulations.   RMSE for initial bandwidth \(\beta = 1\) is $1.4086\times 10^{-5}$. Columns 2 and 3 show the results corresponding to the 1st and 2nd iterations of bandwidth selection methods.  Column 4 and 5 show the result using oracle and semi-oracle bandwidths, respectively.}
\label{table:diff_band}
\end{center}
\end{table}
 \begin{remark}\label{ExplLocvsGlob}
The estimator with  local bandwidth has the flexibility to adjust its bandwidth at different times based on the data. Therefore, theoretically, this estimator should be able to achieve a lower {\Blue value of the integrated asymptotic variance $\int_{0}^{T}\bar{\delta}^{2}(\tau)d\tau$, which, as defined in \thref{bandwidth}, is given by}:
\begin{equation*}
\Delta_n^{1/4} {\Blue \int_0^T\left( \frac{1}{\beta_{t}}\delta_1^2(t)+\beta_{t} \delta_2^2(t)\right) d t}.
\end{equation*}
However, {\Blue our} simulations show that the performance of the local bandwidth 
is almost the same as that of the homogeneous bandwidth. To further investigate this phenomenon, in {\Blue the} left panel of Figure \ref{fig:zeta1bdw}, we show the estimated RMSE for different {\Blue times  \( \tau\)} against the parameter $\beta$ in the bandwidth formula $b_{n}=\beta\Delta_{n}^{1/4}$. As before we simulate the Heston model (\ref{HestonMld00}) with the same parameters as in (\ref{ParValHston00}), but with the vol vol parameter $\gamma= 1/252$. We can conclude from the figure that the optimal {\Blue $\beta$-value is almost the same for different $\tau$'s, and this value is also close to the theoretical optimal homogeneous} bandwidth based on the asymptotic variance of the estimator. Thus, an estimator with homogeneous bandwidth can achieve a similar result without extra computation cost. This trend is less obvious when the vol vol parameter $\gamma$ is relatively small. In the right panel of Figure \ref{fig:zeta1bdw} we show the estimated RMSE vs. $\beta$ when $\gamma = 0.5/252$. {\Blue In that case, the} perceived almost flat trend as the bandwidth increases shows  that  the realized variance can serve as a good proxy of the spot volatility, at least for the purpose of tuning the parameters of the estimators, since the spot volatility estimator degenerates to the integrated volatility estimator when the bandwidth gets large. Note, however, that the MSE paths are slowly  tickling up as $\beta$ increases and each of those paths exhibit an optimal bandwidth. These are again  relatively close for different times $\tau$ and also close to  the theoretical optimal {\Blue homogeneous} bandwidth. In conclusion, when the vol vol parameter is not known, the theoretical optimal bandwidth can provide a good guideline for the empirical experiments and a homogeneous bandwidth is sufficient in achieving similar result as local bandwidth while reducing the estimation error and computation cost caused by {\Blue the} latter. 
\begin{figure}
    \centering
    \includegraphics[scale = 0.45]{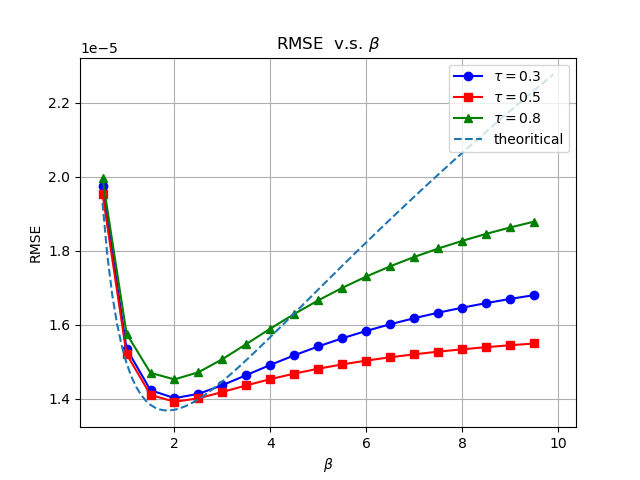}
    \includegraphics[scale = 0.45]{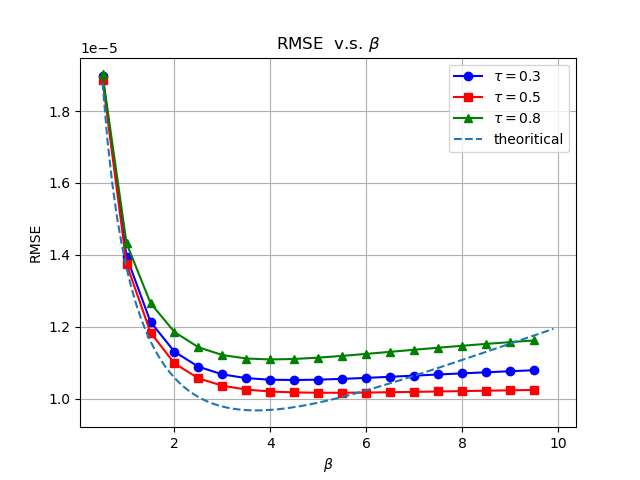}
    \caption{Left Panel: MSE v.s. bandwidth when {$\gamma = 1/252$}. Right panel: MSE v.s. bandwidth when {$\gamma = 0.5/252$}}
    \label{fig:zeta1bdw}
\end{figure}
\end{remark}


\newpage
\subsection{Comparison with TSRSV}\label{ComSubH}
{\Blue In this section, we adopt the model \eqref{eq:sv1f_model} with the same parameters as \cite{zu2014estimating}}:
\begin{equation}
\mu = 0.03, \;\beta_1 = 0.125,\;\alpha=-0.025, \;\rho = -0.3,\; \beta_0 = \beta^2_1/(2\alpha).
\end{equation} 
{\Blue We also take $\gamma_0\sim\mathcal{N}\left(0, -\frac{1}{2\alpha}\right)$ and $J=0$.}
The microstructure noise  $\epsilon^n_i$ is set to be $\epsilon_{i}^{n}:=\epsilon_{t_{i}} \stackrel{i.i.d.}{\sim} \mathcal{N}\left(0,\omega^2\right) $, {\Blue where, as in \cite{zu2014estimating}, $\omega^2$ can take one of three possible levels: $0.0001$, $0.001$, and $0.01$.}

For the TSRSV estimator, we implement the smoothing version of the TSRSV  (see \cite{zu2014estimating}, Section 3.1), denoted as {\Blue $\hat{c}_{TSRSV}$}, and calculate the bandwidth and scale parameters according to Section 3.4 in \cite{zu2014estimating}. For our pre-averaging estimator, we implement the non-truncated version (denoted as {\Blue $\hat{c}_{PA}$}) with exponential kernel and use the iterative method described in {\Blue Subsection} \ref{band_section} for bandwidth selection. {\Blue We consider two sampling frequencies: 1-sec or 5-sec.} In Table \ref{table:TSRSV}, we report the RMSE of the two estimators. {\Blue 
As shown in the table, the pre-averaging estimator has a superior performance, especially when the noise level is large.} 

\begin{table}[h!]
\begin{center} 
\begin{tabular}{  c   |c  c |c c  |c  c    }
\hline
   \multirow{2}{5em}{Frequency} &\multicolumn{2}{c}{$\omega^2 = 0.0001$} &
  \multicolumn{2}{c}{$\omega = 0.001$} & \multicolumn{2}{c}{$\omega = 0.01$} \\
  & $\hat{c}_{PA}  $ &$\hat{c}_{TSRSV}$ & $\hat{c}_{PA}  $ &$\hat{c}_{TSRSV}$ & $\hat{c}_{PA}  $ & $\hat{c}_{TSRSV}$ \\
\hline 
1 sec & 0.0411 & 0.0727 & 0.0546 & 0.1121 & 0.0634& 0.2345\\
5 sec & 0.0505 & 0.1487 & 0.0649 & 0.1392 & 0.1066 &0.3005
\end{tabular}
\caption{Comparison between TSRSV and kernel pre-averaging estimator. {\Blue We set the parameter $\theta$ in \eqref{AsympCndkn} to be $5$, $5$, and $1.5$ for noise level $0.0001, 0.001,0.01$ respectively. For the TSRSV estimator, to reduce the computational cost, instead of choosing the initial bandwidth using the cross validation method proposed in \cite{kristensen2010nonparametric} as did in Section 4.2.1 \cite{zu2014estimating}, we use the bandwidth already selected in  \cite{zu2014estimating}, Table 8-10, as the initial values to estimate the vol vol. Note that the obtained RMSE of the smoothing TSRSV estimator under the different noise levels (0.0727, 0.1121, and 0.2345) match the results in \cite{zu2014estimating}, who reported the values 0.094, 0.118, and 0.223, respectively.}}
\label{table:TSRSV}
\end{center}
\end{table}


\subsection{{\Blue Comparison between the truncated and untruncated estimators}}\label{CmpTrNoTR}
{\Blue In this subsection, we study the two versions of the truncated estimators (\ref{PreAverEst0}) and the non-truncated estimator (\ref{eq:non_truncated_pre}) under various levels of jump size and sample frequency, using the simulation setting in \cite{Yu2}. The parameters therein are chosen from \cite{Huang2005}}:
\begin{equation}
\mu = 0.03,\; \beta_1 = 0.125,\;\alpha=-0.1, \;\rho = 0, \;\beta_0 = 0.
\end{equation}  
We also conduct our study in the same experiment design {\Blue as} \cite{Yu2}. {\Blue More specifically, we consider three levels of jump activity (no jumps, compound Poisson jumps, and Variance Gamma jumps); 3 noise levels ($\omega = 0025,0.035,0.05 $); and three different sample frequencies (1 observation every 10 sec, every 30 sec, and every 60 sec.). In the case of finite activity jumps,  $J_{t}=\sum_{j=1}^{N_{t}} Z_{\tau_{j}}$ with  $Z_{\tau_{j}} \sim N(0,\sigma_Y^2) $ and {\Blue $\{N_t\}_{t\geq{}0} \sim Poisson(3)$}, while in the case of infinite activity jumps, $J_t = cG_t + \eta \tilde{W}_{G_t} $ with $G_t \sim Gamma(t/b,b)$, $b = 0.23$, $c = -0.2$, $\eta = 0.2$, and $\tilde{W} $ is an independent Brownian motion,} as  in \cite{Mancini2009}.

{\Blue In Table \ref{table:finite_jump}, we report the RMSE of our non-truncated estimators (\ref{eq:non_truncated_pre}) and truncated estimators (\ref{PreAverEst0}), denoted as $\hat{c}_{T,1}$, $\hat{c}_{T,2}$ and $\hat{c}_{Non}$, respectively.
We observe that when there are no jumps, all three estimators have similar performance, with the non-truncated estimator giving slightly better results. However, when jumps are present, the truncated estimators have a much superior performance, especially when the jump size is large, and $\hat{c}_{T,2} $ appears to be more effective at eliminating jumps compared with $\hat{c}_{T,1} $.}

{\Blue \cite{Yu2} also proposed  a pre-averaging kernel estimators {\Blue for the spot volatility. As mentioned in the introduction, their estimator has a different de-biasing term, which could affect the finite-sample performance, and their asymptotic normality is only established with} suboptimal convergence rate.} Our results in Table \ref{table:finite_jump} are comparable with \cite{Yu2}, Section 5. For example, the RMSE 0.1421 under $\omega = 0.025, \sigma_Y = 0$ with 10 sec data is close to the RMSE $\sqrt{0.0201} = 0.1417$ in \cite{Yu2}.

\begin{table}[h!]
\begin{center} 
\begin{tabular}{  c   |c  c c |c c c  |c  c  c   }
\hline
   \multirow{2}{5em}{Frequency} &\multicolumn{3}{c}{$\omega = 0.025$} &
  \multicolumn{3}{c}{$\omega = 0.035$} & \multicolumn{3}{c}{$\omega = 0.05$} \\
  & $\hat{c}_{T,2}$ &$\hat{c}_{T,1}$ & $\hat{c}_{Non} $ & $\hat{c}_{T,2}$ & $\hat{c}_{T,1}$ &  $\hat{c}_{Non} $ & $\hat{c}_{T,2}$ &$\hat{c}_{T,1}$ & $\hat{c}_{Non} $\\
\hline 
\multicolumn{10}{l}{Scenario A: Diffusion with no jumps $\sigma_Y =0 $ }\\
10 sec &0.1421&	 0.1421&	0.1420&	0.1474&	0.1474&	0.1472&	0.1689&	0.1690&	0.1677\\
30 sec &0.1724&	 0.1724&	0.1719&	0.1763&	0.1763&	0.1758&	0.1846&	0.1846&	0.1837\\
60 sec &0.2057&	 0.2057&	0.2054&	0.2085&	0.2085&	0.2083&	0.2161&	0.2161&	0.2160\\
\hline 
\multicolumn{10}{l}{Scenario B: Diffusion with small jumps $\sigma_Y =0.5 $ }\\
10 sec &0.1311&	0.1630&	1.0120&	0.1347&	0.1641&	1.0004&	0.1592&	0.1805&	0.9943\\
30 sec & 0.1660&	0.1893&	0.9758&	0.1690&	0.1913&	0.9798&	0.1797&	0.1971&	0.9783\\
60 sec &0.2083& 0.2209&	0.9738&	0.2113&	0.2215&	0.9732&	0.2205&	0.2255&	0.9730\\
\hline
\multicolumn{10}{l}{Scenario C: Diffusion with large jumps $\sigma_Y =1.5 $ } \\
10 sec & 0.1952	&0.9402	&9.4648&	0.2013&	0.9397&	9.3509&	0.2090&	0.9013&	9.1702\\
30 sec & 0.2262 &0.9183	&9.0481&	0.2256&	0.9132&	9.0438&	0.2279&	0.8990&	8.9954\\
60 sec & 0.2640	&0.8949	&8.9794&	0.2635&	0.8892&	8.9731&	0.2590&	0.8707&	8.9514\\
\hline
\multicolumn{10}{l}{Scenario D: Diffusion with jumps of infinite activity } \\
10 sec& 0.1246&	0.1247&	0.1271&	0.1330&	0.1331&	0.1349&	0.1529&	0.1530&	0.1547\\
30 sec&0.1576&	0.1577&	0.1594&	0.1587&	0.1588&	0.1598&	0.1728&	0.1729&	0.1740\\
60 sec&0.1940&	0.1940&	0.1952&	0.1959&	0.1959&	0.1966&	0.2063&	0.2063&	0.2076\\

\end{tabular}
\caption{The RMSE of the pre-averaging estimators. {\Blue We set the parameter $\theta$ in \eqref{AsympCndkn} to be $5$, $3$, and $2$ for 10-sec, 30-sec, and 60-sec data, respectively. The truncation level is set to be $v_n = \alpha \sqrt{BPV \phi_{k_n}(g)}(\Delta_n)^{0.49}$, with $BPV = \frac{\pi}{2} \sum_{i=2}^{n}\left|\Delta_{i-1}^{n} X \| \Delta_{i}^{n} X\right|$ and calculated on sparsely sampled data (5 min frequency). When $\sigma_Y = 0$, $\alpha = 5,4,4$ for 10-sec, 30-sec and 60-sec data, respectively; when $\sigma_Y = 0.5$, $\alpha = 5.5,3.5,3$ for 10-sec, 30-sec and 60-sec data, respectively; when $\sigma_Y = 1.5$, $\alpha = 5.2,3,2.5$ for the respective frequencies; and, finally, in the case of jump with infinite activity, we set $\alpha = 6,4,4$ for  10-sec, 30-sec and 60-sec data, respectively.}}
\label{table:finite_jump}
\end{center}
\end{table}

 

\section{{\Blue Conclusions}}\label{ConcludeSec}
{\Blue We introduce high-frequency-based kernel estimators of the spot volatility in both the absence and presence of microstructure noise. One of the key differences of our results from those of earlier literature is to consider a general kernel in an asymptotic regime for the bandwidth that leads to optimal convergence rates for the resulting kernel estimators. Under this regime, kernels of unbounded support offer improved performance compared to uniform or other kernels with bounded support. General two-sided kernels of unbounded support were already advocated in the work of \cite{FigLi}, where it was proved for the first time that exponential kernels are optimal, hence, formally validating an old conjecture of \cite{foster1994continuous_optKernel}. Unfortunately, \cite{FigLi} imposed strong assumptions for the validity of their results, the most important of which are the absence of leverage effects, microstructure noise, and jumps. These three effects are, of course, pervasive in real transaction data. In this work, we are able to relax all of those constraints and consider a rather general model. We further develop a feasible implementation of the proposed estimators. Via Monte Carlo experiments, we confirm the superior performance of the proposed estimators.}

\section*{Acknowledgements}
{\Blue The authors are grateful to the Associate Editor and two anonymous referees for their multiple suggestions that help to significantly improve the original manuscript.}

\appendix \label{appendix}

\section{Proof of Theorem \protect \ref{thm_no_noise}}\label{PrfOfMnRsltTh1}
We follow the steps in the proof of Theorem 13.3.3 in \cite{JacodProtter} (which implies Theorem 13.3.7). 
 By virtue of localization, without loss of generality, we assume throughout the proof that   {$|\delta(t, z)| \leq \Gamma(z)$ , $\mid \tilde{\delta}(\omega, t, z) | \wedge 1 \leq \Lambda(z) $   and  $$\Gamma(z)+\Lambda(z)+\int \Gamma(z)^{r} \lambda(d z)+ \int\Lambda(z)^{2} \lambda(d z)+\left|\mu_{t}\right|+\left|\sigma_{t}\right|+\left|X_{t}\right|  + \left|\rho_{t}\right| + \left|\tilde{\sigma}_{t}\right| + \left|\tilde{\mu}_{t}\right|\leq A,$$ (see Section 4.4.1 and (6.2.1) in \cite{JacodProtter} and Appendix A.5 in \cite{jacodaitsahalia} for details).} We use \(C\) to represent a generic constant that may change from line to line.

\subsection{Elimination of the jumps and the truncation}
We denote the process 
$\int_{0}^{t} \int_{\mathbb{R}^{d}} \delta(s, x)\mu(d s, d x)$ by   $  \delta \star\mu_{t}$ and we set 
\begin{align*}
&X^{\prime \prime}=\left\{\begin{array}{ll}
{\Blue  \delta \star \mathfrak{p},} & \text { if } r \leq 1, \\
{\Blue \delta\star(\mathfrak{p}-\mathfrak{q}),} & \text { if } r>1,
\end{array}\right. \quad X^{\prime}=X-X^{\prime \prime},\quad 
z_{n}=\left\{\begin{array}{ll}
\sqrt{m_{n}} & \text { if } \beta<\infty, \\
\frac{1}{\sqrt{m_{n} \Delta_{n}}} & \text { if } \beta=\infty.
\end{array}\right.
\end{align*}
To explicitly indicate the process $Y$ for which the spot estimator is calculated, we use the notation $\hat{c}\left({m}_{n},v_n,Y\right)_{\tau}$.
{\begin{lemma} \thlabel{truncated_diff}
\begin{enumerate}
\item When $X = X^{\prime}$ or when (\ref{eq:m_n_range}) and \(\varpi \leq \frac{1-a}{r}\) hold, we have, as $n\to\infty$, 
\begin{equation} \label{eq:cts_2_estimators}
\mathbb{P}\left(\hat{c}^{n}\left(m_{n}, v_{n}, X\right)_t \neq \hat{c}^{n}\left(m_{n}, X\right)_t\right) \rightarrow 0,
\end{equation} 
\item Under (\ref{eq:m_n_range}) and (\ref{eq:r_range_truncated}), we have
\begin{equation} \label{eq:diff_to0}
z_{n}\left(\hat{c}^{n}\left(m_{n}, v_{n}, X\right)_t-\hat{c}^{n}\left(m_{n}, v_{n}, X^{\prime}\right)_t\right) \stackrel{\mathbb{P}}{\longrightarrow} 0.
\end{equation}
\end{enumerate}
\end{lemma}}

\begin{proof}
We first need some estimates. 
\begin{itemize}
\item  Recall $|\delta(t, z)|\leq \Gamma(z) $ with  $\Gamma$ bounded and $\int \Gamma(z)^{r} \lambda(d z)<\infty$. Let $ J =\Gamma * \mathfrak{p}$ when $r \leq 1$ and $J=X^{\prime \prime}$ otherwise, so we have $\left|\Delta_{i}^{n} X^{\prime \prime}\right| \leq\left|\Delta_{i}^{n} J\right|$. Set   \begin{equation} \label{eq:notation_UVW}
U_{i}^{n}=\frac{\left|\Delta_{i}^{n} X^{\prime}\right|}{\sqrt{\Delta_{n}}}, \quad V_{i}^{n}=\frac{\left|\Delta_{i}^{n} J\right|}{\Delta_{n}^{\varpi}} \wedge 1, \quad W_{i}^{n}=\frac{\left|\Delta_{i}^{n} J\right|}{\sqrt{\Delta_{n}}} \wedge 1.
\end{equation} By (13.2.22)-(13.2.23) in \cite{JacodProtter}, for a suitable sequence \(\phi_n\)  going to 0, we have, for any $m>0$, 
\begin{equation*}
\begin{split}
\mathbb{E}\left(\left(U_{i}^{n}\right)^{m} | \mathcal{G}_{(i-1) \Delta_{n}}\right) &\leq C_{m},\\
\mathbb{E}\left(\left(V_{i}^{n}\right)^{m} | \mathcal{F}_{(i-1) \Delta_{n}}\right) &\leq \Delta_{n}^{(1-r \varpi)\left(1 \wedge \frac{m}{r}\right)} \phi_{n},\\
 \mathbb{E}\left(\left(W_{i}^{n}\right)^{m} | \mathcal{F}_{(i-1) \Delta_{n}}\right) &\leq \Delta_{n}^{\frac{2-r}{2}\left(1 \wedge \frac{m}{r}\right)} \phi_{n}.
\end{split}
\end{equation*}
\item { Recall \(\Gamma \) is bounded and $|\delta(t, z)| \leq \Gamma(z)$. By (2.1.44) in \cite{JacodProtter}, we have}
      \begin{equation} \label{eq:deltaXbound}
      \E\left(\left.\left(\Delta_j X\right)^2\right|\mathcal{F}_{({j}-1) \Delta_{n}} \right)\leq C \Delta_n.
      \end{equation}
\end{itemize} 
\begin{enumerate}
\item We define a kernel restricted on a bounded support \((-M,M]\)  for an integer \(M\) and the corresponding spot volatility estimator  {at} \(\tau \in [(i-1)\Delta_n, i\Delta_n)\):          
\begin{align*}
&K^{M}(x) = K(x) \mathbbm{1}_{|x|\leq M}\\
&\hat{c}^{n,M}\left(m_{n}, v_{n}, X\right)_{\tau} = \sum_{j=i - M m_n +1}^{i+M m_n} K_{{m}_{n}\Delta_n}\left(t_{j-1}-\tau\right)\left(\Delta_{j}^{n} X\right)^{2}\mathbbm{1}_{\left\{\left|\Delta_{j}^{n} X\right| \leq v_{n}\right\}}\\
&\hat{c}^{n,M}\left(m_{n}, X\right)_{\tau} = \sum_{j=i - M m_n +1}^{i+M m_n} K_{{m}_{n}\Delta_n}\left(t_{j-1}-\tau\right)\left(\Delta_{j}^{n} X\right)^{2}.
\end{align*}
{For all  \(\epsilon>0 \), we} have 
\begin{equation*}
      \begin{split}
      & \mathbb{P}\left(\left|\hat{c}^{n}\left(m_{n}, v_{n}, X\right)_{\tau} - \hat{c}^{n}\left(m_{n}, X\right)_{\tau}\right|> \epsilon \right) \\
      &\quad\leq\; \mathbb{P}\left(\left|\hat{c}^{n}\left(m_{n}, v_{n}, X\right)_{\tau} - \hat{c}^{n,M}\left(m_{n},v_n, X\right)_{\tau}\right|  > \epsilon/3\right) \\
       &\qquad +\mathbb{P}\left(\left|\hat{c}^{n,M}\left(m_{n}, v_{n}, X\right)_{\tau} - \hat{c}^{n,M}\left(m_{n}, X\right)_{\tau}\right|  > \epsilon/3\right)\\
       & \qquad+\mathbb{P}\left(\left|\hat{c}^{n,M}\left(m_{n},  X\right)_{\tau} - \hat{c}^{n}\left(m_{n}, X\right)_{\tau}\right|  > \epsilon/3\right).
      \end{split}
      \end{equation*}
By (\ref{eq:deltaXbound}),  
\begin{equation*}
\begin{split}
 &\mathbb{P}\left(\left|\hat{c}^{n,M}\left(m_{n},  X\right)_\tau - \hat{c}^{n}\left(m_{n}, X\right)_\tau\right|  > \epsilon/3\right)\\
 & \leq C\frac{1}{\epsilon} \E \left|\hat{c}^{n,M}\left(m_{n},  X\right)_\tau - \hat{c}^{n}\left(m_{n}, X\right)_\tau\right|\\
 & \leq C\frac{1}{\epsilon} \Delta_n \left(\sum_{j = 1}^{i - M m_n} + \sum_{j = i+M m_n +1}^{n}\right)\left| K_{m_n \Delta_n}(t_{j-1} - \tau)\right|\\
 & \leq C \frac{1}{\epsilon}  \left(\int_{M}^{\infty}\left| K(u)\right| d u +  \int_{-\infty}^{-M} \left|K(u) \right| d u\right), \text{ as } n \to \infty.
\end{split}
      \end{equation*} 
{The last integral converges to $0$ as $M\to\infty$.  Similarly, for all} \(\epsilon>0\), \[\lim_{M\to \infty} \lim_{n\to \infty} \mathbb{P}\left(\left|\hat{c}^{n,M}\left(m_{n},  v_n,X\right)_\tau - \hat{c}^{n}\left(m_{n},v_n, X\right)_\tau\right|  > \epsilon/3\right) =0.\]  Next, $\mathbb{P}\left(\left|\hat{c}^{n,M}\left(m_{n}, v_{n}, X\right)_\tau - \hat{c}^{n,M}\left(m_{n}, X\right)_\tau\right| > \epsilon/3 \right)\leq \sum_{j=i - M m_n +1}^{i+M m_n} a(n,j)$, where $$
\begin{aligned}
a(n, j) &=\mathbb{P}\left(|\Delta_{j}^{n} X|>v_{n}\right) 
\leq \mathbb{P}\left(|\Delta_{j}^{n} X^{\prime}|>v_{n} / 2\right)+\mathbb{P}\left(|\Delta_{j}^{n} X^{\prime \prime}|>v_{n} / 2\right).
\end{aligned}
$$By Markov inequality and (\ref{eq:notation_UVW}), 
\begin{equation*}
\begin{split}
\mathbb{P}\left(|\Delta_{j}^{n} X^{\prime}|>v_{n}/2\right)&\leq C_{m} \Delta_{n}^{m(1 / 2-\varpi)} \mathbb{E}\left(U_{j}^{n}\right)^{m}\leq  C_{m} \Delta_{n}^{m(1 / 2-\varpi)},\\
\mathbb{P}\left(|\Delta_{j}^{n} X^{\prime \prime}|>v_{n} / 2\right) &\leq C \mathbb{E}\left(V_{j}^{n}\right)^{r}  \leq \Delta_{n}^{1-r \varpi} \phi_{n},
\end{split}
\end{equation*}
where \(\phi_n \to 0\) as \(n \to \infty\) and \(m > 0\) is arbitrary. 							Since \(\varpi < \frac{1}{2}\), by taking \(m\) big enough, we then deduce \(a(n,j) \leq \Delta_n^2\) when \(X = X^{\prime}\) . When X is discontinuous, we have \(a(n,j) \leq C \Delta_n^{1-r\varpi} \phi_n\).   The same result holds if (\ref{eq:m_n_range}) and \(\varpi \leq \frac{1-a}{r} \). 												Finally, for all \(\epsilon>0,\)
\begin{equation*}
\lim_{n\to 0}\mathbb{P}\left(\left|\hat{c}^{n}\left(m_{n}, v_{n}, X\right)_\tau - \hat{c}^{n}\left(m_{n}, X\right)_\tau\right|> \epsilon \right) \rightarrow 0.
\end{equation*}
\item Next, we have
\begin{equation*}
\hat{c}^{n}\left(m_{n}, v_{n}, X\right)_\tau-\hat{c}^{n}\left(m_{n}, v_{n}, X^{\prime}\right)_\tau =  \Delta_n\sum_{j=1}^{n}K_{{m}_{n}\Delta_n}\left(t_{j-1}-\tau\right) \eta_{j,}^{n},
\end{equation*} 
where \(\eta_{j}^{n} = \frac{1}{\Delta_n}\left(\left(\Delta_{j}^{n} X\right)^{2}\mathbbm{1}_{\{|\Delta^n_j X| < v_n\}}-\left(\Delta_{j}^{n} X^{\prime}\right)^{2}\mathbbm{1}_{\{|\Delta^n_j X^{\prime}| < v_n\}}\right)\). By Lemma 13.2.6 in \cite{JacodProtter}, with arbitrary fixed \(\epsilon >0\) when r >1 and \(a=0\)  when \(r \leq 1\), there is a sequence \(\phi_n\)  going to 0 as \(n \to \infty\), such that
\begin{equation*}
\mathbb{E}\left({|\eta_{j}^{n}|}\mid \mathcal{F}_{(i-1) \Delta_{n}}\right) \leq\left(\Delta_{n}^{\frac{2-r}{2}\left(1 \wedge \frac{1 }{r}\right)-\epsilon}+\Delta_{n}^{(2-r)\varpi-\epsilon}\right) \phi_{n}.
\end{equation*}
Therefore, by {$\Delta_n\sum_{j=1}^{n}K_{{m}_{n}\Delta_n}\left(t_{j-1}-\tau\right)  \to  \int K(u) du$},
we have
\begin{equation*}
{\E\left|\hat{c}^{n}\left(m_{n}, v_{n}, X\right)_\tau-\hat{c}^{n}\left(m_{n}, v_{n}, X^{\prime}\right)_\tau\right|} \leq C \left(\Delta_{n}^{\frac{2-r}{2}\left(1 \wedge \frac{1 }{r}\right)-\epsilon}+\Delta_{n}^{(2-r)\varpi-\epsilon}\right) \phi_{n}.
\end{equation*} 
Under (\ref{eq:m_n_range}) we have (\ref{eq:diff_to0}) as soon as (\ref{eq:r_range_truncated}) {holds}.
\end{enumerate}
\end{proof}

Following the proof of {Lemma} 13.3.11 in \cite{JacodProtter}, {the following result follows from} \thref{truncated_diff}:
\begin{lemma} \thlabel{cts_deduce_discts}
Assume that (a) of \thref{thm_no_noise} holds for the non-truncated version (\ref{MDKNN}). Then (b) of the theorem also holds, as well as (a) and (c) for the truncated version.
\end{lemma}

\subsection{Proof in the continuous case}
With the previous lemma, it remains to prove the stable convergence (\ref{eq:stable_convergence}) under the following assumption:
\begin{assumption}\thlabel{X_continuous}
We have (\ref{eq:X}) with \(X\) continuous, $c_t = \sigma^2_t$ satisfies (\ref{eq:sigma}), the {processes} \(\mu,\tilde{\mu},\sigma,\tilde{\sigma}\) are bounded,  and $|\tilde{\delta}(\omega, t, z) | \wedge 1 \leq \Lambda(z) $ with a {bounded} function \(\Lambda\) on \(E\) {satisfying} $\int_{E} \Lambda(z)^{2} \lambda(d z)<{\infty}$.
\end{assumption}

Now we proceed our proof with the non-truncated estimator, {which for easiness of notation is denoted as ${\hat{c}\left({m}_{n}\right)_{\tau}} :=\sum_{i=1}^{n} K_{{m}_{n}\Delta_n}\left(t_{i-1}-\tau\right)\left(\Delta_{i}^{n} X\right)^{2}$.}
We {first} introduce some notation.
Recall that \(U^n_i := U_{i\Delta_n} \) and, for \(t \in ((i-1)\Delta_n, i\Delta_n ]\),													    let
\begin{equation} \label{eq:y^n}
\begin{split}
&V^{n}_t :=\sum_{j=1}^{n} K_{{m}_{n}\Delta_n}\left(t_{j-1} - t \right)\left( \left(\Delta_{j}^{n} W\right)^2-\Delta_n\right),\\
&V^{\prime n}_t :=\Delta_n \sum_{j=1}^{n} K_{{m}_{n}\Delta_n}\left(t_{j-1} - t \right)\left( B^n_{j} - B^n_{i}\right),\\
&Z^n_t := c^n_i V^n_t, \quad  Z^{\prime n}_t := \tilde{\sigma}^n_{i} V^{\prime n}_t,\\
&Z^{\prime \prime n}_t = \hat{c}\left({m}_{n}\right)_{t}-c_{t} - Z^n_t - Z^{\prime n}_t.
\end{split}
\end{equation}
All the three cases in \thref{thm_no_noise} {for the continuous case} follows from the next two lemmas:
\begin{lemma}\thlabel{st_no_noise}
Under \thref{kernel,X_continuous}, with $\Delta_{n}\to{}0$, ${m}_{n}\Delta_n\rightarrow 0$, and ${m}_{n} \sqrt{\Delta_n} \rightarrow \infty$, we have the following stable convergence in law:
\begin{equation*}
\left(\sqrt{{m}_{n}} Z^n_t, \frac{1}{\sqrt{{m}_{n}\Delta_n}} Z_t^{\prime n} \right) \stackrel{st}{\longrightarrow} \left(Z_{t}^{(0)}, Z_{t}^{\prime(0)} \right),
\end{equation*}
where \(Z^{0}_t\) and \(Z_{t}^{\prime(0)}  \) are defined in (\ref{eq:limiting_Z}).
\end{lemma}

\begin{lemma}\label{remainder}
Under \thref{kernel,X_continuous}, we have for all \( t \in [0,T]\), \begin{equation*}
z_n^{(0)} Z^{\prime \prime n}_t \stackrel{\mathbb{P}}{\longrightarrow} 0,
\end{equation*}
where \( z_n^{(0)} = \left\{\begin{array}{ll}{{m}_{n}^{1 / 2}}, & {\text { if } {m}_{n} \Delta_{n}^{1/2} \to \beta <\infty} \\ {\frac{1}{\sqrt{{m}_{n} \Delta_{n}}}}, & {\text { if } {m}_{n} \Delta_{n}^{1/2}\to \beta =\infty}\end{array}\right. .\)
\end{lemma}
We prove these two lemmas in the next two subsections.
\subsubsection{Proof of Theorem \protect \ref{st_no_noise}}
We first show \begin{equation} \label{eq:st_v}
\left(\sqrt{{m}_{n}}V^n_t, \frac{1}{\sqrt{{m}_{n}\Delta_n}}V^{\prime n}_t \right) \stackrel{st}{\longrightarrow}{\left( V, V^{\prime}\right)},
\end{equation}
where $(V,V')$ are defined in (\ref{eq:Y_Yprime}).
Denote the bandwidth of the kernel as \(b_{n}:= {m}_{n} \Delta_n\), recall \(t \in ((i-1)\Delta_n, i\Delta_n ], \) we can write the pair \(\left(\sqrt{m_n}V^n_t, \frac{1}{\sqrt{{m}_{n} \Delta_{n}}}V^{\prime n}_t\right) \) as \(\sum_{j = 1}^n\left(\zeta_{j}^{n}(t), \zeta_{j}^{\prime n}(t)\right)\),  where 
\begin{equation*}
\begin{split}
\zeta^n_j(t) &=\sqrt{{m}_{n}}K_{b_n}\left( t_{j-1} - t\right)\left( \left(\Delta_{j}^{n} W\right)^2-\Delta_n\right),\\
 \zeta^{\prime n}_j(t) & = \frac{\Delta_n}{\sqrt{{m}_{n}\Delta_n}}\left\{\begin{array}{ll}{0} & {\text { if } j=1 } \\{-\left(\sum_{l=1}^{j-1}K_{b_n}(t_{l-1}-t)\right) \Delta^n_j B} & {\text { if } 2\leq j \leq i } \\ {\left(\sum_{l=j}^{n}K_{b_n}(t_{l-1}-t)\right) \Delta^n_j B} & {\text { if } i < j\leq n}\end{array}\right. . 
\end{split}
\end{equation*}
Then we notice that \(\left(\zeta_{j}^{n}(t), \zeta_{j}^{\prime n}(t)\right)\) is \(\mathcal{F^{(0)}}_{t_j}\) measurable and {with \(\mathcal{F}_j:= \mathcal{F}^{(0)}_{t_j}\),} 
\begin{equation*}
\begin{split}
\sum_{j=1}^n \E\left(\left.\zeta^n_j(t)\right| \mathcal{F}_{j-1}^{(0)} \right) &= 0,\\
\sum_{j=1}^n \E\left(\left.\zeta^{\prime n}_j(t)\right| \mathcal{F}_{j-1}^{(0)} \right) &= 0.
\end{split}
\end{equation*}
Recall 	that \(\rho_s = d\left<W,B\right>_{s}/ds\) is c\`ad\`ag and bounded on the interval \([t_{j-1}, t_j] \).  By It\^o lemma, Cauchy-Schwartz inequality, and Doob's inequality, we have 
\begin{equation*}
\begin{split}
\left|\E\left( \left.\left(\Delta_{j}^{n} W\right)^2 \Delta_j^n B\right|\mathcal{F}_{j-1}^{(0)}\right) \right| &  = \left|\int_{t_{j-1}}^{t_{j}}2 \E\left( \left. \rho_s\left(\int_{t_{j-1}}^{s} dW_u\right)\right|\mathcal{F}_{j-1}^{(0)}\right) d s \right|\\
&  = \left|\int_{t_{j-1}}^{t_{j}}2 \E\left( \left. \left(\rho_s - \rho_{t_{j-1}}\right)\left(\int_{t_{j-1}}^{s} dW_u\right)\right|\mathcal{F}_{j-1}^{(0)}\right) d s \right|\\
& \leq \int_{t_{j-1}}^{t_{j}}2 \sqrt{\E\left( \left. \left(\rho_s - \rho_{t_{j-1}}\right)^2\right|\mathcal{F}_{j-1}^{(0)}\right)   \Delta_n} d s\\
& \leq C \Delta_n^{3/2} \sqrt{\E\left( \left. \left(\rho_{t_j} - \rho_{t_{j-1}}\right)^2\right|\mathcal{F}_{j-1}^{(0)}\right)  }.
\end{split}
\end{equation*}  Then, by a change of variable,   \begin{equation*}
\begin{split}
&\left|\sum_{j=1}^n \E\left(\left.\zeta^n_j(t) \zeta^{\prime n}_j(t)\right| \mathcal{F}_{j-1} \right)\right|\\
&\quad   \leq \sqrt{\Delta_n}\sum_{j = 2}^{i}  \left|K_{b_n}\left(t_{j-1} - t \right) \right|\left( \sum_{l=1}^{j-1} \left| K_{b_n} (t_{l-1} - t)\right|\right) \left|\E\left( \left.\left(\Delta_{j}^{n} W\right)^2 \Delta_j^n B\right|\mathcal{F}_{j-1}^{(0)}\right) \right|\\
& \qquad + \sqrt{\Delta_n}\sum_{j=i+1}^n\left|K_{b_n}\left(t_{j-1} - t \right) \right|\left( \sum_{l=j}^{n} \left| K_{b_n} (t_{l-1} - t)\right|\right)\left|\E\left( \left.\left(\Delta_{j}^{n} W\right)^2 \Delta_j^n B\right|\mathcal{F}_{j-1}^{(0)}\right)\right|\\
&\quad \leq  C\Delta_n^{2}\sum_{j = 2}^{i} \left|K_{b_n}\left(t_{j-1} - t \right) \right|\left( \sum_{l=1}^{j-1} \left| K_{b_n} (t_{l-1} - t)\right|\right) \max_{j}\sqrt{\E\left( \left. \left(\rho_{t_j} - \rho_{t_{j-1}}\right)^2\right|\mathcal{F}_{j-1}^{(0)}\right)  }  \\
& \qquad + C\Delta_n^{2}\sum_{j=i+1}^n\left|K_{b_n}\left(t_{j-1} - t \right) \right|\left( \sum_{l=j}^{n} \left| K_{b_n} (t_{l-1} - t)\right|\right)\max_{j}\sqrt{\E\left( \left. \left(\rho_{t_j} - \rho_{t_{j-1}}\right)^2\right|\mathcal{F}_{j-1}^{(0)}\right)  } \\
&\quad  \leq C\int \left|K(u)\right| \left|L(u)\right| du  \max_{j}\sqrt{\E\left( \left. \left(\rho_{t_j} - \rho_{t_{j-1}}\right)^2\right|\mathcal{F}_{j-1}^{(0)}\right)  }  .
\end{split}
\end{equation*}
We notice that \(\rho \) is right-continuous and uniformly bounded on [0,T], thus, we have
\[
	\max_{j}\E\left( \left. \left(\rho_{t_j} - \rho_{t_{j-1}}\right)^2\right|\mathcal{F}_{j-1}^{(0)}\right)  \rightarrow 0.
\]
 Therefore, \[\sum_{j=1}^n \E\left(\left.\zeta^n_j(t) \zeta^{\prime n}_j(t)\right| \mathcal{F}_{j-1} \right)  \rightarrow 0, \text{ as } n \rightarrow \infty.\]
 Next, we can deduce the following by the Riemann sum theorem and change of variables:
\begin{equation*}\begin{split}
\sum_{j=1}^n \E\left( \left. \zeta_j^n(t)^2 \right| \mathcal{F}_{j-1} \right)
&= 2 \sum_{j=1}^n {m}_{n} \Delta_n^2K_{b_n}^2(t_{j-1} - t)\\
&\longrightarrow 2 \int K^2(u)du,\\
\sum_{j=1}^n \E\left( \left. \zeta_j^{\prime n}(t) ^2 \right| \mathcal{F}_{j-1} \right) &=  \frac{\Delta_n^2}{{m}_{n}} \left(\sum_{j=i+1}^n \left(\sum_{m=j}^n K_{b_n}\left(t_{m-1} - t\right) \right)^2 +  \sum_{j=2}^{i} \left(\sum_{m=1}^{j-1} K_{b_n}\left(t_{m-1} - t\right) \right)^2\right)\\
&{\sim   \frac{1}{{m}_{n} \Delta_n}}\left(\int_t^T \left(\int_v^T K_{b_n}(s-t)ds  \right)^2 d v  +  \int_0^t \left(\int_0^v K_{b_n}(s-t)d s  \right)^2 d v \right) \\
&\longrightarrow \int L^2(u) du,
\end{split}
\end{equation*}
where \(L(t)=\int_{t}^{\infty} K(u) d u \mathbbm{1}_{\{t>0\}}-\int_{-\infty}^{t} K(u) d u \mathbbm{1}_{\{t \leq 0\}}\).
Note that:
\begin{equation*}
\begin{split}
&\sum_{j = 1}^n\{\E\left(\left.\zeta^n_j(t)^4\right|\mathcal{F}_{i-1}\right) +\E\left(\left.\zeta^{\prime n}_j(t)^4\right|\mathcal{F}_{i-1}\right)\} \\
&=\sum_{j = 1}^n {m}_{n}^2 \Delta_n^4 K_{b_n}^4(t_{j-1}-t) \E \left(U_j^2-1 \right)^4  \\
& \qquad + \frac{\Delta_n^4}{{m}_{n}^2} \left(\sum_{j=i+1}^n \left(\sum_{m=j}^n K_{b_n}(t_{m-1}-t) \right)^4   + \sum_{{j=2}}^{{i}} \left(\sum_{m=1}^{{j-1}} K_{b_n}\left(t_{m-1} - t\right) \right)^2\right)\\
& \leq\frac{C}{{m}_{n}}\int K^4(u)d u + \frac{C}{{m}_{n}^2\Delta_n} \int L(u)^4du  \longrightarrow 0,
\end{split}
\end{equation*}
where \(U_j\) is a standard normal distribution and C is a generic constant.
To apply Theorem 2.2.15 in \cite{JacodProtter}, we further need to show that 
\begin{equation}\label{eq:mult_mart}
{\rm (i)} \;\sum_{j=1}^n \E\left(\left. \zeta^n_j(t)\left(M_{{t_j}} - M_{{t_{j-1}}} \right)\right|\mathcal{F}_{{j-1}} \right) \rightarrow 0, \quad 
{\rm (ii)}\;\sum_{j=1}^n \E\left(\left. \zeta^{\prime n}_j(t)\left(M_{{t_j}} - M_{{t_{j-1}}} \right)\right|\mathcal{F}_{{j-1}} \right)\rightarrow 0,
\end{equation}whenever \(M \) is either one of the component of \(\left(W,B \right)\) or is in the set  $\mathcal{N}$ containing all bounded $\left(\mathcal{F}_{t}^{(0)}\right)$-martingales orthogonal (in the martingale sense) to \(\left(W,B \right)\). 
When \(M = W \) or \(B\), (\ref{eq:mult_mart}-i) holds true since it is the  $\mathcal{F}_{(j-1)\Delta_n}$-conditional expectation of an odd function of the increments of the process \(W\) after time \((j-1) \Delta_{n}\). {On the other hand, by the boundedness of the process $\rho$, we have \( |\E\left(\left.\Delta_j B\Delta_j W \right|\mathcal{F}_{j-1} \right)|= \E\left(\left.\left|\int_{t_{j-1}}^{t_j} \rho_s d s \right|\right| \mathcal{F_{j-1}}\right)\leq C \Delta_{n}\), for some constant $C$ and, thus, (\ref{eq:mult_mart}-ii) can be shown as follows:}
\begin{equation*}
\begin{split}
\sum_{j=1}^n \E\left(\left. \zeta^{\prime n}_j(t)\left(M_{t_j} - M_{t_{j-1}} \right)\right|\mathcal{F}_{j-1} \right) & \leq  \frac{\Delta_n^{3/2}}{\sqrt{{m}_{n}}} \left(\sum_{j=i+1}^n \left|\sum_{m=j}^n K_{b_n}\left(t_{m-1} - t\right) \right| +  \sum_{j=2}^{i} \left|\sum_{m=1}^{j-1} K_{b_n}\left(t_{m-1} - t\right) \right|\right)\\
&\leq  C \frac{1}{m_n \Delta_n} \int  \left|L(u)\right| du \rightarrow 0.
\end{split}
\end{equation*}Suppose now that \(N\) is a bounded martingale, orthogonal to \(\left(W, B\right)\). By It\^{o}'s formula we see that \(\zeta^n_j(t)\) can be written as \(\sqrt{{m}_{n}}K_{b_n}\left( t_{j-1} - t\right)\int_{t_{j-1}}^{t_j} 2 \left(W_s - W_{t_{j-1}} \right) d W_s \),  i.e., a stochastic integral with respect to \(W\) on the interval \(\left[(j-1) \Delta_{n},j \Delta_{n}\right]\). Similarly,  \( \zeta^{\prime n}_j(t)\) is a stochastic integral with respect to \(B\) on the same interval. Then the orthogonality of \(N\) and \(\left(W,B\right)\) implies (\ref{eq:mult_mart}). 
Now, we can apply Theorem 2.2.15 in \cite{JacodProtter} and show that 
\begin{equation*}
\left(\sqrt{{m}_{n}}V^n_t, \frac{1}{\sqrt{{m}_{n}\Delta_n}}V^{\prime n}_t \right) \stackrel{st}{\longrightarrow} \left( V, V^{\prime }\right),
\end{equation*}where \(V, V^{\prime} \) is defined in (\ref{eq:Y_Yprime}).
Finally, recall that \[Z^n_t := c^n_i  V^n_t, \quad  Z^{\prime n}_t := \tilde{\sigma}^n_i  V^{\prime n}_t.\]
From the c\`adl\`ag property of \(\sigma\) and {$\tilde{\sigma}$}, we see that \(c^n_i \rightarrow c_t\) {and \(\tilde{\sigma}^n_i \rightarrow \tilde{\sigma_t}\)}, for \(t \in ((i-1)\Delta_n, i\Delta_n ] \). Then \thref{st_no_noise}  follows from (\ref{eq:st_v}) and the following property of the stable in law convergence: \[Z_{n} \stackrel{st}{\longrightarrow} Z, \quad Y_{n} \stackrel{\mathbb{P}}{\longrightarrow} Y \quad \Rightarrow \quad\left(Y_{n}, Z_{n}\right) \stackrel{st}{\longrightarrow}(Y, Z).\]

\subsubsection{Proof of Lemma \protect \ref{remainder}}
 For \(t \in\left((i-1) \Delta_{n}, i \Delta_{n}\right]\), we can rewrite \( Z^{\prime \prime n}_t  \) defined in (\ref{eq:y^n}) as follows:
\begin{equation*}
Z^{\prime \prime n}_t=\sum_{j=1}^{5} \zeta^{n}_j(t),
\end{equation*}
where{
\begin{equation*}
\begin{split}
\zeta^n_1(t) &= c^n_i \Delta_n \sum_{j=1}^n K_{b_n}(t_{j-1}-t)  -c_t\\
\zeta^n_2(t) & = \sum_{j=1}^n K_{b_n}(t_{j-1}-t) \left(\left(\Delta^n_j X\right)^2 - c^n_{j-1}\left(\Delta^n_j W\right)^2  \right)\\
\zeta^n_3(t) &=\sum_{j=1}^n K_{b_n}(t_{j-1}-t) \tilde{\sigma}^n_i\left( \left(\Delta^n_j W\right)^2 - \Delta_n \right)  \left(B^n_{j} - B^n_i  \right)\\
\zeta^n_4(t) &= \sum_{j=1}^n K_{b_n}(t_{j-1}-t) \left( c^n_{j-1} - c^n_i -\tilde{\sigma}^n_i  \left(B^n_{j}- B^n_i \right)  \right)  \left(\Delta_j^n W\right)^2
\end{split}
\end{equation*}}
Therefore, it is enough to prove that, for \(l = 1, 2,3,4\) and all \(t \in [0,T]\), we have \begin{equation} \label{eq:zeta_go_to_0}
z_n^{(0)} \zeta^n_l(t) \stackrel{\mathbb{P}}{\rightarrow} 0.
\end{equation}
\begin{proof}[Proof of (\ref{eq:zeta_go_to_0}) for $l=1$]
 By Lemma 3.1 in \cite{FigLiSumplement1} with f=1 and \thref{kernel} we have 
 \begin{equation*}
\Delta_n \sum_{j=1}^n K_{b_n}(t_{j-1}-t)  - \int_{0}^T K_{b_n}(s-t) d s = \frac{1}{2}\left(K(A^+) - K(B^-) \right)\frac{\Delta_n}{b} + \mathrm{o}\left(\frac{\Delta_n}{b} \right) = \mathrm{O}\left(\frac{\Delta_n}{b}  \right),
\end{equation*} 
where \(\left( A,B\right) \) is the support of \(K\) and \(-\infty \leq A < 0 < B \leq  \infty \). Therefore, the boundedness of \( c\) implies
\begin{equation*}
\begin{split}
\zeta^n_1(t) &= c^n_i \left(\int_0^T K_{b_n}(s-t)ds\right)    -c_t + \mathrm{O}\left(\frac{\Delta_n}{b} \right) = c^n_i  -c_t + C\int_{(0, T)^{c}} K_{b_n}(t-\tau) d t  +  \mathrm{O}\left(\frac{\Delta_n}{b} \right).
\end{split}
\end{equation*}
Also, we can deduce the following from (\ref{eq:sigma}): 
\begin{equation*}
 \E \left( c_i-c_t\right)^2 \leq C \Delta_n, \text{ for } t \in \left((i-1) \Delta_{n}, i \Delta_{n}\right].
 \end{equation*}
\thref{kernel} implies  {that} \(x^{1 / 2} \int_{x}^{\infty} K(u) d u \rightarrow 0,\) as \(x \rightarrow \infty\). We then have
 \begin{equation*}
b_n^{-1 / 2} \int_{(0, T)^{c}} K_{b_n}(t-\tau) d t=\frac{1}{\sqrt{b_n}}\left(\int_{-\infty}^{\frac{\tau}{b_n}} K(u) d u+\int_{\frac{T-\tau}{b_n}}^{\infty} K(u) d u\right) \rightarrow 0, \text { as } n \rightarrow \infty.
\end{equation*}
Thus, \( z_n^{(0)} \zeta^n_1(t) \rightarrow 0\) since \(z_n \sqrt{\Delta_n} \rightarrow 0, z_n^{(0)} \frac{\Delta_n}{b_n}\rightarrow 0 \) and \(z_n^{(0)} = \left\{\begin{array}{ll}{\beta b_n^{-1/2}} & {\text { if } {m}_{n} \Delta_{n}^{1 / 2} \rightarrow \beta<\infty} \\ {b_n^{-1/2}} & {\text { if } {m}_{n} \Delta_{n}^{1 / 2} \rightarrow \beta=\infty}\end{array}\right.\).

\end{proof}

\begin{proof}[Proof of (\ref{eq:zeta_go_to_0}) for $l=2$]
Let \( \rho^n_j(t) = \Delta_j^n X - \sigma^n_{j-1}\Delta^n_j W \). In view of  (2.1.44) in \cite{JacodProtter}, for \(q \geq 2\), we have:
\begin{equation*}
\mathbb{E}\left(\left|\rho_{j}^{n}( t)\right|^{q}\right) \leq K_{q} \Delta_{n}^{1+q / 2}, \quad \E\left(\left|\sigma_{j-1}\Delta^n_j W\right|^q \right)  \leq C \Delta_n^{q/2}.
\end{equation*}
Then, since \( \left|\left(\Delta^n_j X\right)^2 - \sigma_{j-1}^2\left(\Delta^n_j W\right)^2  \right|\leq  2 \left(\left|\rho^n_j(t) \right|^2 + \left|\rho^n_j(t) \right|\left| \sigma_{j-1}^2\Delta^n_j W\right| \right)\),  the inequalities above and the Cauchy-Schwartz inequality yield
\begin{equation*}
\begin{split}
\E \left|\zeta_2^n(t)\right| & \leq 2 \sum_{j=1}^n \left| K_{b_n}(t_{j-1}-t)\right| \E \left(  \left|\rho^n_j(t) \right|^2 + \sqrt{\E\left|\rho^n_j(t) \right|^2 \E\left| \sigma_{j-1}^2\Delta^n_j W\right|^2}\right)\\
& \leq C\sum_{j=1}^n \left| K_{b_n}(t_{j-1}-t) \right| \left(\Delta_n^2 + \Delta_n^{3/2} \right) \sim\int K(u) du \sqrt{\Delta_n}.
\end{split}
\end{equation*}
We then have the result since \(z_n \sqrt{\Delta_n}\rightarrow 0\).

\end{proof}

\begin{proof}[Proof of (\ref{eq:zeta_go_to_0}) for $l=3$]
\(\zeta_3^n(t) \) can be written as \(\tilde{\sigma}^n_i \Phi^n(t) \)  where each $\tilde{\sigma}^n_i$ is bounded $\mathcal{F}^{(0)}_i$ measurable and 
\begin{equation*}
\Phi^n(t) = \sum_{j=1}^n K_{b_n}(t_{j-1}-t) \left( \left(\Delta^n_j W\right)^2 - \Delta_n \right)  \left(B^n_{j-1} - B^n_i  \right).
\end{equation*}
We can compute that \(\E \left( \Phi^n(t) \right) =0\) and  \(\left| \E\left(\Delta_j^n W \Delta_j^n B\right) \right| = \left|\E\left(  \int_{t_{j-1}}^{t_j} \rho_s d s  \right) \right|\leq C \Delta_n\). Notice that  \(\left(\Delta^n_j W\right)^2 - \Delta_n , B^n_{j-1} - B^n_i  \) are independent when \(j \geq i+1\) and \(\left(\Delta^n_j W\right)^2 - \Delta_n, B^n_{i} - B^n_j  \) are independent when \( j \leq i.\) Then, by tower property property,   we have  
\begin{equation*}
\begin{split}
&\E \left(\Phi^n(t)\right)^2 \\
& = \sum_{j=i+1}^n K^2_{b_n}(t_{j-1}-t)  2\Delta_n^2 (j-1-i)\Delta_n +\sum_{j=1}^{i} K^2_{b_n}(t_{j-1}-t) \E \left(\left(\left(\Delta_j^n W\right)^2 - \Delta_n \right)^2 \left(B^n_i - B^n_{j}+ \Delta_j^n B \right)^2 \right)\\
& \leq 2\Delta_n^2 \sum_{j=i+1}^n K^2_{b_n}(t_{j-1}-t)(t_{j-1}-t_i) + \sum_{j=1}^{i} K^2_{b_n}(t_{j-1}-t)\left( 2\Delta_n^2  (t_{i}-t_j) + \sqrt{\E\left(\left(\Delta_j^n W\right)^2 - \Delta_n \right)^4 \E \left(\Delta_j^n B \right)^4}\right)  \\
&  \leq 2\Delta_n^2 \sum_{j=i+1}^n K^2_{b_n}(t_{j-1}-t)(t_{j-1}-t_i) + \sum_{j=1}^{i} K^2_{b_n}(t_{j-1}-t)\left( 2\Delta_n^2  (t_{i}-t_j) + C_1 \Delta_n^3\right) \\
& \sim\Delta_n \left( \int_{t}^T K_{b_n}^2(s-t) (s-t)ds  - \int_{0}^t K_{b_n}^2(s-t) (s-t)ds\right) \\
& \sim\Delta_n \left(\int_0^{\infty} K^2(u) u d u - \int_{-\infty}^0  K^2(u) u d u\right),
\end{split}
\end{equation*}
where \(C_1 = \sqrt{\E\left((\chi_1^2 -1)^4\right) \E\left( \left(\chi_1^2\right)^2\right)}\). Then \( \frac{1}{\sqrt{\Delta_n}} \Phi^n(t) \) is bounded in probability, and the result follows, since \(z_n \sqrt{\Delta_n}\rightarrow 0\).
\end{proof}

\begin{proof}[Proof of (\ref{eq:zeta_go_to_0}) for $l=4$]
{Let \( \eta^n_j= \left( c_{j-1}^{n} -c_{i}^{n}- \tilde{\sigma}_{i}^{n}\left(B^n_{j-1}-B^n_{i}\right)\right) = \int_{t_i}^{t_{j-1}} \tilde{\mu}_s ds + \int_{t_i}^{t_{j-1}}\left(\tilde{\sigma_s} - \tilde{\sigma}^n_i \right)dB_s + M_{t_{j-1}} - M_{i} ,\)  where $M=\widetilde{\delta} \star(\mathfrak{p}-\mathfrak{q})$.  
Following the same argument for proof of (13.3.37) for (j = 6) in \cite{JacodProtter}, on the set $\Omega(n, N, \varepsilon)=\left\{\left|\Delta M_{s}\right| \leq \varepsilon,  \forall s \in\left(i-N m_n\Delta_n, i+N m_n\Delta_{n}\right]\right\}$, with the notation $
\gamma^{n}_j=\frac{1}{t_{j-1} -t_i  } \mathbb{E}\left(\int_{t_i}^{t_{j-1}}\left|\tilde{\sigma}_{s}-\tilde{\sigma}_{t_i}\right|^{2} d s\right),
$ we deduce that for \(j \in \left(i-N m_n\Delta_n, i+N m_n\Delta_{n}\right]   \),
\begin{equation*}
\mathbb{E}\left(\left(\eta^n_j \right)^21_{\Omega(n, i, \varepsilon)}\right) \leq C (t_{j-1} -t_i ) \rho(n, j,\varepsilon), \quad \text{ with } \rho(n, \varepsilon)=\frac{t_{j-1} -t_i  }{\varepsilon}+\gamma^{n}_j+\phi(\varepsilon),
\end{equation*} 
where $\phi(\epsilon)  = \int_{\{\Lambda(z)<\varepsilon\}} \Lambda(z)^{2} \lambda(d z)$ going to 0 as \(\epsilon \to 0.\) Since \( \tilde{\sigma}\) is c\`adl\`ag and bounded, we see that \(\gamma^n_j \rightarrow 0  \) for all \(j\) and thus \(\rho(n,j,\epsilon) \to 0.\)
From (2.1.44) in \cite{JacodProtter}, we also have for {all $j$, $ \mathbb{E}\left(\left(\eta^n_j \right)^2\right) \leq C (t_{j-1} - t_i)$ and, thus,} 
\begin{equation*}
\begin{split}
&\E\left|\zeta_4^n(t)\right|  =\E\left|\zeta_4^n(t)\right|\mathbbm{1}_{\Omega(n, N, \varepsilon)} + \E\left|\zeta_4^n(t)\right|\mathbbm{1}_{\Omega(n, N, \varepsilon)^{c}} \\
& \leq C \sum_{j=1}^{n} \left|K_{b_n}(t_{j-1}-t)\right|\sqrt{\E\left(\eta_{j}^n\right)^2 \E \left(\Delta_j^n W\right)^4}\\
&\quad+C\left( \sum_{j=1}^{i-N m_n}+ \sum_{j=i+N m_n+1}^{n}\right)\left|K_{b_n}(t_{j-1}-t)\right|\Delta_n\sqrt{|t_{j-1} -t_i |}  \\
& \leq C N\sqrt{N m_n \Delta_n \rho(n,\epsilon)}\\
&\quad+C\sqrt{m_n \Delta_n}\left(\int_N^{\infty} K(u) \sqrt{u} d u + \int_{-\infty}^{-N} K(u) \sqrt{u} d u\right).
\end{split}
\end{equation*}Additionally, we have $\lim _{\varepsilon \rightarrow 0} \limsup _{n} \rho(n, \varepsilon)=0$ and \[\lim_{N \to \infty}\left(\int_N^{\infty} K(u) \sqrt{u} d u + \int_{-\infty}^{-N} K(u) \sqrt{u} d u\right) = 0.\]
The result follows by \(z_n\sqrt{b_n} < \infty\).}
\end{proof}

\section{Proof of Theorem \protect \ref{thm}}\label{PrfOfMnRslt}

 {Again, by virtue of localization, without loss of generality, we assume throughout the proof that   $|\delta(t, z)| \leq \Gamma(z)$ , $\mid \tilde{\delta}(\omega, t, z) | \wedge 1 \leq \Lambda(z) $   and  $$\Gamma(z)+\Lambda(z)+\int \Gamma(z)^{r} \lambda(d z)+ \int\Lambda(z)^{2} \lambda(d z)+\left|\mu_{t}\right|+\left|\sigma_{t}\right|+\left|X_{t}\right|  + \left|\rho_{t}\right| + \left|\tilde{\sigma}_{t}\right| + \left|\tilde{\mu}_{t}\right|\leq A$$ (see Section 4.4.1 and (6.2.1) in \cite{JacodProtter} and Appendix A.5 in \cite{jacodaitsahalia} for details). }

\subsection{Elimination of the jumps and the truncation}
We set 
\begin{equation*}
X^{\prime \prime}=\left\{\begin{array}{ll}
\delta \star  \mathfrak{p}, & \text { if } r \leq 1, \\
\delta \star (\mathfrak{p}-\mathfrak{q}), & \text { if } r>1,
\end{array}\right.\quad X^{\prime}=X-X^{\prime \prime},\quad
{\tilde{z}_{n}}=\left\{\begin{array}{ll}
{m}_{n}^{1 / 2} \Delta_{n}^{1 / 4},& \text { if } \beta<\infty, \\
\frac{1}{\sqrt{m_{n} \Delta_{n}}}, & \text { if } \beta=\infty.
\end{array}\right.
\end{equation*}
Let $Y^{*}=Y-X+X^{\prime}$ be the continuous process
with microstructure noise and set
\begin{equation} \label{eq:non-truncated_noise}
\hat{c}^{*}\left(k_{n}, {m}_{n}\right)_{\tau} =\frac{1}{\phi_{k_{n}}\left(g\right)} \sum_{j=1}^{n-k_n+1} K_{{m}_{n} \Delta_n}\left(t_{j-1} - \tau\right)\left(\left(\overline{Y}_{j}^{*n}\right)^2   - \frac{1}{2}\widehat{Y}_{j}^{*n}\right) .
\end{equation}
We need some preliminary estimates:\begin{itemize}
\item By Corollary 2.1.9 (a)-(c) in \cite{JacodProtter}, for \(p > 0\) and $q \in [0,1 / r)$,
\begin{equation} 
E\left[\sup _{u \in[0, s ]}\left(\frac{\left|X^{\prime \prime}_{\tau+u}-X^{\prime \prime}_{\tau}\right|}{s^{q}} \wedge 1\right)^{p} \Big| \mathcal{F}_{\tau}^{(0)}\right] \leq Cs^{(1-q r)(p / r \wedge 1)} a(s),
\end{equation} where $a(s) \to 0$ as $s \to 0$.
Let $g_n(t) = \sum_{j =1}^{k_n} g(j/k_n) \mathbbm{ l }_{\left((j-1) \Delta_{n}, j \Delta_{n}\right]}(t) $. 
With $u_n = k_n \Delta_n $ and $\bar{X}^{\prime \prime n}_{i}=\int_{(i-1) \Delta_{n}}^{(i-1) \Delta_{n}+u_{n}} g_{n}(s-(i- 1)\Delta_n) d X_{s}^{\prime \prime}$, {\Green same as (9.2.13) in \cite{JacodProtter},} we deduce that \begin{equation} \label{eq:J_estimates}
    E\left[\left(\frac{\left|\bar{X}^{\prime \prime n}_i\right|}{u_n^{q}} \wedge 1\right)^{p} \Big| \mathcal{F}_{(i-1)\Delta_n}^{(0)}\right] \leq Cu_n^{(1-q r)(p / r \wedge 1)} a_n,
\end{equation} where $a_n \to 0$ as $n \to \infty$.

\item 
By the proof of Lemma 16.4.3 in \cite{JacodProtter}, for all $q>0$,\begin{equation} \label{eq:X_bar_hat_bound}
\begin{split}
&\mathbb{E}\left(\left|\bar{X}^{\prime n}_{i}\right|^{q} \mid \mathcal{F}_{(i-1) \Delta_{n}}\right) \leq C_{q} \Delta_{n}^{q / 4},\\
&\mathbb{E}\left(\left|{\hat{X}_{i}^{ \prime n}}\right|^{q}\Big| \mathcal{F}_{(i-1) \Delta_{n}}\right) 
\leq C_{q}{\Delta_n^{3q/2}},\\
&\mathbb{E}\left(\left|{\hat{X}_{i}^{ n}}\right|^{q}\Big| \mathcal{F}_{(i-1) \Delta_{n}}\right) 
\leq C_{q}\Delta_n^{q/2}\left(\Delta_{n}^{q }+ \Delta_{n}^{q \wedge 1}\right),\\
&\mathbb{E}\left(\left|{\bar{Y}_{i}^{n}}\right|^{q}\Big| \mathcal{F}_{(i-1) \Delta_{n}}\right) 
\leq C_{q} \Delta_n^{q/4} \left(1 + \Delta_n^{-(q/4-1/2)+}\right), \\
&\mathbb{E}\left(\left|{\hat{Y}_{i}^{n}}\right|^{q}\Big| \mathcal{F}_{(i-1) \Delta_{n}}\right) 
\leq C_{q} \Delta_n^{q/2} \left(\Delta_n^q + 1 + \Delta_n^{q\wedge 1}\right).
\end{split}
\end{equation} 
\item From (16.2.3) {in} \cite{JacodProtter}, for all \(p > 0\),  \begin{equation} \label{eq:epsilon_bound}
\begin{split}
\mathbb{E}\left(\left|\bar{\epsilon}(g)_{i}^{n}\right|^{p} \mid \mathcal{H}_{(i-1) \Delta_{n}}\right) \leq C_{p} \Delta_{n}^{p / 4}, \quad 
\mathbb{E}\left(\left|\widehat{\epsilon}_{i}^{n}\right|^{p} \mid \mathcal{H}_{(i-1) \Delta_{n}}\right) \leq C_{p} \Delta_{n}^{p/2}.
\end{split}
\end{equation}
\item {Combining}  (\ref{eq:X_bar_hat_bound}) and  (\ref{eq:epsilon_bound}), and {applying} \((a+b)^q \leq K_q (a^q + b^q)\),   for all $q>0 $,
\begin{equation} \label{eq:Y_star_bound}
\begin{split}
&\mathbb{E}\left(\left|\bar{Y}_{i}^{* n}\right|^{q} \mid \mathcal{F}_{(i-1) \Delta_{n}}\right) \leq C_{q}{\mathbb{E}\left(|\bar{X^{\prime}}(g)_{i}^{n}|^{q} + \left|\bar{\epsilon}(g)_{i}^{n}\right|^{q}\mid \mathcal{F}_{(i-1) \Delta_{n}}\right)}  \leq  C_{q}\Delta_{n}^{q / 4},\\
&\mathbb{E}\left(\left|\hat{Y}^{* n}_i\right|^{q} \mid \mathcal{F}_{(i-1) \Delta_{n}}\right)        \leq C_q \mathbb{E}\left(\left|\hat{X}^{\prime n}_i\right|^q + \left|\hat{\epsilon}^{n}_i\right|^q    \mid \mathcal{F}_{(i-1) \Delta_{n}}\right)  \leq C_q \Delta_n^{q/2}.
\end{split}
\end{equation}
\end{itemize}
The following result will allow us to reduce the proof  to the case of a continuous process and a kernel estimator without truncation.
\begin{lemma} \thlabel{eliminate_jump_noisy}
Under \thref{X},  we have 
\begin{equation} \label{eq:elimation_jump}
 \tilde{z}_n \left|\hat{c}\left(k_{n}, {m}_{n},v_n,l\right)_{\tau}  -  \hat{c}^{*}\left(k_{n}, {m}_{n}\right)_{\tau} \right| \stackrel{\mathbb{P}}{\longrightarrow} 0,
\end{equation}
{\Blue for both $l=1,2$ if $X = X^{\prime}$,  for $l=1$ if (\ref{eq:m_n_range_noise})-(\ref{eq:r_w_condition_noise}) hold, and for $l = 2$ if (\ref{eq:m_n_range_noise}) and (\ref{eq:w_condition_noise_2}) hold.}
\end{lemma}

\begin{proof}
{\Blue Let \(E_{j}^{n}\) denote the conditional expectation with respect to \(\mathcal{F}_{(j-1) \Delta_{n}}\). We first check the proof for the case of $l=2$. The proof for $l=1$ is shown below}. We can write 
\begin{equation*}
\left|\left(\left(\bar{Y}_{j}^{n}\right)^2 -\frac{1}{2}\widehat{Y}_{j}^{n}\right)\mathbbm{1}_{\left\{\left|\bar{Y}_{j}^{n}\right| \leq \nu_{n}\right\}}-\left(\left(\bar{Y}_{j}^{* n}\right)^2-\frac{1}{2}\widehat{Y}_{j}^{* n}\right)\right| \leq \sum_{r=1}^{4} \eta_{j}^{n, r},
\end{equation*} 
where 
\begin{equation}\label{eq:decomp_elimination_jump}
\begin{aligned}
\eta_{j}^{n, 1} &=\left|\left(\bar{Y}_{j}^{n}\right) ^2 \mathbbm{1}_{\left\{\left|\bar{Y}_{j}^{n}\right| \leq \nu_{n}\right\}}-\left(\bar{Y}_{j}^{* n}\right) ^2 \mathbbm{1}_{\left\{\left|\bar{Y}_{j}^{* n}\right| \leq \nu_{n}\right\}}\right|, \\
\eta_{j}^{n, 2} &=\frac{1}{2}\left|\widehat{Y}_{j}^{n}\mathbbm{1}_{\left\{\left|\bar{Y}_{j}^{n}\right| \leq \nu_{n}\right\}}-\widehat{Y}_{j}^{* n}\mathbbm{1}_{\left\{\left|\bar{Y}_{j}^{*n}\right| \leq \nu_{n}\right\}} \right|,  \\
\eta_{j}^{n, 3} &=\left|\bar{Y}_{j}^{* n}\right|^{2} \mathbbm{1}_{\left\{\left|\bar{Y}_{j}^{* n}\right|>\nu_{n}\right\}},\\
\eta_j^{n,4} & =\frac{1}{2} \left|\hat{Y}_{j}^{* n}\right| \mathbbm{1}_{\left\{\left|\bar{Y}_{j}^{* n}\right|>\nu_{n}\right\}}.
\end{aligned}
\end{equation}
When $X = X^{\prime}, \eta_j^{n,1} = 0$. When $r \in (0,2]$, by the proof of Lemma 2 in \cite{chen2018inference}\footnote{{Or follow the proof of Lemma 13.2.6 in \cite{JacodProtter}.}}, under (\ref{eq:J_estimates}), there is a sequence \(a_n \to  0 \) such that \footnote{The notation $\bar Y, \hat Y$ is slightly different in \cite{chen2018inference}, which results in the different form of the following inequality. The relation between \(\varpi\) in our setting and \(\rho \)  in \cite{chen2018inference} is \(\varpi/2 + 1/4 = \rho\).}
\begin{equation} \label{eq:truncation1}
E_j^{n}\left[\left|\eta_{j}^{n, 1}\right|\right] \leq C \Delta_{n}^{(\varpi + 1/2) - r\varpi/2} a_{n}.
\end{equation}
Next, 
 $\eta^{n,2}_j = 0 $ when $X = X'$. When $r \in (0,2]$, we first consider
\begin{equation*}
\begin{split}
\left|\widehat{Y}_{j}^{n}-\widehat{Y}_{j}^{* n}\right|&= \left|\sum_{h=1}^{k_{n}}\left(g\left(\frac{h}{k_{n}}\right)-g\left(\frac{h-1}{k_{n}}\right)\right)^{2} \left(\left(\Delta_{j+h-1}^{n} Y_{j}^{* n} +\Delta_{j+h-1}^{n} X^{\prime \prime}_j\right)^2 - \left(\Delta_{j+h-1}^n {Y}_{j}^{* n}\right)^2\right)\right|\\
& \leq   C \left( \hat{X}^{\prime \prime n}_j + \sqrt{\hat{Y}^{* n}_j \hat{X}^{\prime \prime n}_j} \right).
\end{split}
\end{equation*}
 If we set \(\mu,\sigma = 0\), we have  \(\E^n_j\left|\hat{X}^{\prime \prime n}_j \right|^q= \E^n_j \left|\hat{X}^n_j\right|^q \leq K_{q}\Delta_n^{q/2}\left(\Delta_{n}^{q}+ \Delta_{n}^{q \wedge 1}\right)\) from (\ref{eq:X_bar_hat_bound}) {for $q>0$}. {Combining} with  (\ref{eq:Y_star_bound}) and taking $q = 1$, we have
 \begin{equation*} \label{eq:Y_hat_diff}
E_{j}^{n}\left|\widehat{Y}_{j}^{n}-\widehat{Y}_{j}^{* n} \right|\leq C\left(\Delta_n^{3/2} + \sqrt{\Delta_n^{1/2}\Delta_n^{3/2}}  \right)
\leq C \Delta_n.
\end{equation*}
By the notation (\ref{eq:notation_bar_hat}) and Cauchy-Schwarz inequality, we have for arbitrary $\alpha >0 $, 
 \begin{equation}
 \begin{split}
 \eta^{n,2}_{\Blue j} & \leq \frac{1}{2}\left|\widehat{Y}_{j}^{n}-\widehat{Y}_{j}^{* n}\right| +\widehat{Y}_{j}^{n}  \mathbbm{1}_{\left\{\left|\bar{Y}_{j}^{n}\right|<\nu_{n} < \left|\bar{Y}_{j}^{*n}\right|\right\}} \\
 & \quad +\widehat{Y}_{j}^{*n}  \mathbbm{1}_{\left\{\nu_n/2< \left|\bar{Y}_{j}^{*n}\right|<\nu_{n} < \left|\bar{Y}_{j}^{n}\right|\right\}} + \widehat{Y}_{j}^{*n}  \mathbbm{1}_{\left\{2 \left|\bar{Y}_{j}^{*n}\right|<\nu_{n} < \left|\bar{Y}_{j}^{n}\right|\right\}} \\
 & \leq \frac{1}{2}\left|\widehat{Y}_{j}^{n}-\widehat{Y}_{j}^{* n}\right| +\widehat{Y}_{j}^{n} \frac{\left|\bar{Y}_{j}^{*n} \right|^{\alpha} }{\nu_n^{\alpha}} + \widehat{Y}_{j}^{*n} \frac{\left|\bar{Y}_{j}^{*n} \right|^{\alpha} }{(\nu_n/2)^{\alpha}} + \widehat{Y}_{j}^{*n} \left(\frac{\left|\bar{X}^{\prime \prime n}_j \right| }{\nu_n/2}\wedge 1 \right),
 \end{split}
 \end{equation} where the last term is because $2 \left|\bar{Y}_{j}^{*n}\right|<\nu_{n} < \left|\bar{Y}_{j}^{n}\right|$ implies $ \left|\bar{X}^{\prime \prime n}_j \right|> \nu_n/2$, therefore, $\frac{\left|\bar{X}^{\prime \prime n}_j \right| }{\nu_n/2}\wedge 1 >\mathbbm{1}_{\left\{2 \left|\bar{Y}_{j}^{*n}\right|<\nu_{n} < \left|\bar{Y}_{j}^{n}\right|\right\}}. $
 Applying Cauchy-Schwarz inequality and  (\ref{eq:J_estimates})-(\ref{eq:X_bar_hat_bound}), we have 
\begin{equation}
\begin{split}
 \eta^{n,2}_{\Blue j} &\leq  C\left(\Delta_n + \sqrt{\Delta_n\Delta_n^{\alpha(1/2-\varpi) }  } + \sqrt{\Delta_n2^{\alpha}\Delta_n^{\alpha(1/2-\varpi) }} + \sqrt{\Delta_n \Delta_n^{1/2-\varpi r/2} a_n}\right).
\end{split}
\end{equation}Since \(\varpi < 1/2\), there exists $\alpha$ such that $\alpha\left( 1/2 - \varpi\right)\ge 1/2 $. Then,  
\[
	E_{j}^{n}\left|\eta_{j}^{n, 2}\right| \leq C \left(\Delta_n + \sqrt{\Delta_n^{3/2- \varpi r/2}a_n}\right).
\] 
By Cauchy-Schwarz and Markov's inequality and (\ref{eq:Y_star_bound}), for an arbitrary positive number \(m\), \begin{equation*}
\begin{split}
E_{j}^{n}\left|\eta_{j}^{n, 3}\right|  & \leq \sqrt{E_{j}^{n}\left|\bar{Y}^{* n}_j\right|^{4} \mathbb{P}_{j}^{n}\left(\left|\bar{Y}_{j}^{* n}\right|>v_{n}\right)} \leq C\sqrt{\Delta_n E_{j}^{n}\left|\bar{Y}^{* n}_j\right|^{m}/v_n^m } \leq C \sqrt{\Delta_n^{1+m/4 - m\varpi/2}}.
\end{split}
\end{equation*} 
Similarly, there exists $m$ such that $ E_j^{n}\left|\eta_j^{n, 3}\right| \leq C \Delta_n $. 
 By the same argument, \begin{equation*}
    \begin{split}
E_j^{n}\left|\eta_j^{n, 4}\right|  & \leq C \sqrt{E_j^{n}\left|\hat{Y}^{* n}_j\right|^{2} \mathbb{P}_j^{n}\left(\left|\bar{Y}_j^{*n}\right|>v_{n}\right)} \leq C\sqrt{\Delta_n^{p} \frac{E_j^{n}\left|\bar{Y}^{ n}_j\right|^{m}}{v_n^m}} \leq C \sqrt{\Delta_n^{1 +m/4- m\varpi/2}} \leq C \Delta_n.
\end{split}
\end{equation*}
Now we have the following when $r \in (0,2]$ by combining the above inequalities,
\begin{equation*}
\begin{split}
&{\Green E}\left(|\hat{c}\left(k_{n}, {m}_{n},v_n, 2\right)_{\tau}  -  \hat{c}^{*}\left(k_{n}, {m}_{n}\right)_{\tau} \right)\\
&\leq C\frac{1}{\phi_{k_{n}}\left(g\right)} \sum_{j=1}^{n-k_n+1} K_{{m}_{n} \Delta_n}\left(t_{j-1} - \tau\right)\left[\Delta_{n}^{\varpi + 1/2 -r\varpi/2} a_{n} + \Delta_n + \sqrt{\Delta_n^{3/2- \varpi r/2}a_n}\right]\\
& \leq C \Delta_n^{1/2}\sum_{j=1}^{n-k_n+1} K_{{m}_{n} \Delta_n}\left(t_{j-1} - \tau\right)\left[\Delta_{n}^{\varpi + 1/2 -r\varpi/2} a_{n} +\Delta_n+\sqrt{\Delta_n^{3/2- \varpi r/2}a_n} \right]\\
& \leq C \Delta_n^{-1/2}\left[\Delta_{n}^{\varpi + 1/2 -r\varpi/2} a_{n} +\Delta_n +\sqrt{\Delta_n^{3/2- \varpi r/2}a_n} \right]\\
& \leq C \left[\Delta_{n}^{\varpi  -r\varpi/2} a_{n} +\sqrt{\Delta_n} +\Delta_n^{1/4- \varpi r/4} \sqrt{a_n} \right],
\end{split}
\end{equation*} 
by the property of the kernel. 
{\Blue Therefore, recalling (\ref{eq:m_n_range_noise})},  to obtain \begin{equation*}
\tilde{z}_n \left|\hat{c}\left(k_{n}, {m}_{n},v_n, 2\right)_{\tau}  -  \hat{c}^{*}\left(k_{n}, {m}_{n}\right)_{\tau} \right| \stackrel{\mathbb{P}}{\longrightarrow} 0,
\end{equation*}
 we need the condition (\ref{eq:w_condition_noise_2}). 
This concludes (\ref{eq:elimation_jump}) for $l = 2$. 

When $l =1$, $\hat{c}\left(k_{n}, {m}_{n},v_n, 1\right) - \hat{c}^{*}\left(k_{n}, {m}_{n}\right) $ has a similar decomposition like (\ref{eq:decomp_elimination_jump}) with $\eta^{n,2}_j  = \frac{1}{2}\left| \widehat{Y}_j^{n}-\widehat{Y}_j^{* n}\right|$ and  $\eta^{n,4}_j = 0 $. 
Therefore, 
\begin{equation*}
{\Green E}\left(|\hat{c}\left(k_{n}, {m}_{n},v_n, 1\right)_{\tau}  -  \hat{c}^{*}\left(k_{n}, {m}_{n}\right)_{\tau} \right) \leq C \left[\Delta_{n}^{\varpi  -r\varpi/2} a_{n} +\sqrt{\Delta_n}\right].
\end{equation*} 
In that case, the condition \begin{equation*}
\varpi \geq \frac{(a- \frac{1}{4})\wedge (1-(a- \frac{1}{4} ) ) - \frac{1}{4}}{2-r},
\end{equation*} gives (\ref{eq:elimation_jump}) for $l = 1$. Since $\varpi < 1/2$, we necessarily need that $r<\frac{5}{2}-2\left[\left(a-\frac{1}{4}\right) \wedge\left(1-a+\frac{1}{4}\right)\right]$. 

Finally, when $X=X'$, note that ${\Green E}\left(|\hat{c}\left(k_{n}, {m}_{n},v_n, l\right)_{\tau}  -  \hat{c}^{*}\left(k_{n}, {m}_{n}\right)_{\tau} \right) \leq C \sqrt{\Delta_n}$,  for $l = 1,2 $, since $\eta^{n,1}_j = \eta^{n,2}_j = 0$.
Therefore, we can conclude (\ref{eq:elimation_jump}) {\Blue in all the cases stated in the statement of the lemma.}
\end{proof}

\subsection{Proof of stable convergence in law for continuous process}
With the lemma above, {it suffices to} prove Theorem \ref{thm} for the non-truncated estimator (\ref{eq:non-truncated_noise}) under \thref{X_continuous}. We first introduce some notations needed for the proofs. Then, we recall some needed estimates and preliminary results. Finally, we proceed to prove the result through three lemmas.

\subsubsection*{Needed Notation}
\begin{enumerate}

\item Define
\begin{equation} \label{eq:notation}
\begin{aligned}
    &{\phi(Y)_{i}^{n} = (\overline{Y}_i^n )^2 - \frac{1}{2}\hat{Y}_i^n = (\overline{X}_i^n + \overline{\epsilon}_i^n )^2 - \frac{1}{2}\hat{Y}_i^n},\\
&\phi_{i, j}^{n} = (\sigma_{(i-j-1)\Delta_n}\overline{W}_i^n + \overline{\epsilon}_i^n )^2 - \frac{1}{2}\hat{\epsilon}_i^n, \\
&\Psi_{i,j}^n = \E( \phi_{i,j}^n|\mathcal{H}_{(i-1)\Delta_n}) - (\sigma_{(i-j-1)\Delta_n}\overline{W}_i^n)^2.
\end{aligned}
\end{equation}
\item With any process U, we associate the variables
\begin{equation*}
    \begin{aligned} \Gamma(U)_{i}^{n} &=\sup _{t \in\left[(i-1) \Delta_{n}, i \Delta_{n}+k_n \Delta_n\right]}|U_{t}-U_{(i-1) \Delta_{n}}| ,\\ 
    \Gamma^{\prime}(U)_{i}^{n} &=\left(\mathbb{E}\left( \left.\left(\Gamma(U)_{i}^{n}\right)^{4} \right| \mathcal{F}_{(i-1) \Delta_{n}}\right)\right)^{1 / 4}.\end{aligned}
\end{equation*}
\end{enumerate}

The following decomposition will be instrumental to deduce the behavior of the estimation error:
\begin{equation*}
	\hat{c}\left(k_{n}, {m}_{n}\right)_{\tau}-c_{\tau}
    = \sum_{l=1}^{5} \overline{H}(l)^{n},
\end{equation*}
where
\begin{equation*}
    \begin{aligned}
        \overline{H}(1)^{n}&=\frac{1}{\phi_{k_n}(g)} \sum_{i=1}^{n-k_{n}+1}K_{{m}_{n} \Delta_n}(t_{i-1}-\tau)\left(\phi_{i, 0}^{n} - \E \left(\left.\phi_{i, 0}^{n}\right| \mathcal{F}_{(i-1) \Delta_{n}}\right)\right),\\
        \overline{H}(2)^{n}&=\int_{0}^{T} K_{{m}_{n} \Delta_n}(t-\tau) c_{t} d t- c_{\tau},\\
        \overline{H}(3)^{n}&=\frac{1}{\phi_{k_n}(g)} \sum_{i=1}^{n-k_{n}+1}K_{{m}_{n} \Delta_n}(t_{i-1}-\tau)\left({\phi(Y)_{i}^{n}} - \phi_{i, 0}^{n}\right),\\
        \overline{H}(4)^{n}&=\frac{1}{\phi_{k_n}(g)} \sum_{i=1}^{n-k_{n}+1}K_{{m}_{n} \Delta_n}(t_{i-1}-\tau) \E \left(\left.\Psi_{i, 0}^{n}\right| \mathcal{F}_{(i-1) \Delta_{n}}\right),\\
        \overline{H}(5)^{n}&=\Delta_n \sum_{i=1}^{n-k_{n}+1}K_{{m}_{n} \Delta_n}(t_{i-1}-\tau)  {c_{t_{i-1}}}- \int_{0}^{T} K_{{m}_{n} \Delta_n}(t-\tau) c_{t} d t.
\end{aligned}
\end{equation*}
The 1st term is the statistical error, while the 2nd term is the local approximation error. Each of these will contribute one term to the asymptotic variance in (\ref{CLTTCs}-i). Up to a negligible term, which is analyzed in $\bar{H}(4)^{n}$, the third term is obtained by freezing the volatility $\sigma$ in $\bar{X}_{i}^{n}$ to the value $\sigma_{(i-1)\Delta_{n}}$. The last term analyzes the error due to approximating the integral by its associated Riemann sum.

\subsubsection*{{Some Preliminary Estimates and Results}}

\begin{enumerate}
\item  By Lemma 16.5.14 from \cite{JacodProtter}, for some constant $C$,
    \begin{equation}\label{eq:bound1}
        \left|\E\left(\left.\phi(Y)_{i}^{n}-\phi_{i, 0}^{n} \right| \mathcal{F}_{(i-1) \Delta_{n}}\right)\right| \leq C \Delta_n^{3/4}\left(\Delta_n^{1/4} +\Gamma^{\prime}(\mu)_{i}^{n}+\Gamma^{\prime}(\widetilde{\sigma})_{i}^{n}+\Gamma^{\prime}(\gamma)_{i}^{n}\right).
    \end{equation}
\item  \label{lis:AN}As in Lemma 16.5.15 in \cite{JacodProtter}, if an array $\left(\delta_{i}^{n}\right) $ satisfies \begin{equation} \label{eq: AN}
        0 \leq \delta_{i}^{n} \leq K, \quad \Delta_{n} \mathbb{E}\left(\sum_{i=1}^{n} \delta_{i}^{n}\right) \rightarrow 0 ,
    \end{equation}
    then, for any $q >0$, the array $\left(\left|\delta_{i}^{n}\right|^{q}\right)$ also satisfies (\ref{eq: AN}).
    Furthermore, if U is a c\`adl\`ag bounded process, the two arrays $\left(\Gamma(U)_{i}^{n}\right)$ and $ \left(\Gamma^{\prime}(U)_{i}^{n}\right)$ also satisfy (\ref{eq: AN}).
    
\item  Under \thref{noise} and (\ref{eq:sigma}), by Lemma 16.5.13 in \cite{JacodProtter}, we have for all $q >0$, 
    \begin{equation} \label{eq:phi_bound}
\mathbb{E}\left(\left. \left|\phi\left( Y\right)_{i}^{n}\right|^{q}+\left|\phi_{i, 0}^{n}\right|^{q} \right| \mathcal{F}_{(i-1) \Delta_{n}}\right) \leq C_{q}\Delta_n^{q/2},
\end{equation}
\begin{equation} \label{eq:phi2}
    \mathbb{E}\left(\left.\left|\phi\left(Y\right)_{i}^{n}-\phi_{i, 0}^{n}\right|^{2} \right| \mathcal{F}_{(i-1) \Delta_{n}}\right) \leq C\Delta_n \left(\Delta_n^{1/2}+\left(\Gamma^{\prime}(\sigma)_{i}^{n}\right)^{2}\right).
\end{equation}
Similarly, we can obtain 
\begin{equation} \label{eq: phi_i_j}
\mathbb{E}\left(\left. \left|\phi_{i, j}^{n}\right|^{q} \right| \mathcal{F}_{(i-j-1) \Delta_{n}}\right) \leq C_{q}\Delta_n^{q/2},
\end{equation}
since $\sigma $ is bounded.
\item Let $\gamma_{t}^{\prime}=\mathbb{E}\left(\left.\left|\epsilon_{t}\right|^{3} \right| \mathcal{F}^{(0)}\right)$. Under \thref{noise} and (\ref{eq:sigma}), by Lemma 16.5.12 in \cite{JacodProtter},  $\Psi_{i,j}^n$ defined in  (\ref{eq:notation}) is such that  
    \begin{equation} \label{eq:psi}
    \mathbb{E}\left(\left.   \left|\Psi_{i, j}^{n}\right| \right| \mathcal{F}_{(i-1) \Delta_{n}}\right) \leq C \Delta_n + C \Delta_n ^{3/4}\left(\Gamma^{\prime}(\gamma)_{i}^{n}+\Gamma^{\prime}\left(\gamma^{\prime}\right)_{i}^{n}\right),
    \end{equation}
\begin{equation} \label{eq:psi2}
\mathbb{E}\left(\left.\left|\Psi_{i, j}^{n}\right|^{2} \right| \mathcal{F}_{(i-1) \Delta_{n}}\right) \leq C \Delta_n^{3/2}.
\end{equation}

\item As {$n \rightarrow \infty$ so that ${m}_{n} \rightarrow \infty$ and ${m}_{n} \Delta_n \rightarrow 0$},
\begin{equation} \label{eq:K2}
    \begin{aligned}
         {m}_{n} \Delta_n^2 \sum  K _{{m}_{n} \Delta_n}^2(t_{i-1} - \tau) \rightarrow \int K^2(x) d x,\\
         \Delta_n \sum_{i=j}^{n}\left|K_{{m}_{n} \Delta_n}\left(t_{i-1}-\tau\right)\right| \rightarrow \int|K(x)| d x.
    \end{aligned}
    \end{equation}
\item {By It\^o's} Lemma and Burkholder-Davis-Gundy inequalities (see Section 2.1.5 in \cite{JacodProtter}),   we have for all $s,t \geq 0$ and $p \geq 2$\begin{equation} \label{eq:sigma2_bound}
\mathbb{E}\left(\left. \sup _{r \in[0, s]}\left|\sigma_{t+r}^2-\sigma_{t}^2\right|^{p} \right| \mathcal{F}_{t}\right) \leq C_{p} s,
\end{equation}
\begin{equation} \label{eq:sigma_bound}
\mathbb{E}\left(\left.\sup _{r \in[0, s]}\left|\sigma_{t+r}-\sigma_{t}\right|^{p} \right| \mathcal{F}_{t}\right) \leq C_{p} s.
\end{equation}
{Also, we have}
\begin{equation} \label{eq:Gamma'}
\Gamma^{\prime}(\sigma)_i^n  = \left(\mathbb{E}\left(\left. \sup _{t \in\left[(i-1) \Delta_{n}, i \Delta_{n}+{k_n \Delta_n}\right]}\left|\sigma_{t}-\sigma_{(i-1) \Delta_{n}}\right|^{4} \right| \mathcal{F}_{(i-1) \Delta_{n}}\right)\right)^{1 / 4} \leq C \left(k_n\Delta_n\right)^{1/4}.
\end{equation}
\end{enumerate}

Theorem \ref{thm} will then follow from the following lemmas:
\begin{lemma} \thlabel{part1}
Under \thref{kernel,noise,X_continuous}, with ${m}_{n}\to\infty$, ${m}_{n} \Delta_n \rightarrow 0$, and ${m}_{n}\sqrt{\Delta_n} \rightarrow \infty $, we have
$${m}_{n}^{1/2}\Delta_n^{1/4}\overline{H}(1)^n \stackrel{st}{\longrightarrow} Z_{\tau} ,$$ 
where $Z_{\tau}$ is described in Theorem \ref{thm}. 
\end{lemma}

\begin{lemma} \thlabel{part2}
 Under \thref{kernel,noise,X_continuous},  with ${m}_{n}\to\infty$, ${m}_{n} \Delta_n \rightarrow 0$ and ${m}_{n} \Delta_n^{3/4} \rightarrow \beta \in [0,\infty) $, \begin{equation*}
{m}_{n}^{1 / 2} \Delta_{n}^{1 / 4} \left(\overline{H}(1)^n+ \overline{H}(2)^n\right) \stackrel{st}{\longrightarrow} Z_{\tau}+ \beta Z_{\tau}^{\prime},
\end{equation*}
where  $Z_{\tau}^{\prime} $ is described in Theorem \ref{thm}. 
\end{lemma}

\begin{lemma} \thlabel{AN_lemma}
Under  \thref{kernel,noise,X_continuous},  assuming  ${m}_{n}\Delta_n^{3/4} \rightarrow \beta \in [0,\infty] $, we have
\begin{equation}\label{NeglTrms}
	z_n\overline{H}(l)^{n} \stackrel{\mathrm{P}}{\longrightarrow} 0 \text { for } l=3,4,5,
\end{equation}
where 
\[
	z_n = \begin{cases}
      {m}_{n}^{1/2}\Delta_n^{1/4} & \text{ if }  {m}_{n}\Delta_n^{3/4} \to \beta < \infty,\\
      \frac{1}{\sqrt{{m}_{n}\Delta_n}} & \text{ if }   {m}_{n}\Delta_n^{3/4} \to \beta = \infty.
    \end{cases} 
 \]
\end{lemma}
{We prove the lemmas above in three steps. In Step 1, we start to prove the last lemma which is more straightforward than the other two. In Step 2, we prove Lemma \ref{part1}. In Step 3, we show {Lemma} \ref{part2}.}

\subsubsection*{Step 1}
For $l = 3$, set 
\begin{equation*}
    \begin{aligned} \zeta_{i}^{n} &=\frac{1}{\phi_{k_n}(g)}K_{{m}_{n} \Delta_n}(t_{i-1} - \tau) \left(\phi\left(Y\right)_{i}^n-\phi_{i, 0}^{n}\right) .\\
    \end{aligned}
\end{equation*} 
By Lemma 2.2.10 in \cite{JacodProtter}, the result follows if the  array $z_n\E\left(|\zeta_{i}^{n}| \mid \mathcal{F}_{(i-1) \Delta_{n}}\right)$ is asymptotically negligible. To this end, note that (\ref{eq:bound1})  yields 
$$
\E\left(|\zeta_{i}^{n}| \mid \mathcal{F}_{(i-1) \Delta_{n}}\right)\leq C \Delta_n^{1/4}\Delta_n |K_{{m}_{n} \Delta_n}(t_{i-1} - \tau)|  \E\left( \left(\Delta_n^{1/4} +\Gamma^{\prime}(\mu)_{i}^{n}+\Gamma^{\prime}(\widetilde{\sigma})_{i}^{n}+\Gamma^{\prime}(\gamma)_{i}^{n}\right)\right),$$
where recall that we are assuming that $\tilde{\sigma}$, $\mu$, and $ \gamma$  are c\`adl\`ag bounded processes by localization.  
Thus, from Lemma 16.5.15 in \cite{JacodProtter}, $\left((\Gamma^{\prime}(\widetilde{\sigma})_{i}^{n})^2\right), \left((\Gamma^{\prime}({\mu})_{i}^{n})^2\right), \mbox{and } \left((\Gamma^{\prime}(\gamma)_{i}^{n})^2 \right)  $ satisfy (\ref{eq: AN}). By Cauchy-Schwarz inequality and (\ref{eq:K2}), 
\begin{equation} \label{eq:cauchy}
\begin{aligned}
    \Delta_n\sum_{j=1}^{n-k_n+1} |K_{{m}_{n} \Delta_n}(t_{i-1} - \tau)| \E\left(\Gamma^{\prime}({\mu})_{i}^{n}\right) &\leq \sqrt{\sum \Delta_nK_{{m}_{n} \Delta_n}^2(t_{i-1} - \tau) \sum \Delta_n \E\left(\Gamma^{\prime}({\mu})_{i}^{n}\right)^2}\\
    &= \mathrm{o}\left(\frac{1}{\sqrt{{m}_{n} \Delta_n}}\right).
\end{aligned}
\end{equation}We can obtain similar results on $(\Gamma^{\prime}(\widetilde{\sigma})_{i}^{n})$ and $\Gamma^{\prime}(\Upsilon)_{i}^{n} $. Thus,
\begin{equation*}
\begin{aligned}
z_n\sum_{j=1}^{n-k_n+1} \E\left(|\zeta_{i}^{n}| \mid \mathcal{F}_{(i-1) \Delta_{n}}\right) &\leq \mathrm{O}(z_n\Delta_n^{1/2}) +  \mathrm{o}(\frac{z_n}{{m}_{n}^{1/2}\Delta_n^{1/4}})
\,{\stackrel{n\to\infty}{\longrightarrow}}\, 0.
\end{aligned}
\end{equation*}
This finishes the proof of \thref{AN_lemma} for $l = 3$.
\medskip

For $l=4$, by (\ref{AsympCndkn}), (\ref{eq:phi_g}), (\ref{eq:psi}), and (\ref{eq:K2}),  we have
\begin{equation*}
\begin{aligned}
{|\overline{H}(4)^{n}|} &\leq C \frac{1}{k_n}\sum_{j=1}^{n - k_n+1}|K_{{m}_{n} \Delta_n}\left(t_{j-1} - \tau \right)| \left(\Delta_n +  \Delta_n ^{3/4}\left(\Gamma^{\prime}(\gamma)_{i}^{n}+\Gamma^{\prime}\left(\gamma^{\prime}\right)_{i}^{n}\right)\right)\\
&= \mathrm{O}\left(\Delta_n^{1/2}\right) + \mathrm{o}\left(\frac{\Delta_n^{1/4}}{\sqrt{{m}_{n} \Delta_n}}\right),
\end{aligned}
\end{equation*}
where we used a similar argument as in (\ref{eq:cauchy}) to deduce the second inequality above. 
Thus, we deduce (\ref{NeglTrms}) for $l=4$.
\medskip

For $l=5$, we have
\begin{align}
\label{L1l5}
\left|{\overline{H}(5)^{n}}\right|& \leq \int_{(n-k_n+1)\Delta_n}^{T}  \left|K_{{m}_{n} \Delta_n}(t-\tau) \sigma_{t}^{2}\right| d t  \\
\label{L2l5}
&\qquad \quad +\sum_{j = 1}^{n-k_n+1}\int_{t_{j-1}}^{t_{j} }\left|  K_{{m}_{n} \Delta_n}(s-\tau )\sigma^2_s-K_{{m}_{n} \Delta_n}(t_{j-1}-\tau )\sigma^2_{(j-1)\Delta_n}\right| d s\\
\label{L3l5}
&\leq C \frac{1}{{m}_{n} \sqrt{\Delta_n}} + (n-k_n-1) \Delta_n \left( \frac{1}{{m}_{n}\Delta_n^{1/2}} + \frac{1}{{m}_{n}^2\Delta_n}\right)\\
\label{L4l5}
& = \mathrm{O}_{P}\left(\frac{1}{{m}_{n} \sqrt{\Delta_n}}\right),
\end{align}
where the first term in (\ref{L3l5}) follows from the boundedness of $K$ and $\sigma$ as follows:
\begin{equation*}
\begin{aligned}
\int_{(n-k_n+1)\Delta_n}^{T}  \left|K_{{m}_{n} \Delta_n}(t-\tau) \sigma_{t}^{2}\right| d t 
\leq C \frac{1}{{m}_{n} \Delta_n} k_n \Delta_n  = C \frac{1}{{m}_{n} \sqrt{\Delta_n}},
\end{aligned}
\end{equation*}
while the second term in (\ref{L3l5}) can be deduced by (\ref{eq:sigma2_bound}) and Lipschitz property of $K$.  Indeed, for $s \in [t_{j-1},t_j]$ and $b_n := {m}_{n} \Delta_n $,
 \begin{equation*}
\begin{aligned}
&\left|  K_{b_n}(s-\tau )\sigma^2_s-K_{b_n}(t_{j-1}-\tau )\sigma^2_{(j-1)\Delta_n}\right| \\
&\leq \left|  K_{b_n}(s-\tau )\sigma^2_s-K_{b_n}(s-\tau ) \sigma^2_{(j-1)\Delta_n}\right| + \left|  K_{b_n}(s-\tau )\sigma^2_{(j-1)\Delta_n}-K_{b_n}(t_{j-1}-\tau )\sigma^2_{(j-1)\Delta_n}\right|\\
&=\mathrm{O}_{P}\left(\frac{1}{{m}_{n}\Delta_n^{1/2}} \right) + \mathrm{O}_{P}\left(\frac{1}{{m}_{n}^2\Delta_n} \right), \end{aligned}
\end{equation*}
 So, we deduce (\ref{NeglTrms}) for $l=5$.

\subsubsection*{Step 2}
To show \thref{part1}, we need several preliminary lemmas. We employ the `block splitting' method proposed in \cite{JacodProtter} (see Section 16.5.4, page 548 therein). Recall that
\begin{equation*}
\overline{H}(1)^n = \sum_{i=1}^{n-k_{n}+1} \zeta_{i}^{n},
\end{equation*}
where $\zeta_{i}^{n}=\frac{1}{\phi_{k_n}(g)}K_{{m}_{n} \Delta_n}(t_{i-1}-\tau)\left(\phi_{i, 0}^{n}-  \E \left(\left.\phi_{i, 0}^{n}\right| \mathcal{F}_{(i-1) \Delta_{n}}\right)\right)$.  
The variables $\zeta_i^n$ are not martingale differences. To use martingale methods, we fix an integer $m \geq 1$, and {divide} the summands in the definition of $\overline{H}(1)^n $ into blocks of size $mk_n$ and $k_n$. Concretely,
the ${\ell}$th big block, of size $mk_n$, contains the indices between $I(m, n, \ell)=(\ell-1 )(m+1) k_{n}+1$ and $I(m, n, \ell)+m k_{n}-1$. The number of such blocks before time t is $l_{n}(m)=\left[\frac{n-k_n+1}{(m+1) k_{n}}\right]$. These big  blocks are separated by small blocks of size $k_n$, and the ``real'' time corresponding to the beginning of the $\ell$th big block is $t(m, n, \ell)=(I(m, n, \ell)-1) \Delta_{n}$. Then we introduce the summand over all the big blocks,
\begin{equation}\label{DfnZmn}
{Z^{n}(m):=\sum_{{\ell}=1}^{l_{n}(m)} \delta(m)_\ell^n:=\sum_{{\ell}=1}^{l_{n}(m)} \sum_{r=0}^{m k_{n}-1} \zeta_{I(m, n, \ell)+r}^{n},}
\end{equation}
 Note that the sequence $\left( \delta(m)^n_\ell \right) $ are now martingale differences w.r.t. the discrete filtration {$\mathcal{G}_\ell = \mathcal{F}_{(I(m, n, \ell+1)-1) \Delta_{n}}$, for} $\ell=1,\dots,l_{n}(m)$.

We now show that the contribution of the small blocks, i.e. $\overline{H}(1)^n-Z^{n}(m)$, is asymptotically ``negligible'' compared to ${m}_{n}^{-1/2}\Delta_n^{-1/4}$.

\begin{lemma}Under \thref{kernel,noise,X_continuous}, we have
\begin{equation*}
\lim _{m \rightarrow \infty} \limsup _{n \rightarrow \infty} \mathbb{E}\left({m}_{n}^{1/2}\Delta_n^{1/4}\left|\overline{H}(1)^n-{Z^{n}(m)}\right|\right)=0
\end{equation*}
\end{lemma}

\begin{proof}
Denote by $J(n, m)$  the set of all integers $j$ between  1 and $n-k_{n}+1$, which are \textit{not }in the big blocks (i.e., those corresponding to the small blocks). We further divide $J(n,m)$ into $k_n$ disjoint subsets $J(n,m,r)$ for $r = 1,...,k_n$, where $J(n,m,r)$ is the set of all $j \in J(n,m)$ equal to $r$ modulo $k_n$.  Then,
\begin{equation*}
\overline{H}(1)^{n}-Z^{n}(m)=\sum_{r=1}^{k_{n}} \sum_{j \in J(n,m, r)} \zeta_{j}^{n}.
\end{equation*}
Observe that $\mathbb{E}\left(\left.\zeta_{j}^{n} \right| \mathcal{F}_{(j-1) \Delta_{n}}\right)=0$ and $\zeta_{j}^{n}$ is $\mathcal{F}_{\left(j+k_{n}\right) \Delta_{n}}$ measurable. Then $\sum_{j \in J(n,m, r)} \zeta_{j}^{n}$ is the sum of martingale increments, because any two distinct indices in $ J (n, m, r )$ are more than $k_n$ apart.

Therefore by (\ref{eq:phi_bound}) {and the fact that $\mathbb{E}\left(\left.\zeta_{j}^{n} \right| \mathcal{F}_{(j-1) \Delta_{n}}\right)=0$}, for some constant $C$ (changing from line to line) and large enough $n$,
\begin{equation*}
\begin{aligned}
\mathbb{E}\left(\left|\sum_{j \in J(n,m, r)} \zeta_{j}^{n}\right|^{2}\right) &=\mathbb{E}\left(\sum_{j \in J(n, m, r)}\left|\zeta_{j}^{n}\right|^{2}\right) \\
&\leq C \sum_{j \in J(n, m, r)} \frac{\Delta_n}{\phi_{k_n}^2(g)}K_{{m}_{n} \Delta_n}^2(t_{j-1 } - \tau) \\
&\leq C \frac{1}{(m+1)k_n^3{m}_{n} \Delta_n} \int K^2(u) du \\
&\leq C \frac{\Delta_n^{1/2}}{m {m}_{n}},
\end{aligned}
\end{equation*}
where the last inequality holds because of (\ref{AsympCndkn}) and the second inequality holds because, recalling that two consecutive $j$'s in $J(n,m,r)$ are separated by $(m+1)k_{n}$, we have
\begin{equation*}
\begin{aligned}
 (m+1)k_n {m}_{n} \Delta_n^2\sum_{j \in J(n,m, r)}  K_{{m}_{n} \Delta_n}^2(t_{j-1 } - \tau) \, \stackrel{n\to\infty}{\longrightarrow}\, \int K^2(u)d u.
\end{aligned}
\end{equation*}
Then, 
\begin{equation*}
\mathbb{E}\left({m}_{n}^{1/2}\Delta_n^{1/4}\left|\overline{H}(1)^n-Z^{n}(m)\right|\right) \leq C {m}_{n}^{1/2}\Delta_n^{1/4}k_n \sqrt{\frac{\Delta_n^{1/2}}{m {m}_{n}}} \leq C\left(\frac{1}{\sqrt{m}}\right),
\end{equation*}
for large enough $n$. As $m \rightarrow \infty$, the above quantity goes to $0$ and we get the result.
\end{proof} 

Next, we modify the ``big-blocks'' process $Z^{n}(m)$ defined in (\ref{DfnZmn}) in such a way that each summand involves the volatility at the beginning of the corresponding large block. Recalling the notation in (\ref{eq:notation}), we set
\begin{align}
&\eta_{i, r}^{n} = \frac{1}{\phi_{k_n}(g)} K_{{m}_{n} \Delta_n}\left(t_{i-1}-\tau\right)\left(\phi_{i, r}^{n} - \E \left(\left.\phi_{i, r}^{n}\right| \mathcal{F}_{(i-r-1) \Delta_{n}}\right)\right),\label{DfnEta0}\\
&\eta_{i,r}^{\prime n} = \frac{1}{\phi_{k_n}(g)} K_{{m}_{n} \Delta_n}(t_{i-1}-\tau)\left(\E \left(\left.\phi_{i, r}^{n}\right| \mathcal{F}_{(i-r-1) \Delta_{n}}\right)- \E \left(\left.\phi_{i, r}^{n}\right| \mathcal{F}_{(i-1) \Delta_{n}}\right)\right),\label{DfnEtaprime0}\\
&M^{n}(m)=\sum_{i=1}^{l_{n}(m)} \sum_{r=0}^{m k_{n}-1} \eta_{I(m,n,i)+r, r}^{n}, \quad M^{\prime n}(m)=\sum_{i=1}^{l_{n}(m)}\sum_{r=0}^{m k_{n}-1} \eta_{I(m,n,i)+r, r}^{\prime n}.
\label{DfnEtaprime0b}
\end{align}

\begin{lemma}
Under \thref{kernel,noise,X_continuous}, for a fixed $m$,
\begin{equation*}
\lim _{n \rightarrow \infty} \mathbb{E}\left({m}_{n}^{1/2}\Delta_n^{1/4} \left|Z^{n}(m)-M^{n}(m)-M^{\prime n}(m)\right|\right)=0.
\end{equation*}
\end{lemma}
\begin{proof}
We use a similar method as in the previous lemma: Let $J^{\prime}(n, m)$ the set of all integers $j$ between 1 and $n - k_n + 1$, which are \textit{inside} the big blocks, that is of the form $j = I(m,n,i) + l$ for some $i \geq 1$ and $l \in \{0,\cdots,m k_n - 1\}$. Let $J^{\prime}(n,m,r)$ be the set of all $j \in J^{\prime}(n, m)$ equal to $r$ modulo $k_n$. We can then write
\begin{equation*}
Z^{n}(m)-M^{n}(m) - M^{\prime n}(m) = \sum_{r=1}^{k_{n}} \sum_{j \in J^{\prime }(n,m, r)} \theta_{j}^{n},
\end{equation*}
where $\theta_{j}^n = \frac{1}{\phi_{k_n}(g)}K_{{m}_{n} \Delta_n}(t_{j-1}-\tau)\left( \phi_{j, 0}^{n} - \phi_{j,l}^n -\E \left(\left.\phi_{j, 0}^{n}-\phi_{j, l}^{n}\right| \mathcal{F}_{(j-1) \Delta_{n}}\right) \right) ,$ when $j=I(m, n, i)+l$.
Note  $\phi_{j,0} $ and $\phi_{j,l}$ have the same noise part, $-\frac{1}{2}\hat{\epsilon}_j^n $, and the cross term $ \overline{W}_j^n \overline{\epsilon}_j^n $ has expectation $0$. Then, for some constant $C$ and large enough $n$,
\begin{equation*}
\begin{aligned}
\E \left| \theta_j^n \right|^2 &\leq \frac{1}{\phi_{k_n}^2(g)}K_{{m}_{n} \Delta_n}^2(t_{j-1}-\tau)\E \left| \phi_{j, 0}^{n} - \phi_{j,l}^n \right|^2\\
& =  \frac{1}{\phi_{k_n}^2(g)}K_{{m}_{n} \Delta_n}^2(t_{j-1}-\tau) \E\left(\left(\sigma_{(j-1)\Delta_n}^2 - \sigma_{(j-l-1)\Delta_n}^2\right)^2 \left(\overline{W}_j^n\right)^4 \right)\\
&\leq C  K_{{m}_{n} \Delta_n}^2(t_{j-1}-\tau)  m k_n  \Delta_n^3 \quad \mbox{for } j \in J^{\prime}(n,m,r),
\end{aligned}
\end{equation*}
where the last inequality follows by conditioning on $\mathcal{F}_{(j-1)\Delta_{n}}$, using that  $\E\Big[\left( \overline{W}_j^n\right)^4 |\mathcal{F}_{(j-1)\Delta_{n}}\Big]  =  3 \phi_{k_n}(g)^2 \Delta_n^2$, and applying (\ref{eq:sigma2_bound}).
As in the proof of the previous lemma,
\begin{equation*}
\begin{aligned}
\mathbb{E}\left(\left|\sum_{j \in J^{\prime}(n, m, r)}\theta_{j}^{n}\right|^{2}\right)&= \mathbb{E}\left(\sum_{j \in J^{\prime}(n, m, r)}\left|\theta_{j}^{n}\right|^{2}\right)\\
&\leq C \Delta_n^3 k_n\sum_{j \in J^{\prime }(n,m, r)} K_{{m}_{n} \Delta_n}^2(t_{j-1}-\tau)\\
&\leq C\frac{\Delta_n}{{m}_{n}} \int K^2(u) d u.
\end{aligned}
\end{equation*}
So we have 
\begin{equation*}
   {\mathbb{E}\left({m}_{n}^{1/2}\Delta_n^{1/4} \left|Z^{n}(m)-M^{n}(m)-M^{\prime n}(m)\right|\right) \leq C {m}_{n}^{1/2}\Delta_n^{1/4} k_n \sqrt{\frac{\Delta_n}{{m}_{n}}}= \mathrm{O}(\Delta_n^{1/4}) \rightarrow 0.}
\end{equation*}
\end{proof}

Now we prove $M^{\prime n}(m) $, defined in (\ref{DfnEtaprime0b}), is asymptotically negligible.
\begin{lemma}
Under \thref{kernel,noise,X_continuous},
\begin{equation*}
\lim _{n \rightarrow \infty} \mathbb{E}\left({m}_{n}^{1/2}\Delta_n^{1/4} \left|M^{\prime n}(m)\right|\right)=0.
\end{equation*}
\end{lemma}

\begin{proof}
Recall that $\Psi_{i,j}^n = \E( \phi_{i,j}^n|\mathcal{H}_{(i-1)\Delta_n}) - (\sigma_{(i-j-1)\Delta_n}\overline{W}_i^n)^2$ and, since $\mathcal{H}_{t}=\mathcal{F}^{(0)}\,\otimes \,\sigma\left(\epsilon_{s} : s \in[0,t)\right)$, 
\begin{align*}
	\mathbb{E}\left(\left.\Psi_{i+r, r}^{n} \right| \mathcal{F}_{(i-1) \Delta_{n}}\right)&=\mathbb{E}\left(\left. \phi_{i+r,r}^n \right| \mathcal{F}_{(i-1) \Delta_{n}}\right)-
	\mathbb{E}\left(\left. (\sigma_{(i-1)\Delta_n}\overline{W}_{i+r}^n)^2\right| \mathcal{F}_{(i-1) \Delta_{n}}\right),\\
	\mathbb{E}\left(\left.\Psi_{i+r, r}^{n} \right| \mathcal{F}_{(i+r-1) \Delta_{n}}\right)&=\mathbb{E}\left(\left.\phi_{i+r,r}^n \right| \mathcal{F}_{(i+r-1) \Delta_{n}}\right)-
	\mathbb{E}\left(\left. (\sigma_{(i-1)\Delta_n}\overline{W}_{i+r}^n)^2\right| \mathcal{F}_{(i+r-1) \Delta_{n}}\right).
\end{align*}
Since $\overline{W}_{i+r}^n$ is a linear combination of $W_{(i+r)\Delta_{n}},\dots, W_{(i+r+k_{n}-1)\Delta_{n}}$, we have:
\[
	\mathbb{E}\left(\left. (\sigma_{(i-1)\Delta_n}\overline{W}_{i+r}^n)^2\right| \mathcal{F}_{(i-1) \Delta_{n}}\right)=\mathbb{E}\left(\left. (\sigma_{(i-1)\Delta_n}\overline{W}_{i+r}^n)^2\right| \mathcal{F}_{(i+r-1) \Delta_{n}}\right),
\]
and, thus, 
\begin{align*}
	\eta_{i+r,r}^{\prime n} &= \frac{1}{\phi_{k_n}(g)} K_{{m}_{n} \Delta_n}(t_{i+r-1}-\tau)\left(\E \left(\left.\phi_{i+r, r}^{n}\right| \mathcal{F}_{(i-1) \Delta_{n}}\right)- \E \left(\left.\phi_{i+r, r}^{n}\right| \mathcal{F}_{(i+r-1) \Delta_{n}}\right)\right)\\
	&= \frac{1}{\phi_{k_n}(g)} K_{{m}_{n} \Delta_n}(t_{i+r-1}-\tau)\left(\mathbb{E}\left(\left.\Psi_{i+r, r}^{n} \right| \mathcal{F}_{(i-1) \Delta_{n}}\right)-\mathbb{E}\left(\left.\Psi_{i+r, r}^{n} \right| \mathcal{F}_{(i+r-1) \Delta_{n}}\right)\right).
\end{align*}	
Next, note that, by (\ref{eq:psi2}), we have 
\begin{equation*}
\begin{aligned}
&\E\left|\mathbb{E}\left(\left.\Psi_{i+r, r}^{n} \right| \mathcal{F}_{(i-1) \Delta_{n}}\right)-\mathbb{E}\left(\left.\Psi_{i+r, r}^{n} \right| \mathcal{F}_{(i+r-1) \Delta_{n}}\right) \right|^2\\
&\quad\leq \E\left( \mathbb{E}\left(\left.\Psi_{i+r, r}^{n} \right| \mathcal{F}_{(i+r-1) \Delta_{n}}\right)^2\right)\\
& \quad \leq \E\left( \mathbb{E}\left(\left.\left(\Psi_{i+r, r}^{n}\right)^2 \right|{\mathcal{F}_{(i+r-1) \Delta_{n}}}\right)\right)\leq C \Delta_n^{3/2}.
\end{aligned}
\end{equation*}
We can then deduce that for $r\neq l $,  
\begin{equation*}
\begin{aligned}
\E\left( \eta_{i+r,r}^{\prime n}\eta_{i+l,l}^{\prime n} \right) &\leq \sqrt{\E\left(  \eta_{i+r,r}^{\prime n}\right)^2 \E\left(  \eta_{i+l,l}^{\prime n}\right)^2 }\\
&\leq C \frac{1}{\phi_{k_n}(g)^2} {|K_{{m}_{n} \Delta_n}(t_{i+r-1}-\tau)| |K_{{m}_{n} \Delta_n}(t_{i+l-1}-\tau)|}  \Delta_n^{3/2}.
\end{aligned}
\end{equation*}
Therefore, denoting for simplicity $I_{i}=I(m, n, i)=(i-1 )(m+1) k_{n}+1$,
\begin{equation*}
\begin{aligned}
\mathbb{E} \left|\sum_{r=0}^{m k_{n}-1} \eta_{I(m,n,i)+r, r}^{\prime n}\right|^2 &\leq C \frac{1}{\phi_{k_n}^2(g)} \left(\sum_{r=0}^{m k_{n}-1}{|K_{{m}_{n} \Delta_n}(t_{I_{i}+r-1}-\tau)|} \right)^2\Delta_n^{3/2} \\
&\leq C \frac{1}{k_n^2}\frac{1}{\Delta_n^2}\left(\int_{t_{I_{i}-1}}^{t_{I_{i}+(mk_{n}-1)\Delta_{n}}} {|K_{{m}_{n} \Delta_n}\left(s-\tau\right)|} ds \right)^2 \Delta_n^{3/2}\\
& \leq C \Delta_n^{1/2}\left(\int_{t_{I_{i}-1}}^{t_{I_{i}+(mk_{n}-1)\Delta_{n}}} {|K_{{m}_{n} \Delta_n}\left(s-\tau\right)|} ds\right)^2.
\end{aligned}
\end{equation*}
The result is proved by the following:
\begin{align*}
{m}_{n}^{1/2}\Delta_n^{1/4} {\mathbb{E}\left|M^{\prime n}(m)\right|} &\leq C {m}_{n}^{1/2}\Delta_n^{1/4} \sum_{i=1}^{l_{n}(m)}  \Delta_n^{1/4}\left(\int_{t_{I_{i}-1}}^{t_{I_{i}}+(m k_n-1) \Delta_n} {|K_{{m}_{n} \Delta_n}\left(s-\tau\right)|} ds \right)  \\
&\leq C {m}_{n}^{1/2}\Delta_n^{1/2} \int |K(u)| du
 \rightarrow 0.
\end{align*}
\end{proof}

At this stage we are ready to prove a CLT for the processes $M_n(m)$, for each fixed $m$. We follow the arguments of \cite{JacodProtter} in page 550. For completeness, we {outline them} here.
Let 
\begin{equation} \label{eq:L}
L\left(g\right)_{t}=\int_{t}^{t+1} g(u-t) d W^{1}_{u}, \quad L^{\prime}\left(g\right)_{t}=\int_{t}^{t+1} g^{\prime}(u-t) d W_{u}^{2},
\end{equation} 
where $W^1$ and $W^2$ are two independent one-dimensional Brownian motions defined on an auxiliary space $\left(\tilde{\Omega}, \tilde{\mathcal{F}},\left(\tilde{\mathcal{F}}_{t}\right)_{t \geq 0}, \tilde{\mathbb{P}}\right)$.
The processes $L(g)$ and $L^{\prime}(g)$ are independent, stationary, centered, and Gaussian with covariance 
\begin{equation*}
\begin{aligned} \mathbb{E}\left(L\left(g\right)_{t} L\left(g\right)_{s}\right) &=\int_{t \vee s}^{(t+1) \wedge(s+1)} g(u-t) g(u-s) d u, \\ 
\mathbb{E}\left(L^{\prime}\left(g\right)_{t} L^{\prime}\left(g\right)_{s}\right) &=\int_{t\vee s}^{(t+1) \wedge(s+1)}g^{\prime}(u-t)g^{\prime}(u-s) d u .\end{aligned}
\end{equation*}
Next, {denoting $\widetilde{\mathbb{E}}$ the expectation with respect to $\tilde{\mathbb{P}}$}, let \begin{equation*}
\begin{aligned}
&\mu\left(v, v^{\prime}\right)=\widetilde{\mathbb{E}}\left(\left(v L\left(g\right)_{s}+v^{\prime} L^{\prime}\left(g\right)_{s}\right)^2 -  v^{\prime 2} \phi(g^{\prime})\right),\\
&\mu^{\prime}\left(v, v^{\prime} ; s, s^{\prime}\right)= \widetilde{\mathbb{E}}\left(\left(\left(v L\left(g \right)_{s}+v^{\prime} L^{\prime}\left(g \right)_{s}\right)^2 - v'^{2}  \phi(g^{\prime})\right)  \left(\left(v L\left(g \right)_{s^{\prime}}+v^{\prime} L^{\prime}\left(g \right)_{s^{\prime}}\right)^2-  v^{\prime 2} \phi(g^{\prime})\right) \right),\\
&R\left(v, v^{\prime}\right)=\int_{0}^{2}\left(\mu^{\prime}\left(v, v^{\prime} ; 1, s\right)-\mu\left(v, v^{\prime}\right) \mu\left(v, v^{\prime}\right)\right) d s.
\end{aligned}
\end{equation*}
As argued in the proof of Theorem 7.20 in  \cite{jacodaitsahalia}, one can show that
\[
	{\frac{1}{\theta} R\left(\sigma_{{t}}, \theta v_t  \right) = 4   \left(\Phi_{22} \sigma_{{t}}^{4}/\theta+2 \Phi_{12} \sigma_{{t}}^{2} \gamma_{t}\theta+\Phi_{11} \gamma_{{t}}^{2}\theta^{3}\right),}
\]
 where $v_{t}  = \sqrt{\gamma_{t}}$ is the conditional standard deviation for $\epsilon_{t} $.
 For a fixed $m$ and $t \in [0,T]$, let
 \begin{equation*}
\gamma(m)_{t}=m \mu\left(\sigma_{t}, \theta v_{t}\right), \quad \gamma^{\prime}(m)_{t}=\int_{0}^{m} d s \int_{0}^{m} d s^{\prime} \mu^{\prime}\left(\sigma_{t}, \theta v_{t} ; s, s^{\prime}\right).
\end{equation*}
\begin{lemma} \thlabel{y(m)}
Under \thref{kernel,noise,X_continuous}, for each $m \geq 1$, as $n\to\infty$, the process ${m}_{n}^{1/2}\Delta_n^{1/4}M^n (m)$ converges in law to a r.v. $\overline{Y}(m)$,  which conditionally on $\mathcal{F}$ is a centered Gaussian r.v. with variance
\begin{equation*}
\mathbb{E}\left(\left.\left(\overline{Y}(m)\right)^2 \right|\mathcal{F}\right)=\frac{1}{m+1}\frac{1}{\theta}  \left(\gamma^{\prime}(m)_{\tau}-\gamma(m)_{\tau}^2\right) \int K^2(u)du.
\end{equation*}
\end{lemma}

\begin{proof}
{Recall that 
\begin{align*}
&M^{n}(m)=\sum_{i=1}^{l_{n}(m)} \sum_{r=0}^{m k_{n}-1} \eta_{I(m,n,i)+r, r}^{n},
\end{align*}
where
\begin{align*}
\eta_{i+r, r}^{n} &= \frac{1}{\phi_{k_n}(g)} K_{{m}_{n} \Delta_n}\left(t_{i+r-1}-\tau\right)\left(\phi_{i+r, r}^{n} - \E \left(\left.\phi_{i+r, r}^{n}\right| \mathcal{F}_{(i-1) \Delta_{n}}\right)\right)\\
\phi_{i+r, r}^{n} &= (\sigma_{(i-1)\Delta_n}\overline{W}_{i+r}^n + \overline{\epsilon}_{i+r}^n )^2 - \frac{1}{2}\hat{\epsilon}_{i+r}^n,\\
\overline{W}_{i}^{n}&=
-\sum_{j=1}^{k_{n}} \left(g\left(\frac{j}{k_{n}}\right) - g\left(\frac{j-1}{k_{n}}\right)\right) W_{(i+j-2)\Delta_{n}}
\\
I(m, n, i)&=(i-1 )(m+1) k_{n}+1,\quad l_{n}(m)=\left[\frac{n-k_n+1}{(m+1) k_{n}}\right].
\end{align*}
For  $i=1,\cdots, l_n(m)$, let
 \begin{equation} \label{eq:eta_m}
 \eta(m)_i^n := {m}_{n}^{1/2}\Delta_n^{1/4}\sum_{r=0}^{m k_{n}-1} \eta_{I(m,n,i)+r, r}^{n},\quad \mathcal{G}_{i}^{n}= \mathcal{F}_{(I(m, n, i+1)-1) \Delta_{n}}.
 \end{equation}
For simplicity, we write $I_{i}=I(m,n,i)$. Note that $\eta(m)_{i}^{n}$ depends on 
\begin{align*}
	&\sigma_{(I_{i}-1)\Delta_{n}},W_{(I_{i}-1)\Delta_{n}},\dots,W_{(I_{i+1}-3)\Delta_{n}},\epsilon_{(I_{i}-1)\Delta_{n}},\dots,\epsilon_{(I_{i+1}-3)\Delta_{n}}.
\end{align*}
Therefore, $\eta(m)^{n}_{i}$ is $\mathcal{G}_{i}^{n}$-measurable and, furthermore, $\mathbb{E}[\eta(m)^{n}_{i}|\mathcal{G}_{i-1}^{n}]=0$.
We} will apply {Theorem 2.2.15 in \cite{JacodProtter}} 
to the martingale increments
 $\eta(m)_i^n, i=1, \cdots, l_{n}(m)$.

 By the Jensen type inequality $|\sum_{r=0}^{mk_{n}-1} a_{r}b_{r}|^{4}\leq{}(\sum_{r=0}^{mk_{n}-1} |a_{r}|)^{3}\sum_{r=0}^{mk_{n}-1} |a_{r}|b_{r}^{4}$ and (\ref{eq: phi_i_j}), we have, for each fixed $m$,
 \begin{equation} \label{eq:jensen}
 \begin{aligned}
 \sum_{i=1}^{l_{n}(m)}  \E\left(\left.\left|\eta(m)_i^n\right|^4\right|\mathcal{G}_{i-1}^{n} \right) 
&\leq  {C}\sum_{i=1}^{l_{n}(m)} {m}_{n}^2 \Delta_n^3 \left(\sum_{r=0}^{m k_n -1}\frac{1}{|\phi_{k_n}(g)|} |K_{{m}_{n} \Delta_n}\left(t_{I_{i}+r-1}-\tau\right)| \right)^4 \\
&\leq  C \sum_{i=1}^{l_{n}(m)} m^4 {m}_{n}^2 \Delta_n^3 \left(\frac{1}{m k_n \Delta_n} \int_{t_{I_{i}-1}}^{t_{I_{i}-1}+m k_n \Delta_n} |K_{{m}_{n} \Delta_n} \left(s - \tau\right)| d s\right)^4 \\
&\leq C \sum_{i=1}^{l_{n}(m)} m^4 {m}_{n}^2 \Delta_n^3 \frac{1}{m k_n \Delta_n} \int_{t_{I_i-1}}^{t_{I_i-1}+m k_n \Delta_n} K^4_{{m}_{n} \Delta_n} \left(s - \tau\right) d s\\
&\leq \mathrm{O}\left( {m}_{n}^2 \Delta_n^{5/2} \frac{1}{\left({m}_{n} \Delta_n\right)^3 } \int K^4(u) du  \right)= \mathrm{O}\left(\frac{1}{{m}_{n} \Delta_n^{1/2}}\right) \rightarrow 0.
 \end{aligned}
\end{equation}
Therefore, for every $\varepsilon >0$,
\begin{equation*}
\begin{aligned}
&\sum_{i=1}^{l_{n}(m)}  \E\left(\left.\left|\eta(m)_i^n\right|^2 \mathbbm{1}_{\left|\eta(m)_{i}^{n}\right|^{2} \geq \varepsilon} \right|\mathcal{G}_{i-1}^{n} \right)  \leq 
\frac{1}{\epsilon} \sum_{i=1}^{l_{n}(m)} \E\left(\left.\left|\eta(m)_i^n\right|^4 \right|\mathcal{G}_{i-1}^{n} \right) \,\stackrel{n\to\infty}{\longrightarrow}\, 0.
\end{aligned}
\end{equation*}
It remains to prove that, for a fixed $m$,
\begin{equation} \label{eq:variance}
{S_{n}}:=\sum_{i=1}^{l_{n}(m)} \mathbb{E}\left(\left.\left(\eta(m)_{i}^{n}\right)^2  \right|\mathcal{G}_{i-1}^{n}\right)   \stackrel{\mathbb{P}}{\longrightarrow} \frac{1}{(m+1)\phi^2(g) }\frac{1}{\theta} \int K^2(u)du \left(\gamma^{\prime}(m)_{\tau}-\gamma(m)_{\tau}^2\right),
\end{equation}
 and also, for any bounded \(\mathcal{F}_{t}\)-martingale $N$ that is orthogonal to $W$, or for $N=W$,
\begin{equation}\label{eq:martingale}
\sum_{i=1}^{l_{n}(m)} \mathbb{E}\left(\eta(m)_{i}^{n}\left(N_{(I_{i+1}-1) \Delta_{n}}-N_{(I_i-1) \Delta_{n}}\right) \,\middle\vert\,{G}_{i-1}^{n}\right) \stackrel{\mathbb{P}}{\rightarrow} 0.
\end{equation}
  We start by proving (\ref{eq:variance}). 
Let
\begin{equation} \label{eq:gamma}
\begin{aligned}
\alpha^n_i & := \frac{1}{k_n^2 \Delta_n}\sum_{r=0}^{m k_{n}-1}  K_{{m}_{n} \Delta_n}\left( t_{I_i+r-1} - \tau\right) \phi^{n}_{I_i+r, r}. \\
& =\frac{1}{k_n^2 \Delta_n}K_{{m}_{n} \Delta_n}\left( t_{I_i-1} - \tau\right) \sum_{r=0}^{m k_{n}-1}  \phi_{I_i+r, r}^{n}  + \mathrm{O}_P\left( \frac{1}{{m}_{n}^2 \Delta_n^{3/2}}\right) ,
\end{aligned}
\end{equation}
where for the second equality above we applied \thref{kernel} and  (\ref{eq:phi_bound}) to show \begin{equation*}
\begin{aligned}
 &\frac{1}{k_n^2 \Delta_n}\sum_{r=0}^{m k_{n}-1}  \left| K_{{m}_{n} \Delta_n}\left( t_{I_i+r-1} - \tau\right)  - K_{{m}_{n} \Delta_n}(t_{I_i-1} - \tau ) \right|\E\left|\phi_{I_i+r, r}^{n}\right|\\
 &\quad\leq  C \frac{1}{k_n^2 \Delta_n}\sum_{r=0}^{m k_{n}-1} \frac{1}{{m}_{n} \Delta_n}\frac{mk_n\Delta_n}{{m}_{n} \Delta_n}\Delta_n^{1/2}=O\left(\frac{1}{{m}_{n}^2\Delta_n^{3/2}}\right).
\end{aligned}
\end{equation*}
For $(I(m,n,i)-1)\Delta_n\leq s<(I(n,m,i+1)-1)\Delta_n$, set
\begin{equation*}
\begin{aligned}
\gamma_s^{n} =   \E\left(\left.\frac{1}{k_n^2 \Delta_n}\sum_{r=0}^{m k_{n}-1}  \phi_{I_i+r, r}^{n} \right|{\mathcal{G}^n_{i-1}} \right),
\quad 
\gamma_{s}^{\prime n} =   \E\left(\left.\left( \frac{1}{k_n^2 \Delta_n}\sum_{r=0}^{m k_{n}-1}  \phi_{I_i+r, r}^{n}   \right)^2\right|{\mathcal{G}^n_{i-1}}\right).
\end{aligned}
\end{equation*}
Then, we have
\begin{equation*} 
\begin{aligned}
{S_{n}}
 &={m}_{n} \Delta_n^{1/2}\frac{k_n^4\Delta_n^2}{\phi_{k_n}^2(g)}\sum_{i=1}^{l_n(m)} \left( \E\left(\left.\left(\alpha_i^n\right)^2\right| \mathcal{G}_{i-1}\right) - \left(\E\left(\left.\alpha_i^n\right| \mathcal{G}_{i-1}\right)\right)^2\right)\\
 & = {m}_{n} \Delta_n^{1/2}\frac{k_n^4\Delta_n^2}{\phi_{k_n}^2(g)} \sum_{i=1}^{l_n(m)} K^2_{{m}_{n} \Delta_n}(t_{{I_i-1}} - \tau ) \left( \gamma^{\prime n}_{t_{{I_i-1}}} - \left(\gamma_{t_{{I_i-1}}}^n \right)^2\right) + \mathrm{O}_{P}\left( \frac{1}{{m}_{n} \sqrt{\Delta_n}}\right).
\end{aligned}
\end{equation*}
If we can show that for any $s \in [0,T]$,
\begin{equation} \label{eq:gamma}
\gamma_{s}^{n} \stackrel{\mathbb{P}}{\longrightarrow} \gamma(m)_{s}, \quad \gamma_{s}^{\prime n}\stackrel{\mathbb{P}}{\longrightarrow} \gamma^{\prime}(m)_{s},
\end{equation}
we can obtain (\ref{eq:variance}):
\begin{equation*}\begin{aligned}
{S_{n}}&= \frac{1}{(m+1)k_n \Delta_n} {m}_{n} \Delta_n^{1/2} k_n^2 \Delta_n^2 \int_0^T K^2_{{m}_{n} \Delta_n}(s- \tau ) \left(\gamma^{\prime}(m)_{s}-\gamma(m)_{s}^2\right) ds+ \mathrm{o}_P(1)\\
&\quad = \frac{1}{\theta(m+1)}\int_{\frac{-\tau}{{m}_{n} \Delta_n}}^{\frac{T-\tau}{{m}_{n} \Delta_n}} K^2(u) \left(\gamma^{\prime}(m)_{\tau + u {m}_{n} \Delta_n}-\gamma(m)_{\tau + u {m}_{n} \Delta_n}^2\right) d u + \mathrm{o}_P(1)\\
&\quad \stackrel{\mathbb{P}}{\longrightarrow} \frac{1}{\theta(m+1)}\int K^2(u) du\left(\gamma^{\prime}(m)_{\tau}-\gamma(m)_{\tau }^2\right) ,
\end{aligned}
\end{equation*}
where the last line can be shown as follows. For all $ \epsilon > 0$, there exists an interval $I = [a,b]$ such that $\int_{I^c} K^2(u) du \leq  \epsilon$. Let $I_n =[\frac{-\tau}{{m}_{n} \Delta_n},\frac{T-\tau}{{m}_{n} \Delta_n}]  $, $f_n(u) = K^2(u)\left(\gamma^{\prime}(m)_{\tau + u {m}_{n} \Delta_n}-\gamma(m)_{\tau + u {m}_{n} \Delta_n}^2\right) $ and $f(u)= K^2(u) du\left(\gamma^{\prime}(m)_{\tau}-\gamma(m)_{\tau }^2\right)  $. Then, we have for some constant $C$,
\begin{equation*}
\begin{aligned}
&\limsup_{n\rightarrow 0} \left|\int_{I_n} f_n(u)d u - \int f(u) du    \right|\\
&\leq \limsup_{n\rightarrow 0} \int_{I} \left| f_n(u) - f(u)\right| du + \int_{I_n\cap I^c} \left| f_n(u) \right|du + \int_{I^c}\left| f(u) \right|du\\
&\leq C \epsilon,
\end{aligned}
\end{equation*}
since $\gamma, \gamma^{\prime}  $ are continuous and bounded and $K$ is bounded. The result follows by letting $\epsilon \rightarrow 0$.

To show (\ref{eq:gamma}), we fix $s \in [0, T]$ and apply Lemma 16.3.9 in \cite{JacodProtter} with the sequence $i_n = I(m,n,i)$,  $T_n = (I(m,n,i)-1)\Delta_n $ if $ I(m,n,i-1)\Delta_n \leq s < I(m,n,i) \Delta_n $. Concretely, with the notation
\begin{equation*}
L_{u}^{n} =\frac{1}{\sqrt{k_n \Delta_n}} \overline{W}_{i_n+\left[k_{n} u\right]}^{n}, 
\quad
L_{u}^{\prime n}=\sqrt{k_{n}} \overline{\epsilon}_{i_n+\left[k_{n} u\right]}^{n} ,
\quad
\widehat{L}_{u}^{n} =k_{n} \widehat{\epsilon}_{i_n+\left[k_{n} u\right]}^{n}, 
\end{equation*}
 for $u \in [0,m]$,  we have
\begin{equation*}
\frac{1}{k_n^2 \Delta_n}\sum_{r=0}^{m k_{n}-1}  \phi_{i+r, r}^{n} =  F_{n}\left(\sigma_{T_{n}} L^{n}, L^{\prime n}, \widehat{L}^{n}  \right),
\end{equation*}
where $F_n$ is the function  on $\mathbb{D}\times \mathbb{D}\times\mathbb{D}$ (here $\mathbb{D} = \mathbb{D}\left([0,m] : \mathbb{R}^{1}\right)$ is the Skorokhod space), defined by 
\begin{equation} \label{eq:F_n-notation}
F_{n}(x, y, z)=\frac{1}{k_{n}} \sum_{r=0}^{m k_{n}-1}\left(\left(x\left(\frac{r}{k_{n}}\right)+\frac{1}{\sqrt{k_{n}^2 \Delta_n }} y\left(\frac{r}{k_{n}}\right)\right)^2 -  \frac{1}{2k_{n}^2\Delta_n } z\left(\frac{r}{k_{n}}\right)\right) .
\end{equation}
Note that the functions $F_n, F_n^2$ converge pointwise to $F, F^2$, respectively, where  \begin{equation*}
F(x, y, z)=\int_{0}^{m}\left\{\left(x\left(s\right)+\theta y\left(s\right)\right)^2 - \frac{1}{2}\theta^{2} z\left(s\right)\right\}d s.
\end{equation*}
Now we deduce from Lemma 16.3.9 in \cite{JacodProtter} that with $Z=1$, $\phi(f) = \int_0^1 f^2(u) du $ and the  notation from (\ref{eq:L})\footnote{Below, we assume that the space $\left(\tilde{\Omega}, \tilde{\mathcal{F}}, \tilde{\mathbb{P}}\right)$, where $W^{1}$ and $W^{2}$ (hence, $L$ and $L'$) are defined, is an extension of the space  $\left(\Omega, \mathcal{F},\mathbb{P}\right)$ and that $W^1$ and $W^{2}$ are independent of $X$ and $\epsilon$.}:
\begin{equation*}\begin{aligned}
\E \left(\left.F_n\left(\sigma_{T_{n}} L^{n}, L^{\prime n}, \widehat{L}^{n}\right)\right|\mathcal{G}_{(i-1) }\right) &\stackrel{\mathbb{P}}{\rightarrow}\E \left( F\left(\sigma_{s} L,  v_s L^{\prime}, 2 \phi(g^{\prime}) \gamma_s \right)\right)
= \gamma(m)_{s}.
\end{aligned}
\end{equation*}
Similarly, 
\begin{equation*}
\E \left(\left.F_n^2\left(\sigma_{T_{n}} L^{n}, L^{\prime n}, \widehat{L}^{n}\right)\right|\mathcal{G}_{(i-1) }\right)\stackrel{\mathbb{P}}{\rightarrow}\gamma(m)_{s}^{\prime},
\end{equation*}
and we conclude 
(\ref{eq:gamma}). This finishes the proof for (\ref{eq:variance}).
 Now we show (\ref{eq:martingale}).  Let $$\zeta_{i}^n=  \frac{{m}_{n}^{1/2} \Delta_n^{1/4}}{\phi_{k_n}(g)}  \sum_{r=0}^{m k_{n}-1}  K_{{m}_{n} \Delta_n}\left( t_{I_i+r,r} - \tau\right) \phi_{I_i+r, r}^{n},$$ 
 and set $D_i^n(N) = N_{(I_{i+1}-1) \Delta_{n}}-N_{(I_{i}-1) \Delta_{n}} $. Since $\mathbb{E}\left(D(N)_{i}^{n} | \mathcal{G}_{i-1}^{n}\right)=0$, we only need to prove that, for any bounded martingale $N$,
\begin{equation} \label{eq:martingale_simple}
\sum_{i=1}^{l_{n}(m)} \mathbb{E}\left(\zeta_i^n D_i^n(N)\,\middle\vert\,{G}_{i-1}^{n}\right) \stackrel{\mathbb{P}}{\rightarrow} 0.
\end{equation}
Following the same argument of (\ref{eq:jensen}) and inequality (\ref{eq: phi_i_j}),  we have \begin{equation}\begin{aligned} \label{eq:zeta2}
 \sum_{i=1}^{l_{n}(m)}\E\left( \zeta_i^n\right) ^2 &=   \frac{{m}_{n} \Delta_n^{1/2}}{\phi_{k_n}^2(g)} \sum_{i=1}^{l_{n}(m)} \E\left( \left(\sum_{r=0}^{m k_n -1} K_{{m}_{n} \Delta_n}\left(t_{I_i+r-1}-\tau\right) \left|\phi_{I_i+r, r}^n\right|\right)^2  \,\middle\vert\,{G}_{i-1}\right)\\
&\leq   C\sum_{i=1}^{l_{n}(m)}\frac{{m}_{n} \Delta_n^{3/2}}{\phi_{k_n}^2(g)}\left(\sum_{r=0}^{m k_n -1}K_{{m}_{n} \Delta_n}\left(t_{I_i+r-1}-\tau\right) \right)^2\\
&\leq  C\frac{{m}_{n} \Delta_n^{3/2}}{\phi_{k_n}^2(g)} \sum_{i=1}^{l_{n}(m)} m^2 k_n^2 \left(\frac{1}{m k_n \Delta_n} \int_{t_{I_i-1}}^{t_{i-1}+m k_n \Delta_n} K_{{m}_{n} \Delta_n} \left(s - \tau\right) d s\right)^2 \\
&\leq C \sum_{i=1}^{l_{n}(m)} m^2 {m}_{n} \Delta_n^{3/2} \frac{1}{m k_n \Delta_n} \int_{t_{I_i-1}}^{t_{I_i-1}+m k_n \Delta_n} K^2_{{m}_{n} \Delta_n} \left(s - \tau\right) d s\\
&\leq \mathrm{O}\left(\int K^2(u) du  \right) = \mathrm{O}\left(1\right).
\end{aligned}
\end{equation}
If N is a square-integrable martingale, the  Cauchy-Schwarz inequality yields, 
\begin{equation*}\begin{aligned}
\sum_{i=1}^{l_{n}(m)} \mathbb{E}\left(\zeta_i^n D_i^n(N)\,\middle\vert\,{G}_{i-1}^{n}\right) &\leq  \sqrt{\left(\sum_{i=1}^{l_{n}(m)} \E\left( \zeta_i^n\right) ^2 \right) \left(  \sum_{i=1}^{l_{n}(m)} \E\left( D_i^n(N)\right) ^2\right)}\\
 &\leq C  \sqrt{\E N_T^2}.
\end{aligned}
\end{equation*}
Note with notation (\ref{eq:notation}) and $$\zeta_{i}^{\prime n}=  \frac{{m}_{n}^{1/2} \Delta_n^{1/4}}{\phi_{k_n}(g)}  \sum_{r=0}^{m k_{n}-1}  K_{{m}_{n} \Delta_n}\left( t_{I_i+r,r} - \tau\right) \Psi_{I_i+r, r}^{n},$$the same argument also yields

\begin{equation} \label{eq:mart_bounded}
\E\left(\zeta_{i}^{\prime n} D^n_i(N) |\mathcal{G}_{i-1}\right) \leq C \Delta_n^{1/4}\sqrt{\E N^2_T}. 
\end{equation}

As shown in page 552 of \cite{JacodProtter}, we just need to prove (\ref{eq:martingale_simple}) for $N \in \mathcal{N}^{(i)}, i=0,1$, where $\mathcal{N}^{(0)}$ is the set of all bounded $\left(\mathcal{F}_{t}^{(0)}\right)$-martingales orthogonal to $W$ and $\mathcal{N}^{(1)}$ is the set of all martingales having $N_{\infty}=h\left(\chi_{t_{1}}, \ldots, \chi_{t_{w}}\right)$, where \(h\) is a Borel bounded function on \(\mathbb{R}^{w}\) and \(t_{1}<\cdots<t_{w}\) and \(w \geq 1\). 
When \(N\) is either $W$ or in   \(\mathcal{N}^{(0)}\), \(D(N)_{i}^{n}\) is \(\mathcal{H}_{\infty}\) measurable. Therefore \(\mathbb{E}\left(\zeta_{I_i}^{n} D(N)_{i}^{n} | \mathcal{G}_{i-1}^{n}\right)\) is equal to \begin{equation*}
\mathbb{E}\left(\zeta_{I(m, n, i)}^{\prime n} D(N)_{i}^{n} | \mathcal{G}_{i-1}^{n}\right)+\frac{{m}_{n}^{1/2} \Delta_n^{1/4}}{\phi_{k_n}(g)}  \mathbb{E}\left(\left.\sum_{r=0}^{m k_{n}-1} K_{{m}_{n} \Delta_n}\left( t_{I_i+r,r} - \tau \right)\left(\sigma_{(I_i-1)\Delta_n}\overline{W}^n_{I_i+r}\right)^2D(N)_{i}^{n} \right| \mathcal{G}_{i-1}^{n}\right).
\end{equation*}
The second term vanishes when $N=W$ since it is the $\mathcal{F}_{(I_i-1)\Delta_n}$-conditional expectation of an odd function of the increments of the process \(W\) after time \((I_i-1) \Delta_{n}\). Suppose now that \(N\) is a bounded martingale, orthogonal to \(W\). By It\^o's formula we see that \(\left(\overline{W}_{j}^{n}\right)^{2}\) is the sum of a constant (depending on \(n )\) and of a martingale which is a stochastic integral with respect to \(W,B\) on the interval \(\left[(j-1) \Delta_{n},\left(j+k_{n}-1\right) \Delta_{n}\right]\). Then the orthogonality of \(N\) and \(W\) implies this second term above vanishes as well. So in view of (\ref{eq:mart_bounded}), we have the following inequality which 	implies the result: \[\mathbb{E}\left(\zeta_{I_i}^{n} D(N)_{i}^{n} | \mathcal{G}_{i-1}^{n}\right) \leq C \Delta_n^{1/4}\sqrt{\E N^2_T}.  \]
When \(N \in \mathcal{N}^{(1)}\) is associated with \(h\) and \(w\) and the \(t_{i}\)'s, 
the same argument in \cite{JacodProtter}  and the inequality  \(\E \left(\zeta_i^n\right)^2 \leq C\frac{1}{{m}_{n} \sqrt{\Delta_n}} \) deduced from (\ref{eq:zeta2}) yield \begin{equation*}
\mathbb{E}\left(\sum_{i=1}^{l_{n}(m)}\left|\mathbb{E}\left(\zeta_{I_i}^{n} D(N)_{i}^{n} | \mathcal{G}_{i-1}^{n}\right)\right|\right) \leq C w\left(\Delta_n^{1/4}+\frac{1}{{m}_{n} \sqrt{\Delta_n}}\right),
\end{equation*}
and (\ref{eq:martingale}) is shown. This finishes the proof for \thref{y(m)}.
\end{proof}
	
 The only thing left to prove {\thref{part1}} is the stable convergence in law $\overline{Y}(m) \stackrel{st }{\longrightarrow} Z_{\tau} $, as $m \rightarrow \infty$. For this, we only need to show that, for each {$\tau \in (0,T)$, as $m\to\infty$,}
 \begin{equation*}
 {\frac{1}{m+1}\left(\gamma^{\prime}(m)_{\tau}-\gamma(m)_{\tau}^2\right) \stackrel{st }{\longrightarrow}   R\left(\sigma_{\tau}, \theta v_{\tau}\right)}.
 \end{equation*}
Recall that the process $(L,L^{\prime})$ is stationary, and the variables $(L_t,L^{\prime}_t)$ and $(L_s,L^{\prime}_s)$ are independent if $|s - t| \geq  1$. So $\mu^{\prime}\left(v, v^{\prime} ; s, s^{\prime}\right)  = \left(\mu\left(v, v^{\prime}\right)\right)^2$ when $|s-s^{\prime}|\geq 1$ and  $\mu^{\prime}\left(v, v^{\prime} ; s, s^{\prime}\right)=\mu^{\prime}\left(v, v^{\prime} ; 1, s^{\prime}+1-s\right)$, for all $s,s^{\prime} \geq 0$ with $s^{\prime} +1-s \geq 0$. 
 Then if $m\geq 2$ and {letting} $\mu=\mu\left(\sigma_{\tau}, \theta v_{\tau}\right)$ and $\mu^{\prime}\left(s, s^{\prime}\right)=\mu^{\prime}\left(\sigma_{\tau},  \theta v_{\tau} ; s, s^{\prime}\right),$ we have 
 \begin{equation*}
 \begin{aligned}
 &\frac{1}{m+1}\left(\gamma^{\prime}(m)_{\tau}-\gamma(m)_{\tau} \gamma(m)_{\tau}  \right) \\
 &= \frac{1}{m+1}\int_{0}^{m} d s \int_{0}^{m} \mu^{\prime}\left(s, s^{\prime}\right) d s^{\prime}-m^{2} \mu^2\\
 & = \frac{1}{m+1}\int_{0}^{m} d s \int_{(s-1)^{+}}^{m \wedge(s+1)}\left(\mu^{\prime}\left(1, s^{\prime}+1-s\right)-\mu^2\right) d s^{\prime}\\
 & =\frac{m-1}{m+1}  \int_{0}^{2}\left(\mu^{\prime}\left(1, s^{\prime}\right)-\mu^2\right) d s^{\prime} + \frac{1}{m+1}\int_{0}^{1} d s \int_{1-s}^{2}\left(\mu^{\prime}\left(1, s^{\prime}+1-s\right)-\mu^2\right) d s^{\prime}\\
 &\rightarrow R\left(\sigma_{\tau}, \theta v_{\tau}\right),
 \end{aligned}
 \end{equation*}
 since $\mu, \mu^{\prime}$ are bounded. This finishes the proof for \thref{part1}.
 
\subsubsection*{Step 3}
We now show \thref{part2}.
\begin{proof}[Proof of \thref{part2}]
Let $b_n= {m}_{n} \Delta_n$ and $t(i) = (I(m,n,i)-1)\Delta_n $, {where the notation for $I(m,n,i)$ can be found after step 2 above}. From the proof of Theorem 6.2 in \cite{FigLi} {and recall we have bounded jumps}, we have 
\begin{equation*}
b_n^{-1/2}\int_0^T K_{b_n}\left(t-\tau\right)\left(\sigma^2_t - \sigma^2_{\tau}\right) dt = b_n^{-1/2} \Lambda_{\tau-\sqrt{b_n}} \int_{\tau-\sqrt{b_n}}^{T} {L}\left(\frac{t-\tau}{b_n}\right) d B_{t}+o_{P}(1),
\end{equation*} 
where ${L}(t)=\int_{t}^{\infty} K(u) d u \mathbbm{1}_{\{t>0\}}-\int_{-\infty}^{t} K(u) d u \mathbbm{1}_{\{t \leq 0\}} $.
Also we have 
\begin{equation*}
\begin{aligned}
b_n^{-1/2}\int_{(0,T)^{c}} K_{b_n}(t-\tau)d t & = 
\frac{1}{\sqrt{b_n}}\left( \int_{-\infty}^{\frac{-\tau}{b_n}}K\left( u\right) du  + \int_{\frac{T-\tau}{b_n}}^{\infty}K(u) du  \right) \rightarrow 0, \text{ as } n \rightarrow \infty
\end{aligned}
\end{equation*}
since \thref{kernel} imply that $x^{1/2} \int_x^{\infty} K(u) du \rightarrow 0$, as {$x\rightarrow \infty$}.
So, for a fixed $m$, we can rewrite ${m}_{n}^{1 / 2} \Delta_{n}^{1 / 4}\overline{H}(2)^n $ as
\begin{equation}\begin{aligned} \label{eq:alpha_notation}
 & \beta b_n^{-1/2}\int_{0}^{T} K_{b_n}(t-\tau)\left(\sigma_{t}^{2}-\sigma_{\tau}^{2}\right) d t+{o_{P}(1)}\\
&= \beta b_n^{-1/2}\Lambda_{\tau-\sqrt{b_n}}\sum_{i:t(i)>\tau -\sqrt{b_n}}^{l_{n}(m)}  \int_{t(i)}^{t(i)+(m+1) k_n\Delta_n}{L}\left(\frac{t-\tau}{b_n}\right) d B_{t}+\mathrm{o}_{P}(1)\\
&=:\sum_{i=1}^{l_{n}(m)}\alpha(m)_i^n+\mathrm{o}_{P}(1),
\end{aligned}
\end{equation}with $\alpha(m)^n_i = 0 $ if $i$ is such that {$t(i)\leq{}\tau -\sqrt{b_n} $}.

Combine with the proof of \thref{part1}, we can deduce the following lemma.
\begin{lemma}
Under \thref{kernel,noise} and (\ref{eq:sigma}),  with ${m}_{n}\to\infty$ and ${m}_{n}\Delta_n^{3/4} \rightarrow \beta \in (0, \infty)$,  
\begin{equation*}
\lim_{m\to \infty }\limsup_{n \to \infty}  {m}_{n}^{1 / 2} \Delta_{n}^{1 / 4} \E\left|\overline{H}(1)^n+ \overline{H}(2)^n - \sum_{i=1}^{l_{n}(m)}\zeta(m)_{i}^{n}  -\sum_{i=1}^{l_{n}(m)}\alpha(m)_i^n  \right| = 0,
\end{equation*}
where $\zeta(m)_i^n := {m}_{n}^{1/2}\Delta_n^{1/4}\frac{1}{\phi_{k_n}(g)}K_{b_n}\left(t(i) -\tau\right) \sum_{r=1}^{m k_n -1} \phi_{I_i+r,r}^n$ with notation (\ref{eq:notation}).
\end{lemma}


Now \thref{part2} follows if we apply Theorem 2.2.15 in \cite{JacodProtter} to  the
sum of martingale differences $\left(\zeta(m)_{i}^{n}+\alpha(m)_i^n\right) $ and the filtration $\mathcal{G}_{i} = \mathcal{F}_{\left(I_{i+1}-1\right)\Delta_n}$, and show that 
\begin{equation*}
\sum_{i=1}^{l_{n}(m)} \left(\zeta(m)_{i}^{n}+\alpha(m)_i^n\right) \stackrel{st}{\longrightarrow} Z_{\tau} + \beta Z^{\prime}_{\tau}.
\end{equation*}
To this end, we first need to show, for a fixed $m$,
\begin{align}
&\sum_{i=1}^{l_{n}(m)} \E \left(\left(\zeta(m)_{i}^{n} \right)^2 |\mathcal{G}_{i-1}\right) \rightarrow  \frac{1}{m+1}\frac{1}{\theta}  \left(\gamma^{\prime}(m)_{\tau}-\gamma(m)_{\tau}^2\right) \int K^2(u)du, \label{eq: part1}\\ 
&\sum_{i=1}^{l_{n}(m)} \E \left(\left(\alpha(m)_i^n\right)^2 |\mathcal{G}_{i-1}\right) \rightarrow 
\beta^2 \Lambda_{\tau}^2\int {L}^2(u)du, \label{eq:part2} \\
&\sum_{i=1}^{l_{n}(m)} \E \left(\left(\zeta(m)_{i}^{n} \alpha(m)_i^n\right) |\mathcal{G}_{i-1}\right) \rightarrow 0. \label{eq: cross}
\end{align}
The proof of (\ref{eq: part1}) can be found in the proof for \thref{part1}. (\ref{eq:part2}) can be directly derived from the definition (\ref{eq:alpha_notation}): \begin{equation*}
\begin{split}
\sum_{i=1}^{l_{n}(m)} \E \left(\left(\alpha(m)_i^n\right)^2 |\mathcal{G}_{i-1}\right) & =\beta^2 b_n^{-1}\Lambda^2_{\tau-\sqrt{b_n}}\sum_{i:t(i)>\tau -\sqrt{b_n}}^{l_{n}(m)}  \int_{t(i)}^{t(i)+(m+1) k_n\Delta_n}{L}^2\left(\frac{t-\tau}{b_n}\right) dt \rightarrow \beta^2 \Lambda_{\tau}^2\int {L}^2(u)du.
\end{split}
\end{equation*}
So we only need to show  (\ref{eq: cross}). 
With the notation (\ref{eq:notation}), we have \begin{equation*}
\begin{aligned}
&\E \left(\left. \left(\sum_{r=0}^{m k_n -1}\phi_{I_i+r,r}^n\right) \int_{t(i)}^{t(i)+(m+1) k_n\Delta_n}{L}\left(\frac{t-\tau}{b_n}\right) d B_{t}\right|\mathcal{G}_{i-1}\right)\\
& = \E \left(\left.\int_{t(i)}^{t(i)+(m+1) k_n\Delta_n}{L}\left(\frac{t-\tau}{b_n}\right) d B_{t}\E\left(\left.\left(\sum_{r=0}^{m k_n -1}\phi_{I_i+r,r}^n\right)  \right|\mathcal{H}_{t(i)} \right) \right|\mathcal{G}_{i-1}\right)\\
& =  \sigma_{t(i)}^2\E\left(\left. \left(\sum_{r=0}^{m k_n-1}\left(\overline{W}^n_{t(i)+r}\right)^2 \right) \int_{t(i)}^{t(i)+(m+1) k_n\Delta_n}{L}\left(\frac{t-\tau}{b_n}\right) d B_{t}\right|\mathcal{G}_{i-1} \right)\\
&\qquad + \E\left(\left.\left(\sum_{r=0}^{m k_n-1}\Psi_{t(i)+r,r}\right) \int_{t(i)}^{t(i)+(m+1) k_n\Delta_n}{L}\left(\frac{t-\tau}{b_n}\right) d B_{t}\right|\mathcal{G}_{i-1} \right)\\
& :=A_i+B_i.
\end{aligned}
\end{equation*}
 Let $U_{i,r}^s = \int_{t(i)+r\Delta_n}^s g_n\left(\frac{u-(t(i)+r\Delta_n)}{k_n \Delta_n}\right) d W_u$, $g_n(t) = \sum_{r = 1}^{k_n}g\left(\frac{r}{k_n}\right)\mathbbm{1}_{\left[\frac{(r-1)\Delta_n}{k_n \Delta_n}, \frac{r\Delta_n}{k_n \Delta_n}, \right]}(t)$. By It\^o lemma, we have when $t(i) > \tau -\sqrt{b_n} $,  
 \begin{equation*}
\begin{aligned}
A_i =& \frac{1}{k_n^2\Delta_n}\sigma_{t(i)}^2\E \left( \left.\sum_{r=0}^{m k_n-1} \int_{t(i)+r\Delta_n}^{t(i)+(r+k_n)\Delta_n} U_{i,r}^s g_n\left(\frac{s-(t(i)+r\Delta_n)}{k_n \Delta_n}\right)  d W_s \int_{t(i)}^{t(i)+(m+1) k_n\Delta_n} {L}\left(\frac{s-\tau}{b_n}\right)d B_s\right|\mathcal{G}_{i-1} \right)\\
&= \frac{1}{k_n^2\Delta_n}\sigma_{t(i)}^2\E \left( \left.\sum_{r=0}^{m k_n-1} \int_{t(i)+r\Delta_n}^{t(i)+(r+k_n)\Delta_n} U_{i,r}^s g_n\left(\frac{s-(t(i)+r\Delta_n)}{k_n \Delta_n}\right)  {L}\left(\frac{s-\tau}{b_n}\right) \rho_s d s\right|\mathcal{G}_{i-1} \right) = 0,
\end{aligned}
\end{equation*}
since $\E\left( \left.U_{i,r}^s\right|\mathcal{G}_{i-1}\right) = 0.$


As for $B_i$, we can apply Cauchy-Swacharz inequality. By (\ref{eq:psi2}) and the boundedness of ${L}$,
\begin{equation*}
\begin{aligned}
B_i^2 &\leq\E\left(\left.\left(\sum_{r=0}^{m k_n-1}\Psi_{t(i)+r,r}\right)^2\right|\mathcal{G}_{i-1}\right) \int_{t(i)}^{t(i)+(m+1) k_n\Delta_n}{L}^2\left(\frac{s-\tau}{b_n}\right) d s \\
&\leq C (m k_n)^2 \Delta_n^{3/2} (m+1) k_n \Delta_n \leq C \Delta_n.
\end{aligned}
\end{equation*}
Finally, we can show
\begin{equation*}
\begin{aligned}
\sum_{i=1}^{l_{n}(m)} \E \left(\left(\zeta(m)_{i}^{n} \alpha(m)_i^n\right) |\mathcal{G}_{i-1}\right) &=\sum_{i:t(i) >\tau - \sqrt{b_n}}^{l_{n}(m)}   C b_n  K_{b_n}(t(i) - \tau ) \Lambda_{\tau - \sqrt{b_n} }\left(A_i+B_i\right) \\
  &\leq C b_n \sum_{i=1}^{l_n(m)} \left| K_{b_n}(t(i) - \tau )\right| \left(\Delta_n^{1/2}  \right)  = \mathrm{O}\left({m}_{n} \Delta_n\right)\rightarrow 0.
\end{aligned}
\end{equation*}
Now we single out a two dimension Brownian motion $\tilde{W} = (W, B)$, and a subset  $\mathcal{N}$ of bounded martingales, all orthogonal to $\tilde{W}$.  Let \(D_{i}^{n}(N)=N_{\left(I_{i+1}-1\right) \Delta_{n}}-N_{\left(I_{i}-1\right) \Delta_{n}}\). We need to prove \begin{equation*}
\sum_{i=1}^{l_{n}(m)} \mathbb{E}\left(\left(\zeta(m)_{i}^{n} + \alpha(m)_i^n\right) D_{i}^{n}(N) | \mathcal{G}_{i-1}^{n}\right) \stackrel{\mathbb{P}}{\rightarrow} 0,
\end{equation*}
whenever N is one of the component of \(\tilde{W}\) or is in the set $\mathcal{N}$.
Since \([W_t,B_t] \leq [W_t,W_t] = t \), we can deduce \( \sum_{i=1}^{l_{n}(m)} \mathbb{E}\left(\left(\zeta(m)_{i}^{n} + \alpha(m)_i^n\right) D_{i}^{n}(N) | \mathcal{G}_{i-1}^{n}\right) \stackrel{\mathbb{P}}{\rightarrow} 0,\) for the same reason as in proving (\ref{eq:martingale}). 

Next, 
\begin{equation*}
\sum_{i=1}^{l_{n}(m)} \E \left(\left(\zeta(m)_{i}^{n} +\alpha(m)_i^n\right)^4 |\mathcal{G}_{i-1}\right)  \stackrel{\mathbb{P}}{\rightarrow} 0
\end{equation*}
can be easily deduced by straightforward computation and (\ref{eq:jensen}).\\
Thus, let $m \to \infty$, we can conclude  \({m}_{n}^{1/2}\Delta_n^{1/4}\left(\overline{H}(1)^n + \overline{H}(2)^n \right) \)  converges stably in law to a random variable defined on a good extension \(\left(\tilde{\Omega}, \tilde{\mathcal{F}},\left(\tilde{\mathcal{F}}_{t}\right)_{t>0}, \tilde{\mathbb{P}}\right)\) of the space \(\left(\Omega, \mathcal{F},\left(\mathcal{F}_{t}\right)_{t \geq 0}, \mathbb{P}\right)\), and conditionally on \(\mathcal{F}\), are a Gaussian random variable with conditional variance $\delta_1^2 + \delta_2^2$. Combining with \thref{part1}, we can finally deduce that \begin{equation*}
{m}_{n}^{1 / 2} \Delta_{n}^{1 / 4}\left(\overline{H}(1)^{n}+\overline{H}(2)^{n}\right) \stackrel{st}{\longrightarrow} Z_{\tau}+\beta Z_{\tau}^{\prime},
\end{equation*} 
where \(Z_{\tau}, Z_{\tau}^{\prime} \) are defined on \(\left(\tilde{\Omega}, \tilde{\mathcal{F}},\left(\tilde{\mathcal{F}}_{t}\right)_{t>0}, \tilde{\mathbb{P}}\right)\)and conditionally independent with  $$
\begin{array}{l}{\widetilde{\mathbb{E}}\left(Z_{\tau}^{2} | \mathcal{F}\right)=\delta_{1}^{2} =4\left(\Phi_{22} \sigma_{\tau}^{4} / \theta+2 \Phi_{12} \sigma_{\tau}^{2} \gamma_{\tau} \theta+\Phi_{11} \gamma_{\tau}^{2} \theta^{3}\right) \int K^{2}(u) d u},\\ 
{\widetilde{\mathbb{E}}\left(Z_{\tau}^{\prime 2} | \mathcal{F}\right)=\delta_{2}^{2} = \Lambda_{\tau}^{2} \iint_{x y \geq 0} K(x) K(y)(|x| \wedge|y|) d x d y}.\end{array}
$$

\end{proof}

\medskip


\bibliographystyle{jtbnew}


\begin{thebibliography}{33}
\expandafter\ifx\csname natexlab\endcsname\relax\def\natexlab#1{#1}\fi
\expandafter\ifx\csname url\endcsname\relax
  \def\url#1{\texttt{#1}}\fi
\expandafter\ifx\csname urlprefix\endcsname\relax\def\urlprefix{URL }\fi
\providecommand{\selectlanguage}[1]{\relax}

\bibitem[{A{\"i}t-Sahalia \& Jacod(2014)}]{jacodaitsahalia}
\textsc{A{\"i}t-Sahalia, Y. \& Jacod, J.} (2014).
\newblock \emph{High-frequency financial econometrics}.
\newblock Princeton University Press.

\bibitem[{A{\"\i}t-Sahalia \& Xiu(2019)}]{ait2019principal}
\textsc{A{\"\i}t-Sahalia, Y. \& Xiu, D.} (2019).
\newblock Principal component analysis of high-frequency data.
\newblock \emph{Journal of the American Statistical Association}
  \textbf{114}(525), 287--303.

\bibitem[{Alvarez \emph{et~al.}(2012)Alvarez, Panloup, Pontier \&
  Savy}]{alvarez2012estimation}
\textsc{Alvarez, A., Panloup, F., Pontier, M. \& Savy, N.} (2012).
\newblock Estimation of the instantaneous volatility.
\newblock \emph{Statistical inference for stochastic processes} \textbf{15}(1),
  27--59.

\bibitem[{Bandi \& Russell(2008)}]{Bandi}
\textsc{Bandi, F. \& Russell, J.} (2008).
\newblock {Microstructure noise, realized volatility and optimal sampling}.
\newblock \emph{Review of Economic Studies} \textbf{75}, 339--369.

\bibitem[{Barndorff-Nielsen \emph{et~al.}(2008)Barndorff-Nielsen, Hansen, Lunde
  \& Shephard}]{barndorff2008designing}
\textsc{Barndorff-Nielsen, O.~E., Hansen, P.~R., Lunde, A. \& Shephard, N.}
  (2008).
\newblock Designing realized kernels to measure the ex post variation of equity
  prices in the presence of noise.
\newblock \emph{Econometrica} \textbf{76}(6), 1481--1536.

\bibitem[{Campbell \emph{et~al.}(1997)Campbell, Lo \& MacKinlay}]{Campbell}
\textsc{Campbell, J., Lo, A. \& MacKinlay, A.} (1997).
\newblock \emph{{The econometrics of Financial Markets}}.
\newblock Princeton.

\bibitem[{Chen(2019)}]{chen2018inference}
\textsc{Chen, R.~Y.} (2019).
\newblock Inference for volatility functionals of multivariate {\Blue It{\^o}}
  semimartingales observed with jump and noise.
\newblock Tech. rep., Working paper. Available at arXiv:1810.04725v2.

\bibitem[{Fan \& Wang(2008)}]{fan2008spot}
\textsc{Fan, J. \& Wang, Y.} (2008).
\newblock Spot volatility estimation for high-frequency data.
\newblock \emph{Statistics and its Interface} \textbf{1}(2), 279--288.

\bibitem[{Figueroa-L\'opez \& Mancini(2019)}]{FigueroaLopezMancini:2019}
\textsc{Figueroa-L\'{o}pez, J.E. \& Mancini, C.} (2019)
\newblock Optimum Thresholding Using Mean and Conditional Mean Square Error.
\newblock \emph{J. Econom.}, 208(1):179$-$210.


\bibitem[{Figueroa-L{\'o}pez \& Li(2020{\natexlab{a}})}]{FigLi}
\textsc{Figueroa-L{\'o}pez, J.E. \& Li, C.} (2020{\natexlab{a}}).
\newblock Optimal kernel estimation of spot volatility of stochastic
  differential equations.
\newblock \emph{Stochastic Processes and their Applications.} \textbf{130}(8),
  4693--4720.

\bibitem[{Figueroa-L{\'o}pez \& Li(2020{\natexlab{b}})}]{FigLiSumplement1}
\textsc{Figueroa-L{\'o}pez, J.E. \& Li, C.} (2020{\natexlab{b}}).
\newblock Supplement to “optimal kernel estimation of spot volatility of
  stochastic differential equations”.
\newblock \emph{Available online on https://sites.wustl.edu/figueroa/}.

\bibitem[{Figueroa-L\'opez et al.(2021)}]{gong2021}
\textsc{Figueroa-L\'{o}pez, J.E., Gong, R., \& Han Y.} (2019)
\newblock Estimation of a Tempered Stable L\'evy Model of Infinite Variation.
\newblock \emph{Preprint, ArXiv:2101.00565}, 2021.


\bibitem[{Foster \& Nelson(1996)}]{foster1994continuous_optKernel}
\textsc{Foster, D. \& Nelson, D.} (1996).
\newblock Continuous record asymptotics for rolling sample variance estimators.
\newblock \emph{Econometrica} \textbf{64}(1), 139--174.

\bibitem[{Hansen \& Lunde(2006)}]{HanLun}
\textsc{Hansen, P. \& Lunde, A.} (2006).
\newblock {Realized variance and market microstructure noise}.
\newblock \emph{{ J. Bus. Econom. Statist.}} \textbf{24}, 127--218.

\bibitem[{Huang \emph{et~al.}(2005)}]{Huang2005}
\textsc{Huang, Xin \& Tauchen, George.} (2005).
\newblock {The relative contribution of jumps to total price variance}.
\newblock \emph{{ Journal of financial econometrics}} \textbf{3}(4), 456--499.


\bibitem[{Jacod \emph{et~al.}(2009)Jacod, Li, Mykland, Podolskij \&
  Vetter}]{jacod2009microstructure}
\textsc{Jacod, J., Li, Y., Mykland, P.~A., Podolskij, M. \& Vetter, M.} (2009).
\newblock Microstructure noise in the continuous case: the pre-averaging
  approach.
\newblock \emph{Stochastic processes and their applications} \textbf{119}(7),
  2249--2276.

\bibitem[{Jacod \& Mykland(2015)}]{jacod2015microstructure}
\textsc{Jacod, J. \& Mykland, P.~A.} (2015).
\newblock Microstructure noise in the continuous case: Approximate efficiency
  of the adaptive pre-averaging method.
\newblock \emph{Stochastic Processes and their Applications} \textbf{125}(8),
  2910--2936.

\bibitem[{Jacod \emph{et~al.}(2010)Jacod, Podolskij \& Vetter}]{jacod2010limit}
\textsc{Jacod, J., Podolskij, M. \& Vetter, M.} (2010).
\newblock Limit theorems for moving averages of discretized processes plus
  noise.
\newblock \emph{The Annals of Statistics} \textbf{38}(3), 1478--1545.

\bibitem[{Jacod \& Protter(2011)}]{JacodProtter}
\textsc{Jacod, J. \& Protter, P.} (2011).
\newblock \emph{Discretization of processes}.
\newblock Springer Science \& Business Media.

\bibitem[{Jacod \& Rosenbaum(2013)}]{jacod2013quarticity}
\textsc{Jacod, J. \& Rosenbaum, M.} (2013).
\newblock Quarticity and other functionals of volatility: efficient estimation.
\newblock \emph{The Annals of Statistics} \textbf{41}(3), 1462--1484.

\bibitem[{Jacod \& Todorov(2014)}]{JacodTodorov:2014}
\textsc{Jacod, J. \& Todorov, V.} (2014)
\newblock Efficient Estimation of Integrated Volatility in Presence of Infinite Variation Jumps.
\newblock \emph{The Annals of Statistics}, \textbf{42}(3):1029$-$1069.


\bibitem[{Kristensen(2010)}]{kristensen2010nonparametric}
\textsc{Kristensen, D.} (2010).
\newblock Nonparametric filtering of the realized spot volatility: A
  kernel-based approach.
\newblock \emph{Econometric Theory} \textbf{26}(1), 60--93.

\bibitem[{Li {et~al.}(2019)Li, Liu, Xiu \emph{et~al.}}]{li2019efficient}
\textsc{Li, J., Liu, Y., \& Xiu, D.} (2019).
\newblock Efficient estimation of integrated volatility functionals via
  multiscale jackknife.
\newblock \emph{The Annals of Statistics} \textbf{47}(1), 156--176.

\bibitem[{Li {et~al.}(2017)Li, Todorov \& Tauchen}]{li2017adaptive}
\textsc{Li, J., Todorov, V. \& Tauchen, G.} (2017).
\newblock Adaptive estimation of continuous-time regression models using
  high-frequency data.
\newblock \emph{Journal of Econometrics} \textbf{200}(1), 36--47.

\bibitem[{Li \& Xiu(2016)}]{li2016generalized}
\textsc{Li, J. \& Xiu, D.} (2016).
\newblock Generalized method of integrated moments for high-frequency data.
\newblock \emph{Econometrica} \textbf{84}(4), 1613--1633.

\bibitem[{Mancini \emph{et~al.}(2015)Mancini, Mattiussi \&
  Ren\`o}]{mancini2015estimation}
\textsc{Mancini, C., Mattiussi, V. \& Ren\`o, R.} (2015).
\newblock Spot volatility estimation using delta sequences.
\newblock \emph{Finance \& Stochastics} \textbf{19}(2), 261--293.

\bibitem[{Mancini(2009)}]{Mancini2009}
\textsc{Mancini, C.} (2009).
\newblock Non‐parametric threshold estimation for models with stochastic diffusion coefficient and jumps.
\newblock \emph{Scandinavian Journal of Statistics} \textbf{36}(2), 270-296.


\bibitem[{Mykland \& Zhang(2012)}]{MyklandZhang2012}
\textsc{Mykland, P. \& Zhang, L.} (2012).
\newblock The econometrics of high-frequency data.
\newblock \emph{In Statistical Methods for Stochastic Differential Equations,
  M. Kessler, A. Lindner, and M. Sorensen, eds.} , 109--190.

\bibitem[{Mykland \& Zhang(2009)}]{mykland2009inference}
\textsc{Mykland, P.~A. \& Zhang, L.} (2009).
\newblock Inference for continuous semimartingales observed at high frequency.
\newblock \emph{Econometrica} \textbf{77}(5), 1403--1445.

\bibitem[{Parzen(1962)}]{parzen}
\textsc{Parzen, E.} (1962).
\newblock On estimation of a probability density function and mode.
\newblock \emph{The annals of mathematical statistics} \textbf{33}(3),
  1065--1076.

\bibitem[{Podolskij \& Vetter(2009)}]{podolskij2009estimation}
\textsc{Podolskij, M. \& Vetter, M.} (2009).
\newblock Estimation of volatility functionals in the simultaneous presence of
  microstructure noise and jumps.
\newblock \emph{Bernoulli} \textbf{15}(3), 634--658.

\bibitem[{Rosenblatt(1956)}]{rosenblatt1956remarks}
\textsc{Rosenblatt, M.} (1956).
\newblock Remarks on some nonparametric estimates of a density function.
\newblock \emph{The Annals of Mathematical Statistics} \textbf{27}(3),
  832--837.

\bibitem[{Wand \& Jones(1995)}]{wand1995monographs}
\textsc{Wand, M. \& Jones, M.} (1995).
\newblock \emph{Monographs on statistics and applied probability}.
\newblock Chapman and Hall London, UK.

\bibitem[{Xiu(2010)}]{xiu2010quasi}
\textsc{Xiu, D.} (2010).
\newblock Quasi-maximum likelihood estimation of volatility with high frequency
  data.
\newblock \emph{Journal of Econometrics} \textbf{159}(1), 235--250.

\bibitem[{Yu {et~al.}(2014)Yu, Fang, Li, Zhang, \& Zhao}]{Yuetal2014}
\textsc{Yu, C., Fang, Y., Li, Z., Zhang, B., \& Zhao, X} (2014).
\newblock Non-parametric estimation of high-frequency spot volatility for Brownian semimartingale with jumps.
\newblock \emph{Journal of Time Series Analysis} \textbf{35}, 572-591.

\bibitem[{Yu {et~al.}(2014)Yu, Fang, Li, Zhang, \& Zhao}]{Yu2}
\textsc{Yu, C., Fang, Y., Li, Z., Zhang, B., \& Zhao, X} (2014).
\newblock Kernel Filtering of Spot Volatility in Presence of L\'evy Jumps and Market Microstructure Noise.
\newblock \emph{Preprint available at https://mpra.ub.uni- muenchen.de/63293/1/MPRA\_paper\_63293.pdf}.

\bibitem[{Zeng(2003)}]{Zeng:2003}
\textsc{Zeng, Y.} (2003).
\newblock {A partially observed model for micromovement of asset prices with
  Bayes estimation via filtering}.
\newblock \emph{Mathematical Finance} \textbf{13}(3), 411--444.

\bibitem[{Zhang(2006)}]{Zhang2006}
\textsc{Zhang, L.} (2006).
\newblock Efficient estimation of stochastic volatility using noisy
  observations: A multi-scale approach.
\newblock \emph{Bernoulli} \textbf{12}(6), 1019--1043.

\bibitem[{Zhang \emph{et~al.}(2005)Zhang, Mykland \&
  A{\"i}t-Sahalia}]{zhang2005tale}
\textsc{Zhang, L., Mykland, P.~A. \& A{\"i}t-Sahalia, Y.} (2005).
\newblock A tale of two time scales: Determining integrated volatility with
  noisy high-frequency data.
\newblock \emph{Journal of the American Statistical Association}
  \textbf{100}(472), 1394--1411.

\bibitem[{Zu \& Boswijk(2014)}]{zu2014estimating}
\textsc{Zu, Y. \& Boswijk, H.~P.} (2014).
\newblock Estimating spot volatility with high-frequency financial data.
\newblock \emph{Journal of Econometrics} \textbf{181}(2), 117--135.

\end{thebibliography}

\end{document}